\documentclass[oribibl]{llncs}
\usepackage{multicol}
\usepackage{amsmath}
\usepackage{amssymb}
\usepackage{graphicx}
\usepackage{color}
\usepackage{hhline}
\usepackage{tabu}
\usepackage{smallsec}
\usepackage{enumerate}
\usepackage{tikz} 
\usetikzlibrary{arrows.meta}
\usetikzlibrary{positioning} 

\usepackage{capt-of} 

\usepackage[edges]{forest} 

\newcounter{example-counter}

\newcounter{remark-counter}

\renewenvironment{example}%
{\vskip \abovedisplayskip \refstepcounter{example-counter}%
\noindent {\bf Example \arabic{example-counter}.}}%

\renewenvironment{remark}%
{\vskip \abovedisplayskip \refstepcounter{remark-counter}%
\noindent {\bf Remark \arabic{remark-counter}.}}%

\newcommand{\boxtheorem}{\hfill $\blacksquare$\\}

\newcommand{\mf}[1]{\mathfrak{ #1}}

\newcommand{\msf}[1]{\mbox{\sf #1}}

\newcommand{\red}[1]{\textcolor{red}{#1}}

\newcommand{\green}[1]{\textcolor[rgb]{0.00,0.59,0.00}{#1}}

\newcommand{\nit}[1]{{\it #1}}
\newcommand{\mc}[1]{\mathcal{ #1}}

\newcommand{\comlb}[1]{{\vspace{2mm}\noindent \bf \red{COMM(LEO):}}~ #1   \hfill {\bf
		\red{END.}}\\}
\newcommand{\comfa}[1]{{\vspace{2mm}\noindent \bf \green{COMM(FEL):}}~ #1   \hfill {\bf
	\green{END.}}\\}

\newcommand{\ignore}[1]{}

\newcommand{\In}{\msf{\footnotesize \ in}}
\newcommand{\Out}{\msf{\footnotesize \ out}}

\newcommand{\atoms}{\nit{Atoms}}
\newcommand{\CE}{\nit{C\! E}}
\newcommand{\Shapley}{\nit{Shapley}}
\newcommand{\Resp}{\nit{Re}}
\newcommand{\Den}{D^\nit{en}}
\newcommand{\Dex}{D^\nit{ex}}

\title{\bf Alignment of Explanation Scores in Data Management\vspace{-1mm}}

\author{{\bf Felipe Az\'ua\inst{1}} \ \and \ {\bf Leopoldo Bertossi\inst{2}\thanks{Professor Emeritus. Corresponding Author.}}}

\institute{University of Edinburgh (UK). \ f.r.i.azua@sms.ed.ac.uk
 \and Carleton University (Canada) and IMFD (Chile). bertossi@scs.carleton.ca}

\begin{document}
	\maketitle
    \pagestyle{plain}
    \thispagestyle{empty}

	\begin{abstract}
Different attribution scores have been proposed to quantify the relevance of database tuples for query answering in databases; e.g. Causal Responsibility, the Shapley Value, the Banzhaf Power-Index, and the Causal Effect. They have been analyzed in isolation. This work is a first investigation of score alignment depending on the query and the database; i.e.  on whether they induce compatible rankings of tuples. We concentrate mostly on causality-based scores; and provide a syntactic dichotomy result for queries. More precisely, we characterize a class of Boolean Conjunctive Queries without self-joins, such that: \ On one side, when the query belongs to the class, pairs of scores are always aligned. On the other, if the query falls outside the class, there is a database instance where they are not aligned. \  It turns out that the presence of exogenous tuples makes a crucial difference in this regard.
	
\vspace{2mm}\noindent {\bf Keywords:} Explainable Data Management,
Database Queries,
Causality and Causal Effect,
Probabilistic Databases,
Causal Responsibility,
Shapley Value
	\end{abstract}

\section{Introduction}\label{sec:intro}

\vspace{1mm}
Explanations for query answers in databases (DBs) can be provided in terms of {\em attribution scores}. By assigning numbers to DB tuples, they quantify their relevance or importance for a query answer at hand.\
Among those scores, we find {\em Causal Responsibility} \cite{QA_Causality,tocs}, {\em Resilience}  \cite{wolfgang} (closely related to the former), the {\em Shapley Value} \cite{Shapley_Tuple_Bertossi,sigRec21,sigRec23}, the {\em Banzhaf Power-Index} \cite{Shapley_Tuple_Bertossi} (related to the Shapley value), and the  {\em Causal Effect} \cite{Causal_Effect}.\
These are all {\em local scores}, i.e. for a single query answer obtained from a relational DB \cite{bda22}.\
Although some of those scores have been extended to
probabilistic databases (PDBs) \cite{senellart24,clear25},  in this work we consider only explanation scores for non-probabilistic relational DBs. We do use probability as an auxiliary tool though.

The above mentioned scores have been investigated in isolation from each other, in terms of their general algorithmic, complexity theoretic, and mathematical properties. However, comparisons among them in terms of the results they return have received little attention. There is no reason to expect that those  scores return the same {\em numbers}. However, we might be interested in knowing whether they {\em rank} the DB tuples equally. This is important in that we may want to identify the most relevant tuples for a query answer, but not necessarily to obtain their specific  scores. \ If this is the case, it is useful to know in advance that for certain classes of queries, this form of tuple relevance is invariant under different scores.

In this work, we address the problem of identifying conditions on {\em Boolean Conjunctive Queries} (BCQs) and DBs  under which attribution scores are aligned or not (more details below in this section).\
We concentrate mainly on three of those scores, as previously defined and applied to DBs: Responsibility  and the Causal Effect, denoted with {\sf Resp}  and {\sf CES}, respectively; and {\sf CES} and the Shapley Value, denoted {\sf Shapley}. We unveil conditions on queries and DBs that lead to the alignment of these scores  or the lack of it. \ We concentrate mostly on {\em self-join-free} BCQs, where a predicate does not appear repeated \cite{suciu}.

{\sf Resp} is based on {\em responsibility} \cite{chockler} as introduced in the context of {\em actual causality} and  {\em counterfactual interventions} \cite{HP05,halpern15,H16}. The latter can be seen as hypothetical changes performed on a causal model, to identify other changes. By doing so, one can detect cause-effect relationships.
\ {\sf CES} has its origin in  the {\em causal effect} measure of {\em statistical causality and observational studies} \cite{rubin,holland}, where one usually cannot predefine and build control groups, but they have to be recovered from the available data  \cite{gelman,Pearl_CI,roy&salimi23}.

{\sf CES} was first  and successfully applied   as an explanation score in data management in \cite{Causal_Effect}, as an alternative to {\sf Resp}, when the latter did not provide intuitive results. However, there are cases of queries and DBs for which {\sf Resp} returns more intuitive results than {\sf CES} (see \cite[chap. 7]{QA_Caus_Salimi} and \cite{DeclAI26}). Accordingly, an analysis of score alignment is a relevant and interesting problem. \ In this work, {\sf CES} takes a prominent role. It has been less investigated by the data management community than the other scores.\footnote{{\sf CES} was extended in \cite{clear25} to arbitrary PDBs, as the {\em Generalized Causal-Effect Score},  and its general axiomatic properties and complexity were investigated.}

The Shapley Value  is well-known in  \emph{Coalition Game Theory} as a measure of the contribution of individual players to a shared {\em game (or wealth) function} \cite{shapley1953original,roth1988shapley}.  In  \cite{Shapley_Tuple_Bertossi}, it was applied with the query (Boolean or aggregate) as the game function, to quantify the contribution of tuples to the answer. \ The same was done with the {Banzhaf Power-Index} \cite{banzhaf_value}.\
Actually, in \cite{Shapley_Tuple_Bertossi}, it was established that   {\sf CES} coincides with the Banzhaf Power-Index  as applied to DBs. We denote  this particular case of the latter with {\sf BPI}. \   The ``ordinal equivalence", i.e. induction of the same order, of general game-theoretic scores, e.g. of the Shapley Value and the Banzhaf Power-Index, has been investigated  in game-theory  \cite{freixas,freixas12}. \ However, those results  do not provide useful information in our setting.\footnote{Those results basically reduce the problem of score alignment to that of alignment of the integrals defining them.}

In actual causality \cite{HP05}, counterfactual interventions affect {\em endogenous variables}, but not  {\em exogenous variables}. The latter are taken as a part of a given context. They are not subject to causal analysis, but they may have an effect on other variables, a Boolean query, in our case. The endogenous variables are those under closer causal investigation. Through interventions on them, one may discover some other cause-effect relationships. \ In our case, this amounts to having a partition of the DB into {\em endogenous and exogenous tuples}. \ As in the general setting, interventions, affecting only endogenous tuples,  can be conceived as tuple insertions and deletions. When a query is monotone, only deletions matter.

This kind of partition of the DB has also been considered with {\sf Resp} \cite{QA_Causality}, \ with both {\sf Shapley} and {\sf BPI} \cite{Shapley_Tuple_Bertossi}, and {\sf CES} \cite{Causal_Effect,clear25}. It turns out that exogenous tuples have an effect on the values taken by the scores, and their computational properties. \ Actually, as we show in this work, the possible presence of exogenous tuples may also have an impact on score alignment.

In this work, our main contributions are:
\vspace{-2mm}\begin{enumerate}[(a)]
\item A full syntactic characterization of {\em self-join-free} BCQs \ (SJF-BCQs) \ for which {\sf Resp} and {\sf CES} are aligned for every DB that may contain exogenous tuples. \ This is our main result.

\item An analysis of the alignment of  {\sf Resp} and {\sf CES} for DBs that do not have exogenous tuples. In particular, we present syntactic classes of SJF-BCQs where these scores are always aligned.

\item An analysis of the alignment of {\sf CES} and {\sf Shapley}. In particular, we identify  broad classes of SJF-BCQs for which  {\sf CES} and {\sf Shapley} are {\em not} aligned on DBs that may have exogenous tuples.

\item An analysis of score alignment for BCQs {\em with self-joins}. In particular, we identify a syntactic class of queries with self joins for which {\sf CES} and {\sf Resp} are always aligned on DBs that may contain exogenous tuples. We also identify a syntactic class of queries for which there are  instances with exogenous tuples where the scores are not aligned.

\ignore{\item  The syntactic characterization of a broad class of BCQs, with or without self-joins, for which there is an database instance where {\sf CES} and {\sf Resp} are not aligned.\ignore{, also also {\sf CES} and {\sf Shapley} are not aligned.}}

\ignore{\item  A full syntactic characterization of SJF-BCQs for which {\sf Shapley} and {\sf CES} and are aligned for every DB that may contain exogenous tuples.}
\end{enumerate}

Our results establish, in particular, a {\em dichotomy} of the following form: \ There is syntactically defined class $\mc{C}^{\!\nit{Re}}$ \ignore{and $\mc{C}^{\nit{Sh}}$} of SJF-BCQs, such that a query belongs to $\mc{C}^{\!\nit{Re}}$ \ignore{, resp. $\mc{C}^{\nit{Sh}}$,} iff  {\sf Resp} and {\sf CES}\ignore{, resp. {\sf Shapley} and {\sf CES},} are aligned for every possible DB, with or without exogenous tuples. \ In the negative cases, i.e. when the query does not belong to the \ignore{specified} class, one can algorithmically, actually, efficiently, build an instance where the \ignore{corresponding} scores are not aligned. \ignore{For both classes,} The counterexample instance we exhibit has exogenous tuples.\
Membership of this syntactic class (and of the other classes of queries we identify in this work) can be decided in linear time in the size of the query.

As just mentioned, we exhibit a syntactic class of SJF-BCQs for which {\sf Resp} and {\sf CES} are always aligned for every DB {\em without exogenous tuples}. This class is a superset of $\mc{C}^{\!\nit{Re}}$. However,  this class is not dychotomic: If a query does not belong to it, it is unclear if the scores are aligned or not; nor it is known if a counterexample instance without exogenous tuples exists. This also shows the importance of exogenous tuples, in that they are needed to build counterexample instances.

The dichotomy result given in (a) is obtained by means of a reduction technique that we introduce. It works properly and as expected for SJF-BCQs, and requires exogenous tuples.\
However, for instances where only endogenous tuples are allowed, as in (b), we believe a different technique is required; which could lead to a (possibly different) dichotomy. \
Whether there is a dychotomic syntactic class for the alignment of DBs without exogenous tuples is an open problem. \

Moreover, the same technique used for (a) was successfully used for a particular class of queries {\em with self-joins}. However, it is also unclear whether this technique can be used for any arbitrary class of queries with self-joins.

To the best of our knowledge, this work is the first theoretical investigation of the alignment of attribution scores in data management. Beyond the interest in the particular scores we consider, this work also offers a sort of general template and methodology for analyzing the alignment of other possible scores in the context of DBs (that possibly contain exogenous tuples).

This paper is structured as follows. \  In Section \ref{sec:back}, we provide necessary background, including {\sf Resp}, {\sf Shapley} and {\sf BPI};   and related work. \ In Section \ref{sec:CE}, we present  {\sf CES}. \ In Section \ref{sec:alig}, we define alignment, and give some examples of (non-)alignment. \ In Section \ref{sec:ces-resp}, the main section, we investigate the alignment of {\sc CES} and {\sf Resp}, and present the dychotomic class of queries for DBs that may have exogenous tuples. \ In Section \ref{sec:ces-resp_noex}, we consider DBs that do not have exogenous tuples.
\ In Section \ref{sec:ces-shap}, we bring {\sf Shapley} into the alignment analysis. \ In Section \ref{sec:sj-bcq}, we retake the alignment of  {\sc CES} and {\sf Resp} for BCQs that this time may have self-joins. \
 \ In Section \ref{sec:concl}, we draw some conclusions and point to future work.

 This paper is an extended, completely revised, and corrected version of \cite{foiks26}. In particular, the syntactic characterization of the dychotomic class of queries has been corrected here. Particularly relevant, in comparison with \cite{foiks26}, is a general methodological change, based now on the extensive use, in this version, of the lineage of a query as a powerful tool.


\section{Background and Preliminaries} \label{sec:back}

\vspace{2mm}

\subsection{Relational Databases}

A relational schema $\mc{S}$ contains a domain of constants, $\mc{C}$,  and a set of  predicate symbols of finite arities, $\mc{P}$. \ignore{ $\mc{R}$ gives rise to a language $\mf{L}(\mc{R})$ of first-order (FO)  predicate logic.  Variables are  denoted with $x, y, z, ...$; and finite sequences thereof with $\bar{x}, ...$; and constants with $a, b, c, ...$, etc.} An {\em atom} is of the form $P(t_1, \ldots, t_n)$, with $n$-ary $P \in \mc{P}$,   and $t_1, \ldots, t_n$, constants  or variables. \
A {\em tuple} is  a {\em ground} atom, i.e. without  variables. A database (instance), $D$, for $\mc{S}$ is a finite set of ground atoms. The {\em active domain} of $D$, denoted ${\it ad\!om}(D)$, is the set of constants that appear in atoms of $D$. \ For a predicate  $P \in \mc{S}$ and  database $D$ for $\mc{S}$, $P^D$ denotes the extension of $P$ in $D$. \ With $\nit{ad\!om}(P^D)$ we denote the subset of $\nit{ad\!om}(D)$ that contains the constants in $P^D$.

A {\em conjunctive query} \ (CQ), $\mc{Q}(\bar{x})$,  is a formula of the form: \ignore{an existentially quantified conjunction of atoms:}  \ $\exists  \bar{y}\;(P_1(\bar{x}_1)\wedge \dots \wedge P_m(\bar{x}_m))$, whose free variables are $\bar{x} := (\bigcup \bar{x}_i) \smallsetminus \bar{y}$. We consider queries that are associated with the same schema as that for the database instance, $D$, at hand. We say that $\mc{Q}$ is {\em compatible} with $D$.  If $\mc{Q}$ has $n$ free variables,  $\bar{c} \in \mc{C}^n$ \ is an {\em answer} to $\mc{Q}$ from $D$ if $D \models \mc{Q}[\bar{c}]$, i.e.  $Q[\bar{c}]$ is true in $D$  when the variables in $\bar{x}$ are component-wise replaced by the values in $\bar{c}$. $\mc{Q}[D]$ denotes the set of answers to $\mc{Q}$ from $D$. \ $\mc{Q}$ is a {\em Boolean CQ}  when $\bar{x}$ is empty. A BCQ is usually written as $\bar{\exists}(P(\bar{x}_1) \wedge \cdots \wedge P_m(\bar{x}_m))$, where $\bar{\exists}$ corresponds to the existential closure of the conjunction, that is, all the variables are existentially quantified.

 For a Boolean query (BQ),  when it is {\em true} in $D$,  $\mc{Q}[D] := \{1\}$. Otherwise, if it is {\em false}, $\mc{Q}[D] := \{0\}$. \ However, we write $\mc{Q}[D] =1$ and $\mc{Q}[D] =0$, resp.
\ This work concentrates on {\em Monotone BQs} (MBQs),  in particular, BCQs. For MBQs it holds, by definition, $\mc{Q}[D] \leq \mc{Q}[D^\prime]$ when $D \subseteq D^\prime$.
\ignore{We will assume everywhere that a given query $\mc{Q}$ and a given database $D$ instance are {\em compatible} in that $D$ is compatible with the relational schema that also gives rise to $\mc{Q}$'s query language.} \ A CQ is {\em self-join-free} (SJF) if no predicate symbol appears more than once. We abbreviate a self-join-free BCQ with SJF-BCQ.

For a BCQ $\mc{Q}$, $\nit{Atoms}(\mc{Q})$ denotes the set of all atoms in $\mc{Q}$, i.e. $\nit{Atoms}(\mc{Q}) = \{P(\bar{x}_1), \ldots,  P_m(\bar{x}_m)\}$.
\ $\nit{V\!ar}(\mc{Q})$ denotes the set of its variables. For $V \subseteq \nit{V\!ar}(\mc{Q})$, $\nit{Atoms}(\mc{Q},V)$ denotes the set of atoms in $\mc{Q}$ where the variables in $V$ appear. If clear from context, we will drop the query, and write $\nit{Atoms}(V)$. Additionally, if $V=\{v\}$, then we simply write $\atoms(\mc{Q},v)$ and $\atoms(v)$, respectively.

\subsection{Actual Causality and Responsibility in DBs} \label{sub:resp}

For causality purposes, some of the tuples in a DB $D$ are  {\em endogenous}, and can be subject to   {\em counterfactual interventions}, in this case, deletions, insertions, or value updates. The other tuples are {\em exogenous}, and are taken as a given, leaving them outside a specific kind of analysis, but they participate in query answering. \ The partition $D = D^\nit{en} \cup D^\nit{ex}$ is application dependent.

A tuple $\tau \in D^\nit{en}$ is a {\em counterfactual cause}   for a BQ $\mc{Q}$ if: $D \models \mc{Q}$, and $(D \smallsetminus \{\tau\}) \not \models \mc{Q}$. \ A tuple $\tau \in D^\nit{en}$ is an {\em actual cause}   for $\mc{Q}$ if there is $\Gamma \subseteq D^\nit{en}$, such that:  $D \models \mc{Q}$, \ $(D \smallsetminus \Gamma) \models \mc{Q}$, and $(D \smallsetminus (\Gamma \cup \{\tau\})) \not \models \mc{Q}$ \ \cite{QA_Causality}.
\ The set $\Gamma$ is called a {\em contingency set} for $\tau$. $\nit{Cont}(D,\mc{Q},\tau)$, or simply, $\nit{Cont}(D,\tau)$, denotes the set of contingency sets for $\tau$.

The {\em responsibility} of $\tau$ as an actual cause for $\mc{Q}$ is defined by \ $\nit{Re}(D,\mc{Q},\tau) := \frac{1}{1 + |\Gamma|}$,
	where \(|\Gamma|\) is the cardinality of a minimum-size contingency set for \(\tau\) \ \cite{chockler,QA_Causality}. We sometimes write $\nit{Re}(D,\tau)$ or $\nit{Re}(\tau)$. \ A counterfactual cause is an actual cause with maximum responsibility,  $1$. \ If $\tau$ is not an actual cause, its responsibility is defined as $0$.
 Responsibility as used in DBs will be  denoted with {\sf Resp}. It quantifies the causal contribution of a tuple to a query answer \cite{QA_Causality,tocs,flairs17}.

The endogenous/exogenous partition of database tuples is application dependent. For example, in the case of data obtained from the integration of different sources, possibly virtually integrated, those tuples obtained from a particular source might be left as exogenous, and for different reasons, e.g. the tuples themselves cannot be examined/accessed; or they are already trusted, etc. In our work, we take the partition as a given.
 	
\subsection{The Shapley Value and the Banzhaf Power-Index}  \label{sub:shapley-banzhaf}

In the context of DBs, a BQ $\mc{Q}$ becomes the {\em game function}, from subsinstances of a DB $D$ to $\{0,1\}$  \ \cite{Shapley_Tuple_Bertossi,sigRec21,sigRec23}:

    \begin{eqnarray}
    \Shapley(D,\mc{Q},\tau) \  &=& \!\!\!\! \sum_{S \subseteq (D^\nit{en} \smallsetminus \{\tau\})} \!\!\! \!\!\dfrac{|S|! \times (|D^\nit{en}| - |S| - 1)!}{|D^\nit{en}|!}  \times \Delta(\mc{Q},S,\tau), \label{eq:shap}\\
     \nit{BPI}(D,\mc{Q},\tau) \ &:=& \! \! \sum_{S \subseteq D^\nit{en} \smallsetminus \{\tau\}} \!\! \dfrac{1}{2^{|D^\nit{en}|-1}} \times \Delta(\mc{Q},S,\tau), \ \  \mbox{ with }\label{eq:banzhaf}\\
       \Delta(\mc{Q},S,\tau) \ &:=& \ \ \mc{Q}[S \cup D^\nit{ex} \cup \{\tau\}] \ - \  \mc{Q}[S \cup D^\nit{ex}], \mbox{ for } S \subseteq D^\nit{en}. \label{eq:delta}
    \end{eqnarray}
Since $\mc{Q}$ is Boolean, $\mc{Q}[S] \in \{0,1\}$.
\ The Shapley value  {\em is the only measure} of players' contributions to the coalition game's shared wealth function  that satisfies certain desirable properties \cite{shapley1953original,roth1988shapley}. Those properties provide a categorical axiomatization for the Shapley value. There is also a  categorical axiomatization for  BPI \cite{BPI_Math_Props}. \ In the following, we will denote these attribution scores in DBs with {\sf Shapley} and {\sf BPI}, resp.

\subsection{The Lineage of a Query}\label{sec:lin}

We introduce the notion of lineage of a query by means of an example.  \ Consider the instance  $D = \{R(c,b), R(a,d), R(b,a), R(b,b), R(e,f), S(a), S(b),$ $ S(c), S(d)\}$, and the BCQ \ $\mc{Q}\!: \exists x \exists y ( S(x) \land R(x, y) \land S(y))$, not SJF, by the way. \ It holds: \ $D \models \mc{Q}$. \ Here, $\nit{ad\!om}=\{a,b,c,d,e,f\}$.

So far, we are using query satisfaction as associated to classical predicate logic.
However, the same can be  stated
in  propositional terms, opening the opportunity of  using computational techniques
for handling Boolean circuits. This is done  by means of the {\em
lineage of the query} \cite{suciu}.

We start by introducing a propositional variable $X_{\!A}$, for each ground atom (tuple) $A$ obtained by instantiating the DB predicates on $\nit{ad\!om}(D)$. \ Now, if $\mc{Q}$ is a BCQ, as in our example, the lineage $\mc{L}(\mc{Q})$ on $\nit{ad\!om}$ is a propositional formula in Positive Disjunctive Normal Form (DNF):
\begin{eqnarray}\mc{L}(\mc{Q})&:=&(X_{S(a)} \wedge X_{R(a,a)} \wedge X_{S(a)}) \vee (X_{S(a)} \wedge X_{R(a,b)} \wedge X_{S(b)}) \vee \cdots \label{eq:first}\\&& (X_{S(a)} \wedge X_{R(a,d)} \wedge X_{S(d)}) \vee \cdots. \nonumber
\end{eqnarray}
Here, we are omitting conjuncts, such as $(S(a) \wedge R(a,b) \wedge S(c))$, that do not satisfy the join-condition on $\nit{ad\!om}$.

Now, the database instance $D$ at hand acts as a truth-assignment (or valuation) $\sigma^D$: \ $\sigma^D(X_{\!A}) = 1$ iff $A \in D$. With this, some of the $X_A$ may be true and others false, which will simplify the lineage. For example, $\sigma^D(X_{R(a,a)}) = 0$, and the first conjunct in (\ref{eq:first}) can be eliminated.
\ Accordingly, for our query $\mc{Q}$ and instance $D$, the {\em lineage of } $\mc{Q}$ {\em on} $D$ becomes:
\begin{eqnarray}\mc{L}(\mc{Q},D)&:=& (X_{S(b)} \wedge X_{R(b,a)} \wedge X_{S(a)}) \ \vee \ (X_{S(a)} \wedge X_{R(a,d)} \wedge X_{S(d)}) \ \vee \nonumber\\ && (X_{S(c)} \wedge X_{R(c,b)} \wedge X_{S(b)}) \ \vee \ (X_{S(b)} \wedge X_{R(b,b)}).\label{eq:lin}
\end{eqnarray}
It holds: \
$\mc{L}(\mc{Q},D)$ is {\em true} under assignment $\sigma^{\!D}$
iff $D \models \mc{Q}$.

Since exogenous tuples in a DB are considered to be always true, or always present in the DB, i.e. not affected by counterfactual deletions, when $D = D^\nit{en} \cup D^\nit{ex}$ does contain exogenous tuples, for every $A \in D^\nit{ex}$, the propositional variable  $X_A$ used for the lineage becomes $\nit{true}$. Those variables can be eliminated from the lineage. We call this simplified formula the {\em ex-reduced lineage}, and we still denote it by $\mc{L}(\mc{Q},D)$. It contains variables for endogenous tuples only, and it coincides with (\ref{eq:lin}) when $D^\nit{ex} = \emptyset$. \ For example, if $S(d)$ is exogenous in $D$, the second conjunct becomes $(X_{S(a)} \wedge X_{R(a,d)})$. If all the tuples of some conjunct are exogenous, $\mc{L}(\mc{Q},D)$ is equivalent to \nit{true}.

\subsection{Probabilistic Databases} \label{sub:prob-pdbs}

Given a relational DB, $D$, a PDB, $D^p$ associated to $D$ can be conceived as the {\em collection $\mc{W}$ of possible worlds}, i.e. of all subinstances $W$ of  $D$; where each $W \in \mc{W}$ has an  associated probability $p(W)$. It holds  $\sum_{W\in\mc{W}}p(W) = 1$ \cite{suciu}. \ Accordingly, $D^p$  can be identified with a discrete {\em probability space} $\langle \mc{W}, p\rangle$, with $p$ probability distribution on $\mc{W}$. \ The probability of tuple $\tau$ being in $D$  is  \
$P(\tau) \ := \! \! \sum\limits_{\tau \in W \in \mc{W}} p(W)$. 
Since, $D$ is finite, $\mc{W}$ and the $W$ are finite.

A numeric query $\mc{Q}$ on $D$ becomes a random variable on  $\mc{W}$; with Bernoulli distribution if it is a BQ,  in which case the probability of $\mc{Q}$ (being true) is \
$P(\mc{Q})  := P(\mc{Q} = 1) := \sum\limits_{W \in \mc{W}: \ W \models \mc{Q}} p(W)$.

The  particular case of a {\em tuple-independent PDB} (TID) is a regular relational DB whose  relations have an extra attribute  with probabilities for the tuples. Tuples are stochastically independent from each other. \ Example \ref{ex:tid} shows a TID. \ The first tuple, $\tau_1$,  is in the DB with probability $p_1$.\
For a possible world $W \in \mc{W}$, the corresponding (non-probabilistic) relation $R_W$ will contain only some of the tuples in $R$, and has the probability: \ $p(R_W) := \Pi_{_{\tau_i \in R_W}} p_i \ \times \ \Pi_{_{\tau_j \in (R \smallsetminus R_W)}} (1-p_j)$.\ In general, for a TID and a possible world $W$:\
$p(W) \ := \ \Pi_{_{R_W}} p(R_W)$.

\begin{center}
    {\scriptsize
    $\begin{tabu}{c|c|c|c||c|}\hline
        R & A & B & C& p \\ \hline
        \tau_1& a_1 & b_1& c_1&0.8\\
        \tau_2 & a_2& b_2& c_2&0.3\\
        \tau_3 & a_3& b_3& c_3&0.2\\
        \tau_4 & a_4 & b_4& c_4 &0.6\\
        \tau_5& a_5& b_5& c_5&1\\
        \hhline{~----}
    \end{tabu}$}
    \captionof{table}{ \ \ Probabilistic DB $D$.}
    \label{fig:graph}
\end{center}

\begin{example} \label{ex:tid}
    For the TID in Table \ref{fig:graph}, the possible world $R_W = \{\tau_1, \tau_3, \tau_5\}$ has the probability: \ $p(R_W) = 0.8 \times (1-0.3) \times 0.2 \times (1-0.6) \times 1$. \boxtheorem
\end{example}

For a  TID $D^p$ and $W \in \mc{W}$,   $W^p$ denotes the TID that contains only the tuples in $W$ with their probabilities as in $D^p$. For example, $W^p$ could contain only tuples $\tau_2, \tau_4$ with probabilities $0.6$ and $0.3$, resp.  \ We use the notations  $P_{\!D}\!(\mc{Q})$ and $P_{W}\!(\mc{Q})$ to distinguish $D^p$ and $W$ when doing query answering (QA). When clear from context, we simply use $P(\mc{Q})$.

For a TID $D^p$ associated to $D$, exogenous tuples have probability $1$. Then, for every $W \in \mc{W}$ with $D^{\nit{ex}} \not \subseteq W$,  $p(W) = 0$. \ When all the endogenous tuples' probabilities are the same, we have a {\em uniform TID}. \ If this probability is $\frac{1}{2}$, it is a  $U(\frac{1}{2})$-TID.

\subsection{Related Work}

The  Shapley Value and the BPI  have been first used in \cite{Shapley_Tuple_Bertossi,sigRec21}, to quantify the contribution of tuples to answers for BCQs and aggregate queries. \ The authors of \cite{benny22} concentrate on computational aspects of the Shapley value applied to query answering. \ See \cite{sigRec23} for a recent review.

In \cite{bienvenu}, the Shapley value in DBs is connected with the Generalized Model Counting Problem --- in this case, the number of subinstances of a given size that satisfy the query --- to obtain computational and complexity results for new classes of queries, including some unions of CQs. \ In \cite{kara}, the problem of computing the Shapley value of variables in Boolean circuits  is connected with query evaluation in PDBs, obtaining  dichotomy results (similar to those in \cite{dichotomy_UCQ}) for the complexity of Shapley value computation  for query answering in DBs. \ In \cite{senellart24}, the expected value of the Shapley value on TIDs is investigated.

In \cite{deutch2},  {\em rankings} induced by the Shapley values are investigated, rather than  the values themselves. In \cite{deutch3}, the interest is in the combination of {\em data provenance} and the Shapley values, for the computation of the latter.\
In \cite{deutchBanzhaf}, research  delves more deeply into experimental results around the  BPI for query answering.

The computational aspects and the complexity of Responsibility in data management were first investigated in \cite{QA_Causality,QA_CausalityII,tocs}. In \cite{flairs17}, Responsibility under integrity constraints was formalized and investigated. \
Close to Responsibility in DBs, we find the notion of {\em resilience}  \cite{wolfgang}.

The use of the {\em causal effect} in data management was proposed  in  \cite{Causal_Effect}. Its connection with the BPI was first established in \cite{Shapley_Tuple_Bertossi}, from which first results about its  complexity were obtained.

In the context of Coalition Game Theory (CGT), the problem of \emph{Ordinal Equivalence}  that is, that of identifying and characterizing  {\em power indices} that produce the same preorders, was proposed in \cite{Tomiyama87}; and has been further studied in \cite{SwapRobust02,carreras-freixas-ordinal-equivalence,freixas12}. Those results do not directly translate to our setting; they classify the specific games (queries and instances, in our setting) for which {\em Values} (as used in CGT, e.g. Shapley Value) or power indices are aligned, but no specific properties of the game function (queries) are identified.
\ Moreover, in that literature no distinction is made between endogenous and exogenous variables (or players).


\section{The Causal-Effect Score in Databases}
\label{sec:CE}
The definition of the {\em Causal Effect Score} (CES) in DBs that we give is different from -but equivalent to- that in \cite{Causal_Effect}.\footnote{The new one becomes a particular case of the {\em Generalized Causal-Effect Score}  \cite{clear25}, that is defined for arbitrary PDBs. In \cite{Causal_Effect} (see also  \cite{bda22}), for the case of DBs with associated TIDs, interventions are applied on the {\em lineage of the query}, which can be seen  as a causal model. It is common to apply interventions on variables of a causal model \cite{Pearl_CI}.\ Those interventions on the lineage act on the random propositional variables $X_\tau$ associated to tuples $\tau$, which are made {\em true} (inserted) or {\em false} (deleted) via $\nit{do}(X_\tau =1)$ or  $\nit{do}(X_\tau =0)$.}  In the following, $D$ is an instance, $\mc{W}$ is the class of all subinstances of $D$, $D^p$ is a  PDB assigning a distribution on $\mc{W}$. Each tuple $\tau \in D$ has a  probability of being true in $D$ (but $1$ if exogenous).

\ignore{\subsection{Interventions Database  Distributions} }

CES is defined through {\em  interventions} on the distribution of $D^p$. They are of the forms $\nit{do}(\tau \In)$ and $\nit{do}(\tau \Out)$, with the intuitive meaning that $\tau$ is {\em  made true} in $D$. Similarly,  $\nit{do}(\tau \Out)$ indicates that $\tau$ is {\em made false} in $D$.

 We use expressions of the form  $P(\mc{Q}=1~|~\nit{do}(\tau \In))$, where $P$ is the distribution induced by $p$ on the BQ $\mc{Q}$ treated as a random variable on $\mc{W}$. Intuitively, it means ``the probability of the query being true given that tuple $\tau$ is made true". This seeming conditionalizations on $\nit{do}(\tau \In)$ and $\nit{do}(\tau \Out)$ should be seen as perturbations of the original distribution $p$ on $\mc{W}$, as we now define.

\vspace{2mm}
\begin{definition}  \em \label{def:DO}  \cite{clear25}
Consider an instance $D = D^\nit{en} \cup D^\nit{ex}$,   an  associated PDB $D^p=\langle \mc{W}, p\rangle$, and $\tau \in D^\nit{en}$. The interventions on  $D^p$ create the new PDBs:

     \vspace{1mm} \noindent  (a) \ $D^p(\nit{do}(\tau \In)) \ := \ \langle \mc{W},p^{+\tau}\rangle$, such that, for $W \in \mc{W}$:
     \begin{equation*}p^{+\tau}(W)  \ := \sum\limits_{\stackrel{W^\prime \in \mc{W}}{W^\prime \cup \{\tau\} = W}}  p(W^\prime).
     \end{equation*}
    (b) \ $D^p(\nit{do}(\tau \Out)) \ := \ \langle \mc{W},p^{-\tau}\rangle$, such that,  for $W \in \mc{W}$:
      \begin{equation*}p^{-\tau}(W) \ := \sum\limits_{\stackrel{W^\prime \in \mc{W}}{W^\prime \smallsetminus \{\tau\} = W}}  p(W^\prime).
      \end{equation*}
      (c) \ For a BQ $\mc{Q}$, and $v \in \{0,1\}$:
    \begin{eqnarray*}P(\mc{Q}=v~|~\nit{do}(\tau \mbox{\phantom{o}}\In)) \ &:=& \
    p^{+\tau}( \{W \in \mc{W}~|~\mc{Q}[W]=v\} ),\\
    P(\mc{Q}=v~|~\nit{do}(\tau \Out)) \ &:=& \
    p^{-\tau}( \{W \in \mc{W}~|~ \mc{Q}[W] =v\}).
\end{eqnarray*}

\vspace{-0.8cm}\boxtheorem
\end{definition}

\begin{remark} \label{rem:inouts} \ (a) \ In Definition \ref{def:DO}, the PDB could be a TID.

\vspace{1mm}\noindent (b) In Definition \ref{def:DO}(c), and similarly for the other cases: \ $P(\mc{Q}=1~|~\nit{do}(\tau \In))
\ = \!\!\!\! \sum_{W \in \mc{W}, W\cup \{\tau\} \models \mc{Q}}  p(W)$.

\vspace{1mm}\noindent (b) \ For  $\tau$ as a ground query, and as intuitively expected: \ $P(\tau~|~\nit{do}(\tau \Out))  = 0$, and $P(\tau~|~\nit{do}(\tau \In)) := P(\tau = 1~|~\nit{do}(\tau \In)) = \sum_{W \in \mc{W}, \ \tau \in W\cup \{\tau\}} p(W) = 1$.

\vspace{1mm}\noindent
(c) \ When \ $\tau \notin W \in \mc{W}, \ p^{+\tau}(W) = 0$; \ and, when $\tau \in W \in \mc{W}, \ p^{-\tau}(W) = 0$.

\vspace{1mm} \noindent  (d) \ For a TID $D$ and different $\tau,\tau^\prime \in D^\nit{en}$: \ $P^{+\tau}(\tau^\prime) = P^{-\tau}(\tau^\prime) = P(\tau^\prime)$,
where
\begin{eqnarray*}
P^{+\tau}(\tau^\prime) &:=& \sum_{W \in \mc{W}, \tau^\prime \in W}p^{+\tau}(W), \ \ \mbox{ and }\\
P^{-\tau}(\tau^\prime) &:=& \sum_{W \in \mc{W}, \tau^\prime \in W}p^{-\tau}(W).
\end{eqnarray*}
That is, an intervention $\nit{do}(\tau \In)$ ($\nit{do}(\tau \Out)$, resp.) on a TID changes the probability of $\tau$ to 1 \ (0, resp.), leaving the other probabilities unchanged. \boxtheorem
\end{remark}

\begin{example} \label{ex:pdb} \ (ex. \ref{ex:tid} cont.) \ Assume all tuples in the  TID $D$ are endogenous. $\mc{W}$  contains all the subsets of $R$; one of them is $W = \{\tau_1,\tau_3,\tau_5\}$, for which:

\begin{eqnarray*}p^{+\tau_3}(W)   &:=& \sum\limits_{\stackrel{W^\prime \in \mc{W}}{W^\prime \cup \{\tau_3\} = W}}  p(W^\prime) \ = \ p(W) + p(\{\tau_1,\tau_5\})\\ &=& \ p(W) \ + \ 0.8 \times (1-0.3) \times (1-0.2) \times (1-0.6) \times 1.\\
p^{-\tau_3}(W) \ &:=& \sum\limits_{\stackrel{W^\prime \in \mc{W}}{W^\prime \smallsetminus \{\tau_3\} = W}}  p(W^\prime) \ = \ \sum_{\emptyset} p(W^\prime) \ = \ 0.
\end{eqnarray*}

\vspace{-1cm}\boxtheorem
\end{example}

\begin{definition} \label{def:gce} \em
    Let $D$ be a relational instance with an associated $U(\frac{1}{2})$-TID $D^p = $  $\langle \mc{W},$ $p \rangle$, and
    $\mc{Q}$ a BQ. \ The {\em causal-effect score}, denoted  {\sf CES}, of $\tau \in  D^{\nit{en}}$ on $\mc{Q}$ is:

    \vspace{2mm}
    \hspace*{2cm}$\CE(D,\mc{Q},\tau) \ := \ \mathbb{E}(\mc{Q}~|~ \nit{do}(\tau \In)) - \mathbb{E}(\mc{Q}~|~\nit{do}(\tau \Out))$. 
    \boxtheorem
\end{definition}

{\sf CES} captures the expected difference in  query value between  having and not having the tuple in the DB. Recall that every BQ becomes a random variable over the intervened possible worlds. \ For a BQ, it holds:
\begin{eqnarray*}\CE(D,\mc{Q},\tau) &=& \sum_{v \in \{0,1\}} v \times P(\mc{Q} =v~|~\nit{do}(\tau \In)) - \sum_{v \in \{0,1\}} v \times P(\mc{Q} =v~|~\nit{do}(\tau \Out))\\ &=&
P(\mc{Q} = 1~|~\nit{do}(\tau \In)) - P(\mc{Q} = 1~|~\nit{do}(\tau \Out)).
\end{eqnarray*}

\begin{example} \label{ex_paths}  \cite{Causal_Effect} \
	Given instance $D$, we are interested in computing the causal attribution scores of the tuples in relation to making true the  query that asks if there
 is a path from $a$ to $b$, on the basis of relation $E$ showing direct connections between nodes. \ This BQ $\mc{Q}$ can be expressed in Datalog (or a union of CQs for a fixed instance).

 In order to apply {\sf CES},   we build a  $U(\frac{1}{2})$-TID, as shown in Table \ref{tab:my-labelNEW}(a). It holds: \ $\CE(D,\mc{Q},\tau_1) =
   0.65625$, $\CE(D,\mc{Q},\tau_2) = \CE(D,\mc{Q},\tau_3) = 0.21875$, and
   $\CE(D,\mc{Q},\tau_4) = \CE(D,\mc{Q},\tau_5) = \CE(D,\mc{Q},\tau_6) = 0.09375$.

As noticed in \cite{Causal_Effect},  the responsibilities are all  $1\!/3$, despite the fact that they make the query true through paths of different lengths.
   \boxtheorem
   \end{example}

In \cite{Causal_Effect}, it was established that, for a MBQ $\mc{Q}$, a database  $D$, and $\tau \in D^\nit{en}$,  $\tau$ is an \emph{actual cause} of $\mc{Q}$ in $D$ iff $\CE(D,\mc{Q},\tau) > 0$.\footnote{By definition, {\sf Resp}  is also non-negative (see Section \ref{sub:resp}). Moreover, for MBQs, {\sf Shapley}  is non-negative (see Section \ref{sub:shapley-banzhaf}): The factor $\Delta$ of Equation (\ref{eq:shap}) is either 0 or 1 for MBQs, and the coefficient that multiply it is strictly positive.} \
In \cite{Shapley_Tuple_Bertossi}, it was shown that CES coincides with {\sf BPI}: \
$\CE(D,\mc{Q},\tau) \ = \ \nit{BPI}(D,\mc{Q},\tau)$. \ignore{++From the results in \cite{Shapley_Tuple_Bertossi} for the BPI, we obtain that: (a) For hierarchical BCQs without self-joins, computing {\sf CES} can be done in polynomial-time in data complexity, but (b) for non-hierarchical queries, computing {\sf CES} is $\#P$-hard.++}

Notice that {\sf CES} uses an auxiliary PDB, but the original DB is not probabilistic, and {\sf CES} is not a probabilistic score in that the value is not a probability. The score  in Definition \ref{def:gce} coincides with that defined in \cite{Causal_Effect} via the query lineage.
\  {\sf Resp}, {\sf Shapley} and {\sf BPI} of Section \ref{sec:back} are not probabilistic scores, and they do not use an intermediate PDB. \
\emph{In the rest of this work, we will consider only $U(\frac{1}{2})$-TIDs. \ In particular, {\sf CES} will always be applied with this distribution of a DB.}


\section{Score Alignment}\label{sec:alig}

In this section, we define score alignment, and provide motivating examples.

\begin{example} \ (ex. \ref{ex_paths} cont.)
    \label{ex:paths_detailedNEW}
    \  {\sf CES}, {\sf Resp} and {\sf Shapley}  are shown in Table \ref{tab:my-labelNEW}(b). \ Here, the less informative score is {\sf  Resp}, which assigns $\frac{1}{3}$ to all the tuples in $D$, despite the fact that the numbers of tuples in each path are different. \ {\sf CES} and {\sf Shapley} return different scores, but produce the same qualitative {\em rankings} for the tuples, i.e. they are equally ordered according to their scores. \boxtheorem
\end{example}

  \begin{table}[]
           \begin{center}
    {\scriptsize
            $\begin{tabu}{l|c|c||c|}
				\hline
				E^p~  & ~~A~~ & ~~B~~ & p \\ \hline
				\tau_1 & a & b & \mbox{1/2}\\
				\tau_2& a & c& \mbox{1/2}\\
				\tau_3& c & b& \mbox{1/2}\\
				\tau_4& a & d& \mbox{1/2}\\
				\tau_5& d & e& \mbox{1/2}\\
				\tau_6& e & b& \mbox{1/2}\\ \cline{2-4}
			\end{tabu}$}~~~~~~~~~~~~~~~~~
       {\scriptsize  \begin{tabular}{c|c|c|c}
             ~~Tuples~~& ~~{\sf CES} (and {\sf BPI})~~ & ~~{\sf Resp}~~ & ~~{\sf Shapley}~~\\
             \hline
             $\tau_1$               & ~~0.65625~~ & 1/3 & 0.5833 \\
             $\tau_2, \tau_3$        & ~~0.21875~~ & 1/3 & 0.1333 \\
             ~~$\tau_4,\tau_5, \tau_6~~$ & ~~0.09375~~ & 1/3 & 0.05 \\ \hline
        \end{tabular}}
       \vspace{2mm} \caption{\hspace*{0.5cm} (a) \ \ A $U(\frac{1}{2})$-TID $D$. \hspace{0.5cm}   (b) \ \ {\sf CES}, {\sf Resp} and {\sf Shapley} for tuples in $D$. }
        \label{tab:my-labelNEW} \end{center}
\end{table}

Intuitively, two scores are aligned if the induced orders are compatible.

\begin{definition} \em \label{def:rank_alignment}  (score alignment)
    \ Let $D$ be a DB, $\mc{Q}$ a BQ, with $D \models \mc{Q}$, and  $\tau,\tau^\prime \in D^\nit{en}$. \  (a) \ For a  score ${\sf sc}(D,\mc{Q},\cdot)$ as a real-valued function on $D^\nit{en}$, its associated \emph{ranking} is the total preorder relation on $D^\nit{en}$ defined by: \ $\tau \preceq^{\sf sc} \tau^\prime$ \ iff \ ${\sf sc}(D,\mc{Q},\tau) \leq {\sf sc}(D,\mc{Q},\tau^\prime)$, with the strict associated relation: \  $\tau \prec^{\sf sc} \tau^\prime$ \ iff \ ${\sf sc}(D,\mc{Q},\tau) < {\sf sc}(D,\mc{Q},\tau^\prime)$, for which it holds: \ If $\tau \prec^{\sf sc} \tau^\prime$, then $\tau \preceq^{\sf sc} \tau^\prime$, but not $\tau^\prime \preceq^{\sf sc} \tau$.

    \vspace{1mm}\noindent
    (b) \ Score ${\sf sc}^\prime$ is {\em aligned} with score ${\sf sc}$ for $(\mc{Q},D)$ if, for every $\tau,\tau^\prime \in D^\nit{en}$: when ${\sf sc}(\tau) \leq {\sf sc}(\tau^\prime)$ holds, then also ${\sf sc}^\prime(\tau) \leq {\sf sc}^\prime(\tau^\prime)$ holds.

    \vspace{1mm}\noindent (c) Scores ${\sf sc}, \ {\sf sc}^\prime$ are \emph{aligned} for $(\mc{Q},D)$ when at least one of them is aligned with the other. \ Scores ${\sf sc}, \ {\sf sc}^\prime$ are {\em aligned} for $\mc{Q}$ \
    when, for every instance $D$ compatible with $\mc{Q}$, they are aligned for $(\mc{Q},D)$.

    \vspace{1mm}\noindent
    (d) \ $\preceq^{\CE}$, \ $\preceq^{\nit{Re}}$, and \ $\preceq^{\nit{Sh}}$ denote the preorders for {\sf CES}, {\sf Resp} and {\sf Shapley}, resp. \boxtheorem
\end{definition}

Table \ref{tab:my-labelNEW} in Example \ref{ex:paths_detailedNEW} shows that all pairs of scores are aligned.\
The counterexamples right below show that there are pairs $(\mc{Q},D)$ for which rankings are not aligned.

\begin{table}
    \centering
    {\scriptsize
    $\begin{tabu}{l|c|}
        \hline
        R~  & ~~X~~ \\\hline
        \tau_1 & a\\
        \tau_2& b\\
        \tau_3& e\\
        \cline{2-2}
    \end{tabu}$~~~~~~
    $\begin{tabu}{l|c|c|}
        \hline
        S & ~~X~~ & ~~Y~~\\
        \hline
        \tau_4 & a & b\\
        \tau_5 & a & c\\
        \tau_6 & a & d\\
        \tau_7 & b & b\\
        \tau_8 & b & c\\
        \tau_9 & b & d\\
        \tau_{10} & e & f\\
        \cline{2-3}
    \end{tabu}$~~~~~~
    $\begin{tabu}{l|c|}
        \hline
        T & ~~Y~~\\
        \hline
        \tau_{11} & b\\
        \tau_{12} & c\\
        \tau_{13} & d\\
        \tau_{14} & f\\
        \tau_{15} & g\\
        \cline{2-2}
    \end{tabu}$~~~~~~~~~~~
    $\begin{tabu}{l|c|c|}
        \hline
        \tau   & ~~{\sf CES}~~ & ~~{\sf Resp}~~\\
        \hline
        \tau_3 & ~~0.1292~~ & ~~1/3~~\\
        \tau_4 & ~~0.0829~~ & ~~1/5~~\\
        \tau_{11} & ~~0.1868~~ & ~~1/4~~\\
        \tau_{15} & ~~0~~ & ~~0~~\\
        \cline{2-3}
    \end{tabu}$
}
    \vspace{2mm}\caption{ \hspace*{1cm} (a) \ \ Instance $D$. \hspace{1cm} (b) \ \ {\sf CES} and {\sf Resp} for $(D,\mc{Q}_{\sf RST})$. }
    \label{tab:ce_resp}
\end{table}

\begin{example} \ ({\sf CES} vs {\sf Resp}) \label{ex:comp_ces_resp_non_hierarchical}
Consider instance $D$ in Table \ref{tab:ce_resp}(a), with only endogenous tuples, and the query \ $\mc{Q}_{\sf RST}\!: \ \exists x \exists y(R(x) \wedge S(x,y) \wedge T(y))$. 

{\sf CES} and {\sf Resp}  for $\tau_3$, $\tau_4$, $\tau_{11}$ and $\tau_{15}$ are shown in Table \ref{tab:ce_resp}(b). \  The induced orders are: \ $\tau_{15} \prec^{\CE} \tau_{4} \prec^{\CE} \tau_3 \prec^{\CE} \tau_{11}$ \ and  \ $\tau_{15} \prec^{\nit{Re}} \tau_{4} \prec^{\nit{Re}} \tau_{11} \prec^{\nit{Re}} \tau_{3}$. \  Then, {\sf Resp} and {\sf CES} are not aligned for $(\mc{Q}_{\sf RST},D)$. \boxtheorem
\end{example}

\begin{table}
    \centering
    {\scriptsize
    $\begin{tabu}{l|c|c|}
        \hline
        R~  & ~~X~~ & ~~Y~~ \\\hline
        \tau_1 & a & c_1\\
        \tau_2& b & c_2\\
        \tau_3& b & c_3\\
        \cline{2-3}
    \end{tabu}$~~~~~~~~~
    $\begin{tabu}{l|c|c|}
        \hline
        S & ~~X~~ & ~~Z~~\\
        \hline
        \tau_4 & a & c_4\\
        \tau_5 & a & c_5\\
        \tau_6 & b & c_6\\
        \tau_7 & b & c_7\\
        \tau_8 & b & c_8\\
        \tau_9 & b & c_9\\
        \tau_{10} & c & c_1\\
        \cline{2-3}
    \end{tabu}$~~~~~~~~~~~~
    $\begin{tabu}{l|c|c|c|}
        \hline
        \tau   & ~~{\sf CES}~~ & ~~{\sf Resp}~~ & ~~{\sf Shapley}~~\\
        \hline
        \tau_1 & ~~57/256~~ & ~~1/3~~ & ~~400/2520~~\\
        \tau_4 & ~~19/256~~ & ~~1/4~~& ~~151/2520~~\\
        \tau_6 & ~~15/256~~ & ~~1/5~~& ~~169/2520~~\\
        \tau_{10} & ~~ 0 ~~ & ~~0  ~~& ~~0~~\\
        \cline{2-4}
    \end{tabu}$
    }
   \vspace{2mm}
    \caption{\hspace*{1cm} (a) \ \ Instance $D_{\sf RS}$. \hspace{1cm} (b) \ \ {\sf CES}, {\sf Resp} and {\sf Shapley} for $(\mc{Q}_{\sf RS},D_{\sf RS})$.}
    \label{tab:ce_shapley}
\end{table}

\begin{example} \ ({\sf CES} and {\sf Resp} vs {\sf Shapley}) \ \label{ex:comp_ces_shapley}
    Consider instance $D_{\sf RS}$  in Table \ref{tab:ce_shapley}(a), with only endogenous tuples, and the query
    \begin{equation} \label{eq:qrs}
        \mc{Q_{\sf RS}} \!: \ \exists x \exists y \exists z (R(x,y) \land S(x,z))
    \end{equation}
    {\sf CES}, {\sf Resp} and {\sf Shapley}  for $\tau_1$,$\tau_4, \tau_6$ and $\tau_{10}$ are shown in Table \ref{tab:ce_shapley}(b): \  $\tau_{10} \prec^{\CE} \tau_6 \prec^{\CE}\! \tau_4 \prec^{\CE}\! \tau_1$ \ and \ $\tau_{10} \prec^{\nit{Re}} \tau_6 \prec^{\nit{Re}}\! \tau_4 \prec^{\nit{Re}}\! \tau_1$. \
    {\sf CES} and {\sf Resp} are aligned (for these tuples). However,  $\tau_{10} \prec^{\nit{Sh}} \! \tau_4 \prec^{\nit{Sh}} \! \tau_6 \prec^{\nit{Sh}}\! \tau_1$. Then, {\sf Shapley} is not aligned with  {\sf CES} or {\sf Resp} for $(\mc{Q}_{\sf RS},D_{\sf RS})$.  \boxtheorem
\end{example}

As the examples above show, the notion of alignment captures the fact that two scores do not ``contradict" each other: Two scores are aligned when they produce the same rankings, independently of the exact values; and are not aligned if the rankings of two particular tuples are reversed. \
As a consequence, if a score gives the same numerical value to all tuples, then it is aligned with any other score, when it also assigns a constant value or not. This is highlighted in Remark \ref{rem:same-value} below, which is used to prove several of the results below.

\begin{remark} \label{rem:same-value} \
    Given a BCQ $\mc{Q}$ and a database instance $D$, if a score ${\sf sc}_c$ gives the same constant $c$ to all endogenous tuples in $D$ for $(\mc{Q},D)$, then, by Definition \ref{def:rank_alignment}(b), it is aligned with any other score ${\sf sc}$; and then, by Definition \ref{def:rank_alignment}(c), both scores will be aligned for $(\mc{Q},D)$. \boxtheorem
\end{remark}

In principle, we could have defined the alignment of two scores as follows: \ $sc$ and $sc^\prime$ are aligned when, for every $\tau,\tau^\prime$: \ ${\sf sc}(\tau) \leq {\sf sc}(\tau^\prime)$ iff ${\sf sc}^\prime(\tau) \leq {\sf sc}^\prime(\tau^\prime)$. With this alternative definition, that uses an {\em iff}, Remark \ref{rem:same-value} does not hold unless the other score is also constant.\
For instance, consider the scores ${\sf sc}$ and ${\sf sc}^\prime$, with ${\sf sc}(\tau) = 1/2$ for every $\tau$, that is, constant; and ${\sf sc}^\prime$ that either assigns $1/3$ or $2/3$. The scores ${\sf sc}$ and ${\sf sc}^\prime$ are not aligned by this definition, even though the score ${\sf sc}$ never ``contradicts'' the rankings given by ${\sf sc}$.

In a different direction, notice that, in Examples \ref{ex:comp_ces_resp_non_hierarchical} and
\ref{ex:comp_ces_shapley}, that $\tau_{15}$ and $\tau_{10}$, respectively, do not contribute for the query to be true in any subset of the instance they belong.\
These kind of tuples are called \emph{dummy tuples}, as formalized in next definition. Propositions \ref{prop:dummy_tuples_0} and \ref{prop:dummy_equal_scores} describe some of their useful properties.

\begin{definition} \label{def:dummy} \em
    Let $D = \Den \cup \Dex$ and $\mc{Q}$ a BCQ. \
  (a)  $\tau \in \Den$ is a \ \emph{dummy tuple} \ for $\mc{Q}$ \ if, for every $S \subseteq \Den$\!, \ $\mc{Q}[S \cup \Dex] = \mc{Q}[S \cup \Dex \cup \{\tau\}]$.\footnote{That is,  \ $\Delta(\mc{Q},S,\tau) =0$, \ as  defined in  (\ref{eq:delta}).} \ (b)  $\tau \in \Dex$ is a  \emph{dummy tuple} if $\mc{Q}[S] = \mc{Q}[S \cup \{\tau\}]$, for every $S \subseteq D \smallsetminus \{\tau\}$.  \boxtheorem
\end{definition}

Exogenous dummy tuples never affect the value of the query $\mc{Q}$.

\begin{proposition} \label{prop:dummy_tuples_0} \em
    Let $\mc{Q}$ be a BCQ and $D = \Den \cup \Dex \models \mc{Q}$. \ For every $\tau \in \Den$,
    the following statements are equivalent:
   \ (a) \ $\tau$ is a dummy tuple; \ and \ (b) \
        $\Resp(D,\mc{Q},\tau) = \CE(D,\mc{Q},\tau) = \Shapley(D,\mc{Q},\tau) = 0$. \boxtheorem
\end{proposition}

\noindent \textbf{Proof: \ } ($\Rightarrow$) \ If $\tau$ is dummy tuple, then, there is no contingency set such $\tau$ is an actual cause, since $\mc{Q}[S \cup \Dex \cup \{\tau\}]$ is always either 0 or 1, for all $S \subseteq \Den$. It follows that $\Resp(D,\mc{Q},\tau)$ for any $D$ and $\mc{Q}$.\
Moreover, since the factor $\Delta(\mc{Q},S,\tau)$ is $0$, then, by Equation (\ref{eq:banzhaf}) for {\sf BPI} (and hence, {\sf CES}) and Equation (\ref{eq:shap}) for {\sf Shapley}, the corresponding scores are also 0.

\vspace{1mm}\noindent ($\Leftarrow$) \ If $\Resp(D,\mc{Q},\tau)$, then there is no contingency set $\Gamma \subseteq \Den$ such $\mc{Q}[D \smallsetminus \Gamma] = 1$ and $\mc{Q}[D \smallsetminus (\Gamma \cup \{\tau\})] = 0$, it follows that either $\mc{Q}[S \cup \Dex \{\tau\}]$ is $0$ or $1$, for all $S \subseteq \Den$. Therefore, $\tau$ is dummy.\
For Equation (\ref{eq:banzhaf}) for {\sf BPI} ({\sf CES}) and Equation (\ref{eq:shap}) for {\sf Shapley} to be 0, $\Delta(\mc{Q},S,\tau)$ must be also $0$, since their coefficients are always positive, and $\mc{Q}[S]$ cannot be negative (due to monotonicity). The statement then follows. \boxtheorem

\begin{proposition} \label{prop:dummy_equal_scores} \em
    Let $\mc{Q}$ be a BCQ, and $D$ an instance, such that $D \models \mc{Q}$, and  $\tau_d \in D$ be an endogenous dummy tuple for $\mc{Q}$. \ For every $\tau \in (D^\nit{en} \smallsetminus \{\tau_d\})$, it holds: \ \ ${\sf sc}(D,\mc{Q},\tau) \ = \ {\sf sc}(D \smallsetminus \{\tau_d\}, \mc{Q},\tau)$, \ for \ ${\sf sc} \in \{\Resp,\CE, \Shapley \}$.
\end{proposition}

\vspace{1mm}
\noindent {\bf Proof: \ }
For \nit{Re}: \ Since the tuple $\tau_d$ does not change the value of the query, it does not belong to any contingency set of $\tau$ for the answer of $\mc{Q}$. Therefore,  {\sf Resp} stays  the same with or without $\tau_d$.

\vspace{1mm}
For {\sf CES}: \ $\CE(\tau)$ is the probability that $\tau$ changes the value of $\mc{Q}$ when each other endogenous tuple is independently present with probability $1/2$. Since $\tau_d$ never changes the value of $\mc{Q}$, this probability does not depend on whether $\tau_d$ is present, so it is the same in $D \smallsetminus \{\tau_d\}$.

\vspace{1mm}
For {\sf Shapley}: \ $\Shapley(\tau)$ is the average, over all the orderings of $D^\nit{en}$, of the change in the value of $\mc{Q}$ when $\tau$ is added to the tuples that precede it. Since $\tau_d$ never changes the value of $\mc{Q}$, this change does not depend on whether $\tau_d$ precedes $\tau$. Deleting $\tau_d$ from an ordering of $D^\nit{en}$ does not change the effect of adding $\tau$, so, the average, (the {\sf Shapley}), stays the same.
\boxtheorem

We can see that endogenous dummy tuples do not have an effect on the alignment of scores. Actually, the same holds for exogenous dummy tuples, but we will use not this property.

\begin{remark} \label{rem:dummy} \ The two last propositions allow us to make the assumption that, when proving alignment or non-alignment results, $D$ has no endogenous dummy tuples: Their value given by any score is 0, while non-dummy tuples receive strictly positive scores (see end of Section \ref{sec:CE} and Proposition \ref{prop:dummy_tuples_0}). Furthermore, iteratively removing endogenous dummy tuples from $D$ does not change the scores of the remaining tuples.\boxtheorem
\end{remark}

The next proposition provides a condition for when two pairs $(\mc{Q},D)$ and $(\mc{Q}^\prime,D^\prime)$ assign the same scores to their tuples. In particular, it allows us to work directly with the lineage when comparing scores, and shows that the alignment of two scores for $(\mc{Q},D)$ is equivalent to check the alignment for $(\mc{Q}^\prime,D^\prime)$.

\begin{proposition} \label{prop:same-values-lineage} \em
    Let $\mc{Q}$, $\mc{Q}^\prime$ be BCQs, $D$, $D^\prime$ database instances, and $f\!: \Den \to (D^\prime)^\nit{en}$ a bijection. Then, it holds: \  For every $S \subseteq \Den$, $\mc{Q}[S \cup D^\nit{ex}] = \mc{Q}^\prime[f(S) \cup (D^\prime)^\nit{ex}]$ iff the ex-reduced lineages $\mc{L}(\mc{Q},D)$ and $\mc{L}(\mc{Q}^\prime,D^\prime)$ are logically equivalent after replacing in the former each variable $X_\tau$ by $X_{f(\tau)}$. In that case, for every $\tau \in \Den$, ${\sf sc}(D,\mc{Q},\tau) = {\sf sc}(D^\prime,\mc{Q}^\prime,f(\tau))$, for ${\sf sc} \in \{\Resp, \CE, \Shapley\}$.
\end{proposition}

\noindent {\bf Proof: \ } Denote with $\sigma_S$ the truth assignment for the variables $X_\tau$, with $\tau \in D^\nit{en}$, such that, $\sigma_S(X_\tau) = 1$ iff $\tau \in S$, for some $S \subseteq D^\nit{en}$.\
$\mc{L}(\mc{Q},D)$ is true under $\sigma_S$ iff $\mc{Q}[S \cup \Dex] = 1$.\
Since $f$ is a bijection, every truth assignment $\sigma_S$ has its corresponding truth assignment $\sigma_{f(S)}$ for the variables $X_{f(\tau)}$: $\sigma_{f(S)}(X_{f(\tau)}) = 1$ iff $f(\tau) \in f(S)$.

Notice that the assignments $\sigma_{f(S)}$ cover all possible truth assignments for $X_{f(\tau)}$, with $f(\tau) \in (D^\nit{en})^\prime$.\
Then, $\mc{L}(\mc{Q}^\prime,D^\prime)$ is true under $\sigma_{f(S)}$ iff $\mc{Q}^\prime[f(S) \cup (D^\prime)^\nit{ex}] = 1$; and $\mc{Q}^\prime[f(S) \cup (D^\prime)^\nit{ex}] = 1$ iff $\mc{Q}[S \cup \Dex] = 1$. $\mc{L}(\mc{Q},D)$ and $\mc{L}(\mc{Q}^\prime,D^\prime)$ are logically equivalent when replacing each $X_{\tau}$ in $\mc{L}(\mc{Q},D)$ by $X_{f(\tau)}$. \ Now, assuming this holds, we now prove the last claim in the proposition.

\vspace{1mm} \noindent {\sf Resp}: \  Let $\tau \in \Den$ be an actual cause for $\mc{Q}$. Let $\Gamma$ be a contingency set for $\tau$. Then, $\mc{Q}[(\Den \smallsetminus \Gamma) \cup \Dex] = 1$ and $\mc{Q}[(\Den \smallsetminus (\Gamma \cup \{\tau\})) \cup \Dex] = 0$. By hypothesis, $\mc{Q}^\prime[f(\Den \smallsetminus \Gamma) \cup (D^\prime)^\nit{ex}] = 1$ and $\mc{Q}^\prime[f(\Den \smallsetminus (\Gamma \cup \{\tau\})) \cup (D^\prime)^\nit{ex}] = 0$. Since $f$ is a bijection, $f(\Den \smallsetminus \Gamma) = (D^\prime)^\nit{en} \smallsetminus f(\Gamma)$ and $f(\Den \smallsetminus (\Gamma \cup \{\tau\})) = (D^\prime)^\nit{en} \smallsetminus (f(\Gamma) \cup \{f(\tau)\})$. \ Then,  $f(\Gamma)$ is a contingency set for $f(\tau)$. Moreover, $|\Gamma| = |f(\Gamma)|$. The same argument is made for a contingency set $\Gamma^\prime$ for $f(\tau)$: $f^{-1}(\Gamma^\prime)$ is a contingency set for $\tau$.\
Additionally, by the above, $\tau$ is not an actual cause for $\mc{Q}$ iff $f(\tau)$ is not an actual cause for $\mc{Q}^\prime$.\
It follows that $\Resp(D,\mc{Q},\tau) = \Resp(D^\prime,\mc{Q}^\prime,f(\tau))$.

\vspace{1mm} \noindent \emph{{\sf CES} and {\sf Shapley}}: For a given $\tau \in \Den$, consider that,  for a $S \subseteq \Den \smallsetminus \{\tau\}$, $\Delta(\mc{Q},S,\tau) = 1$ (see Equation (\ref{eq:delta})) iff $\mc{Q}[S \cup \{\tau\} \cup \Dex] = 1$, and $\mc{Q}[S \cup \Dex] = 0$. \
Similarly, $\Delta(\mc{Q}^\prime, f(S), f(\tau)) = 1$ iff $\mc{Q}^\prime[f(S \cup \{\tau\}) \cup (D^\prime)^\nit{ex}] = 1$ and $\mc{Q}^\prime[f(S) \cup (D^\prime)^\nit{ex}] = 0$.
By hypothesis, it follows that $\Delta(\mc{Q},S,\tau) = \Delta(\mc{Q}^\prime,f(S),f(\tau))$. \

Consider Equation (\ref{eq:banzhaf}) for {\sf CES}, noticing that the query only affects the $\Delta$ factor. Furthermore, since $f$ defines a bijection between subsets of $\Den$ and subsets of $(D^\prime)^\nit{en}$, with $|S| = |f(S)|$, it holds $1/2^{|\Den|} = 1/2^{|f(\Den)|}$. \ Thus, $\CE(D,\mc{Q},\tau) = \CE(D^\prime,\mc{Q}^\prime,f(\tau))$.\

Similarly for {\sf Shapley}: By Equation (\ref{eq:shap}), the query only affects the $\Delta$ factor. Since $f$ defines the a bijection between subsets of $\Den$, $|S|! \times (|\Den| - |S| - 1 )/(|\Den|!) = |f(S)|! \times (|f(\Den)| - |f(S)| - 1 )/(|f(\Den)|!)$. It follows that $\Shapley(D,\mc{Q},\tau) = \Shapley(D^\prime,\mc{Q}^\prime,f(\tau))$.
\boxtheorem

In Definition \ref{def:rank_alignment}, about score alignment, we started with the natural condition that  $D \models \mc{Q}$. We will continue making this assumption for the rest of this work.  Actually, scores such as {\sf Resp} are defined on the assumption that $D$ satisfies $\mc{Q}$,  and we want to explain why this is the case.

\begin{remark} \label{rem:sat}
 It is easy to verify that, when $\Dex \models \mc{Q}$, all our scores take the constant value $0$, and then, are aligned. \ Accordingly, we will  assume that $\Dex \not\models \mc{Q}$. \ This is not a strong restriction: Even when some whole relations may be declared as exogenous, we expect queries to involve predicates whose extensions have endogenous tuples. \boxtheorem
\end{remark}

Finally, all the definitions and propositions presented in this section are also valid for MBQs. However, from now on, \emph{we will assume that queries are always BCQs}.


\section{{\sf CES} and {\sf Resp} Alignment} \label{sec:ces-resp}

In this section, we investigate the alignment of  {\sf CES} and {\sf Resp}, with a BCQ $\mc{Q}$ as our main parameter. That is,  the  conditions on $\mc{Q}$ under which the scores are aligned for $\mc{Q}$ \ {\em for every instance}; and also, whether there is an instance $D$ for which they are not aligned.

The main result is the {\em dichotomy} provided by Theorem \ref{theo:ces_resp_aligned}. It  identifies  a syntactic class of SJF-BCQ that characterizes their alignment: $\mc{Q}$ belongs to this class if and only if the scores are aligned for any possible instance.

{\em In this section, actually for the rest of this work except for Section \ref{sec:sj-bcq}, we assume that all queries are SJF}. \ In Section \ref{sec:sj-bcq}, queries that allow self-joins are considered, and  particular alignment and non-alignment results will be provided. A general syntactic characterization of the alignment of {\sf CES} and {\sf Resp} for general BCQs is still open.

The presence of exogenous tuples in a database is permitted; and the alignment results may be affected by their presence or absence.\
In Section \ref{sec:ces-resp_noex}, we investigate alignment in the absence of exogenous tuples. Until then, we allow instances to have them.

Before introducing the main result, we require some syntactic notion in relation to queries.

\begin{definition} \label{def:component_coinvar} \ \em (query components and coincident variables) \
    Let $\mc{Q}$ be a BCQ, say $\mc{Q}: \ \exists  \bar{x}\;(A_1(\bar{x}_1)\wedge \dots \wedge A_m(\bar{x}_m))$.
    \vspace{1mm}
    \noindent  (a) The undirected graph $\mc{G}(\mc{Q}) = \langle G, E\rangle$ associated to $\mc{Q}$ has the set of nodes $G$ formed by the atoms of $\mc{Q}$, and the set $E$ of edges, formed by pairs of atoms, $\{A_i,A_j\}$, that have at least one variable in common.

    \vspace{1mm}
    \noindent  (b) The {\em components of}  $\mc{Q}$ are the connected components of $\mc{G}(\mc{Q})$. They are  denoted with $C_1,\ldots,C_n$. For each component $C_i$,
    $\mc{Q}_i$ denotes the existential closure of the conjunction of atoms in  $C_i$. We call $\mc{Q}_i$ {\em a subquery} of $\mc{Q}$.
\end{definition}

\begin{example} \label{ex:coin} \
    The graph $\mc{G}(\mc{Q})$ associated to the query \ $\mc{Q}\!: \ \exists x \exists y \exists z \exists w (R(x)\land S(x,y) \land T(y) \land U(z,w)\land P(z))$  has the set of nodes \ $\nit{Atoms}(\mc{Q}) = \{R(x), S(x,y),$ $T(y),U(z,w), P(z)\}$. \ The edges are $\{R(x),S(x,y)\}$, $\{S(x,y),T(y)\}$ and \linebreak $\{U(z,w),P(w)\}$. The components are: \ $C_1 = \{R(x),S(x,y),T(y)\}$ and $C_2 = \{U(z,w),P(w)\}$, with associated subqueries  $\mc{Q}_1\!: \exists x \exists y (R(x) \land S(x,y) \land T(y))$ and $\mc{Q}_2 \!: \exists z \exists w (U(z,w) \land P(z))$. \ Components can be conceived as ``independent" subqueries.
    \boxtheorem
\end{example}

We now state the main result of this section.

\begin{theorem}[CES vs Responsibility] \label{theo:ces_resp_aligned} \em
    Let $\mc{Q}$ be an SJF-BCQ, \   with components $C_1,\ldots,C_n$, and their associated subqueries $\mc{Q}_{1}, \ldots, \mc{Q}_n$.\
    It holds:

    \begin{enumerate}[(a)]
    \item {\sf CES} and  {\sf Resp} are aligned for $(\mc{Q},D)$ for every instance $D$  with or without exogenous tuples if: \label{theo:ces_resp_aligned-a}
    \begin{enumerate}[(i)]
        \item $n=1$ and the variables of all of the atoms are the same; or if \label{theo:ces_resp_aligned-coin1}
        \item $n=1$ and the query has at most two atoms; or if \label{theo:ces_resp_aligned-atoms2}
        \item $n \geq 2$ and, for every $i = 1,\ldots,n$, $|\nit{Atoms}(\mc{Q}_i)| = 1$. \label{theo:ces_resp_aligned-compn}
    \end{enumerate}
        \item In any other case, there is an instance $D$ with  exogenous tuples for which {\sf CES} and {\sf Resp} are not aligned.\footnote{This does not prevent the existence, for some queries in this class, of an instance without exogenous tuples where the scores are not aligned.} \label{theo:ces_resp_aligned-b} \boxtheorem
    \end{enumerate}
\end{theorem}

We can obtain an immediate consequence of this theorem: Checking if a query falls or not in the class of queries identified in it, is a simple syntactic verification.

\begin{corollary} \em
 Given an SJF-BCQ $\mc{Q}$, it can be decided in linear time in the size of $\mc{Q}$ whether {\sf CES} and {\sf Resp} are aligned for $(\mc{Q},D)$ for every instance $D$. \boxtheorem
\end{corollary}

\noindent {\bf Proof of the Corollary:}  \
Checking the conditions of Theorem \ref{theo:ces_resp_aligned}(\ref{theo:ces_resp_aligned-a}) only requires computing the components of $\mc{Q}$, counting the atoms of $\mc{Q}$ and of the subqueries $\mc{Q}_i$, and checking whether all atoms of $\mc{Q}$ have the same set of variables. Each of these steps takes linear time in the size of $\mc{Q}$, and, therefore, it can be decided in linear time whether {\sf CES} and {\sf Resp} are aligned for $(\mc{Q},D)$ for every instance $D$, possibly with exogenous tuples. \boxtheorem

The proof of Theorem \ref{theo:ces_resp_aligned} will be concluded in Section \ref{sec:proof}, using results obtained in Sections \ref{sub:ces-resp-aligned} and \ref{sub:ces-resp-non-aligned}. At this point, and before going into the proof, we make some remarks and illustrate the theorem  with some examples.

Notice that a query with a single atom satisfies (i) and (ii). Furthermore, for a SJF-BCQ, the classification of the query depends on neither the order of the variables or its number of occurrences in an atom. Therefore, the syntactic condition in Theorem \ref{theo:ces_resp_aligned}(\ref{theo:ces_resp_aligned-a})(i) does not check whether each variable, in each atom, has the same position or appears the same number of times (see Example \ref{ex:theo-cr} below).\
This will not be the case for BCQs with self-joins, see Section \ref{sec:sj-bcq}.

\begin{example} \label{ex:theo-cr}
    The following table shows examples for each case of Theorem~\ref{theo:ces_resp_aligned}: \vspace{1mm}
    \begin{center}
    {\small
    $\begin{tabu} to \linewidth {X[4,l] X[1,c] X[2,c]}
        \rowfont{\bfseries}
        ~~~~~~~~~~~~~Query & Case & {\sf CES}/{\sf Resp} Aligned \\
        \hline  \phantom{oo} & \phantom{oo} & \phantom{oo}\\
        $\bar{\exists} (R(x,y,x) \land S(y,x) \land T(x,y))$ & (a)(i)   & Yes \\
        $\bar{\exists}  (R(x,y) \land S(x,z,z))$ & (a)(ii)  & Yes \\
        $\bar{\exists} (R(x,y) \land S(z) \land T(w,w))$ & (a)(iii) & Yes \\
        \hline  \phantom{oo} & \phantom{oo} & \phantom{oo}\\
        $\bar{\exists} (R(x) \land S(x,y) \land T(y))$ & (b) & No \\
        $\bar{\exists} (R(x,y,z,w) \land S(x,z))$ & (b) & No \\
        $\bar{\exists} (R(x,y) \land S(x) \land T(z,w) \land U(z,z))$ & (b) & No \\ \hline
    \end{tabu}$}
    \end{center}

    The first three queries satisfy condition (\ref{theo:ces_resp_aligned-a}) of Theorem~\ref{theo:ces_resp_aligned}. Then,  the scores are aligned for every instance, with or without exogenous tuples. Notice that, for all these queries, the order and number of occurrences of the variables is irrelevant.\
    For example, all of the atoms of the first query have different strings of variables, but  their sets of variables are the same; in this way, satisfying the conditions of Theorem \ref{theo:ces_resp_aligned} (\ref{theo:ces_resp_aligned-a})(i).
    \ For the remaining queries, an instance can be built where the scores are not aligned.\boxtheorem
\end{example}

\subsection{Constants and Variable Occurrences in SJF-BCQ} \label{sub:constants} \vspace{2mm}

The previous example shows queries without constants, however, we admit constants in query atoms. This does not create a problem for SJF-BCQ in terms of their classification, but the analysis and proofs can be simplified if we assume that queries do not have constants. In the rest of this section, we analyze this issue in more formal terms.

Consider a SJF-BCQ $\mc{Q}(\bar{a})$, where $\bar{a}$ represents the list of constants occurring in it, and an instance $D$, such that $D \models \mc{Q}(\bar{a})$.
Since $\mc{Q}(\bar{a})$ is a SJF-BCQ, every predicate occurs at most in one atom; and thus, every tuple of $D$ is associated to a single corresponding atom in $\mc{Q}(\bar{a})$, the one using the same predicate.

A tuple of $D$ that does not have, in the positions where its corresponding atom in $\mc{Q}(\bar{a})$ has constants, those same constants, becomes a dummy tuple. Accordingly, by Remark \ref{rem:dummy}, we may assume that $D$ does not have dummy tuples. Furthermore, all the tuples of a relation of $D$ have the same constants in the positions as their corresponding atom in $\mc{Q}(\bar{a})$. This allows us to safely remove those positions from both the query atoms (leaving only their variables) and from the tuples of the instance.

Removing constants requires changing the predicates' arities, and attributes in relations. Let us denote with $\mc{Q}^\prime$ the query obtained from $\mc{Q}(\bar{a})$ by removing the constants together with the corresponding predicates' positions.  Also, denote with $D^\prime$ the instance obtained from $D$ by removing from each relation those columns (attributes) corresponding to the constants' positions.

This process defines a bijection $f$ from $D$ to $D^\prime$ that preserves exogenous tuples, and $\mc{Q}(\bar{a})[W \cup D^\nit{ex}] = \mc{Q}^\prime[f(W) \cup (D^\prime)^\nit{ex}]$, for every $W \subseteq D^\nit{en}$. Then, by Proposition \ref{prop:same-values-lineage}, for every tuple $\tau \in D$, $\tau$ and $f(\tau)$ have the same {\sf CES}, {\sf Resp} and {\sf Shapley} scores. In particular, any two of those scores are aligned for $(\mc{Q}(\bar{a}),D)$ iff they are aligned for $(\mc{Q}^\prime,D^\prime)$.

Furthermore, if $D^{\prime\prime}$ is such, that $D^{\prime\prime} \models \mc{Q}^\prime$, for which the scores are not aligned, then there is an instance (a counterexample) satisfying  $\mc{Q}(\bar{a})$, where the scores are not aligned:\
Simply insert the constants $\bar{a}$ in the respective positions of the tuples of $D^\prime$.

We illustrate all this with the following example.

\begin{example} \label{ex:constants}
    Consider the query $\mc{Q}(a,b) \! : \exists x \exists y (R(x,y,a) \land S(x,a,b))$, with $a, b$ constants. By removing the constants, we obtain   $\mc{Q}^\prime \! : \exists x \exists y (R^\prime(x,y) \land S^\prime(x))$.

    \begin{table}[h]
        \centering
        {\scriptsize
        $\begin{array}{l|c|c|c|}
            \hline
            R & ~X~ & ~Y~ & ~Z~\\
            \hline
            \underline{\tau_{1}} & a & a & a \\
            \tau_{2} & a & b & a\\
            \tau_{3} & b & a & a\\
            \cline{2-4}
        \end{array}$ \quad $\begin{array}{l|c|c|c|}
            \hline
            S & ~X~ & ~Y~ & ~Z~\\
            \hline
            \tau_{4} & a & a & b\\
            \tau_{5} & b & a & b\\
            \cline{2-4}
        \end{array}$
        \quad $\stackrel{f}{\longmapsto}$ \quad
        $\begin{array}{l|c|c|}
            \hline
            R^\prime & ~X~ & ~Y~\\
            \hline
            \underline{\tau_{1}^\prime} & a & a \\
            \tau_{2}^\prime & a & b \\
            \tau_{3}^\prime & b & a \\
            \cline{2-3}
        \end{array}$ \quad $\begin{array}{l|c|}
            \hline
            S^\prime & ~X~ \\
            \hline
            \tau_{4}^\prime & a \\
            \tau_{5}^\prime & b \\
            \cline{2-2}
        \end{array}$~~~~~~~~~~~
        {\scriptsize $\begin{array}{|l|l|c|c|c|}
            \hline
            \tau & \tau^\prime & ~{\sf CES}~ & ~{\sf Resp}~ & ~{\sf Shapley}~\\
            \hline
            \tau_{2} & \tau_{2}^\prime & 0 & 0 & 0\\
            \tau_{3} & \tau_{3}^\prime & 0.25 & 1/2 & 1/6\\
            \tau_{4} & \tau_{4}^\prime & 0.75 & 1/2 & 2/3\\
            \tau_{5} & \tau_{5}^\prime & 0.25 & 1/2 & 1/6\\
            \cline{1-5}
        \end{array}$}
        }
      \vspace{2mm}  \caption{\hspace*{0.2cm} (a) Instance $D$. \ \hspace{0.2cm}(b) \  Instance $D^\prime$ from $D$. \ \hfill (c) \ Scores of  tuples in $\Den$ and $(D^\prime)^\nit{en}$.}
        \label{tab:constants}
    \end{table}

    Instance $D$ in Table \ref{tab:constants}(a), with $D^\nit{ex} = \{\tau_1\}$, satisfies $\mc{Q}(a,b)$.\
    $D$ is transformed into $D^\prime$ in Table \ref{tab:constants}(b), with $(D^\prime)^\nit{ex} = \{\tau_1^\prime\}$.
    \ Each tuple $\tau_i^\prime$ from $D^\prime$ is the same as its corresponding tuple $\tau_i$ in $D$, but with column $Z$ removed for tuples of $R$; and columns $Y$ and $Z$ removed for tuples of $S$: Those are the positions of the constants in the atoms of $\mc{Q}(a,b)$. Then, $f(\tau_i) = \tau_i^\prime$ is a bijection from $D_{\bar{a}}$ to $D^\prime$ that preserves exogenous tuples.

        Moreover, $f$ satisfies the condition of Proposition \ref{prop:same-values-lineage}, which can be verified through the query lineages: \ The ex-reduced lineage of $\mc{Q}(a,b)$ on $D$ is $X_{\tau_4} \lor (X_{\tau_2} \land X_{\tau_4}) \lor (X_{\tau_3} \land X_{\tau_5})$. \ That of $\mc{Q}^\prime$ on $D^\prime$ is $X_{\tau_4^\prime} \lor (X_{\tau_2^\prime} \land X_{\tau_4^\prime}) \lor (X_{\tau_3^\prime} \land X_{\tau_5^\prime})$. We obtain the same formulas after replacing each $X_{\tau_i}$ by $X_{f(\tau_i)}$.

        Again, by Proposition \ref{prop:same-values-lineage}, ${\sf sc}(D,\mc{Q}(a,b),\tau_i) = {\sf sc}(D^\prime,\mc{Q}^\prime,\tau^\prime_i)$, for ${\sf sc} \in \{\CE,\Resp,$ $\Shapley\}$, as shown in Table \ref{tab:constants}(c). Any pair of those scores are aligned for $(\mc{Q}(a,b),D)$ iff they are aligned for $(\mc{Q}^\prime,D^\prime)$.

    Furthermore, since $f$ is a bijection, we can also go the other way around: adding to each tuple of $D^\prime$ the constants of its atom in $\mc{Q}(a,b)$, in the same positions, gives back $D$, with the same scores. For example, for $\tau_4^\prime = S^\prime(a)$, its corresponding atom in $\mc{Q}(a,b)$ is $S(x,a,b)$; so, we can add the constants $a$ and $b$ in the second and third positions, resulting in $\tau_4 = S(a,a,b)$.
    \boxtheorem
\end{example}

\ignore{ \red{Notice that there is a injective mapping $f$ from $D^\prime$ to  $D$, the initial instance: \ $f(\tau^\prime_{i}) = \tau_i$, for $\tau^\prime_i \in D^\prime$ that preserves exogenous/endogenous tuples.}
 It also holds: (a) $\mc{Q}^\prime[W] = \mc{Q}_{\bar{a}}[f(W)]$ for any $W \subseteq (D^\prime)^\nit{en}$; and (b) The tuples from $D$ that are not transformed do not have any effect in the truth value of $\mc{Q}$, that is, $\mc{Q}_{\bar{a}}[W] = \mc{Q}_{\bar{a}}[W \smallsetminus f(D)]$, for every $W \subseteq D$.}

Moreover, if a SJF-BCQ has a ground atom (atoms without variables, only constants), we can iteratively delete its constants, as above, and therefore, completely remove them all. \
{\em We will assume, from now on, that SJF-BCQs do not have constants, in particular, that they do not have ground atoms.}

 Example \ref{ex:theo-cr} suggests that, for a SJF-BCQ, the order of the variables and the number of times a variable occurs in an atom, are irrelevant for the score alignment. \
Since each predicate appears in only one query atom, a reorder of its arguments (variables and/or constants) only reorders the corresponding columns of the corresponding relation in a given instance.
\ Similarly, if a variable occurs more than once in a query atom, a tuple with different values in those positions is a dummy tuple. Then, since we assume that instances have none of them (see Remark \ref{rem:dummy}), every tuple participating in the query satisfaction has the same constant in those positions, allowing us to keep it only once.

Consider the function $f$ that rearranges or removes those positions in the database tuple. For both cases, $f$ is a bijection that preserves the endogenous/exogenous tuples' classification, and satisfies the condition of Proposition \ref{prop:same-values-lineage}; and therefore, any pair of the scores {\sf CES}, {\sf Resp} or {\sf Shapley} are aligned for $(\mc{Q},D)$ iff they are aligned for $(\mc{Q}^\prime,D^\prime)$.
\ This is  illustrated with the following example.

\

\begin{example} \label{ex:order}
    The query $\mc{Q}\!: \exists x \exists y \, (R(y,x,x) \land S(x))$ gives rise to the new query $\mc{Q}^\prime$, with the third position of atom $R(y,x,x)$ removed; and the remaining two variables, swapped: \  $\mc{Q}^\prime\!: \exists x \exists y \, (R^\prime(x,y) \land S^\prime(x))$ (the same of Example \ref{ex:constants}).

    \begin{table}[h]
    \centering
    {\scriptsize
    $\begin{array}{l|c|c|c|}
        \hline
        R & ~X~ & ~Y~ & ~Z~\\
        \hline
        \underline{\tau_{1}} & a & a & a \\
        \tau_{2} & b & a & a\\
        \tau_{3} & a & b & b\\
        \cline{2-4}
    \end{array}$ \qquad $\begin{array}{l|c|}
        \hline
        S & ~X~ \\
        \hline
        \tau_{4} & a \\
        \tau_{5} & b \\
        \cline{2-2}
    \end{array}$
    \qquad $\stackrel{f}{\longmapsto}$ \qquad
    $\begin{array}{l|c|c|}
        \hline
        R^\prime & ~X~ & ~Y~\\
        \hline
        \underline{\tau_{1}^\prime} & a & a \\
        \tau_{2}^\prime & a & b \\
        \tau_{3}^\prime & b & a \\
        \cline{2-3}
    \end{array}$ \qquad $\begin{array}{l|c|}
        \hline
        S^\prime & ~X~ \\
        \hline
        \tau_{4}^\prime & a \\
        \tau_{5}^\prime & b \\
        \cline{2-2}
    \end{array}$
    }
    \vspace{2mm}\caption{\hspace*{0.5cm} \qquad (a) \ \ Instance $D$. \hspace{1cm} (b) \ \ Instance $D^\prime = f(D)$.}
    \label{tab:order}
\end{table}

    Let $D$ be the instance  in Table \ref{tab:order}(a), with $D^\nit{ex} = \{\tau_1\}$; and
    the function $f$, such that $f(R(u,v,v)) := R^\prime(v,u)$ and $f(S(u)) := S^\prime(u)$. \  $D^\prime$ is the collection of tuples $f(\tau)$, for $\tau \in D$, with $(D^\prime)^\nit{ex} = \{f(\tau_1)\}$. \ This function does two things at once: (i) It removes the repeated occurrence of $x$ (the third position of $R$); and (ii) It reorders the remaining positions, so that $x$ comes first. \ $f$ is a bijection from $D$ to $D^\prime$ that preserves the exogenous tuples.
    \ Notice that in Table \ref{tab:order}(b): $f(\tau_i) = \tau_i^\prime$. For example,  $f(\tau_2) = f(R(b,a,a)) = R^\prime(a,b) = \tau_2^\prime$.

    Furthermore, $f$ satisfies the condition of Proposition \ref{prop:same-values-lineage}, which can be verified through the query lineages: The ex-reduced lineage of $\mc{Q}$ on $D$ is $X_{\tau_4} \lor (X_{\tau_2} \land X_{\tau_4}) \lor (X_{\tau_3} \land X_{\tau_5})$; that is, the one of $\mc{Q}^\prime$ on $D^\prime$ after replacing each $X_{f(\tau_i)}$ by $X_{\tau_i}$. \
    Then, by Proposition \ref{prop:same-values-lineage}, the scores of each $\tau_i$ in $D$ are those of $\tau_i^\prime$ in Table \ref{tab:constants}(c); and any pair of scores are aligned for $(\mc{Q},D)$ iff they are aligned for $(\mc{Q}^\prime,D^\prime)$. \boxtheorem
\end{example}

Next, on the basis of all the above, we make some assumptions on queries. These are only imposed to simplify the proofs in the next sections; but they are not imposed as hypothesis in Theorem \ref{theo:ces_resp_aligned}.
\ \emph{We will assume, from now on, that, for SJF-BCQs, each variable in a atom occurs at most once; and, whenever two atoms share a set of variables, that they appear in the same (relative) order.}

\subsection{{\sf CES} and {\sf Resp} Alignment - The Aligned Case} \label{sub:ces-resp-aligned} \vspace{2mm}

The proofs in this section have, in general, the following structure. For a given syntactic class of SJF-BCQs: \vspace{-1mm}
\begin{enumerate}[1.]
    \item State which score is shown to be aligned with the other. For example, {\sf Resp} is aligned with {\sf CES} (Definition \ref{def:rank_alignment}(b)) if whenever $\CE(\tau) \leq \CE(\tau^\prime)$, then $\Resp(\tau) \leq \Resp(\tau^\prime)$ also holds.\
    This is equivalent to proving the contrapositive, that is, whenever $\Resp(\tau) > \Resp(\tau^\prime)$, then $ \CE(\tau) > \CE(\tau^\prime)$ also holds. This is what we will usually do.
    \item Identify the general structure of the lineage of a query in the class over an arbitrary instance {\em without exogenous tuples}.
    \item Identify the different ways  exogenous tuples may appear  in the lineage, obtaining the ex-reduced lineage.
    \item Using this lineage, compute the scores of endogenous tuples, and establich the alignment stated in Item 1. For example, compute {\sf Resp} for every endogenous tuple, identify the pairs $\tau, \tau^\prime$ with $\Resp(\tau) > \Resp(\tau^\prime)$, and show that $\CE(\tau) > \CE(\tau^\prime)$ for them.
\end{enumerate}

The following proposition  corresponds to Theorem \ref{theo:ces_resp_aligned}(\ref{theo:ces_resp_aligned-a})(i).

\begin{proposition}  \label{prop:cr_aligned_equal_vars} \em
    Let $\mc{Q}$ be a single-component SJF-BCQ where all atoms have the same variables. Then, {\sf CES} and {\sf Resp} are aligned for $(\mc{Q},D)$, for every instance $D$, with or without exogenous tuples.
\end{proposition}

\noindent {\bf Proof: \ } We show that all endogenous tuples have the same responsibility, which, by Remark \ref{rem:same-value}, implies that the scores are aligned.\
Since $\mc{Q}$ is single-component and every atom contains the same variables, it has the  form \
$\mc{Q}: \exists \bar{v} \, (R_1(\bar{v}) \land \cdots \land R_n(\bar{v}))$.\
Then, each tuple in $R_i^D$ (the extension of predicate $R_i$ in $D$) determines a sequence $\bar{c}$ of constants for $\bar{v}$. Denote with $R_j(\bar{c})$ the tuple of $R_j^D$ with those constants, and with $\mc{C}$ the set of those sequences.

The lineage of $\mc{Q}$ in $D$ takes the form:
\begin{equation}
    \mc{L}(\mc{Q},D) = \bigvee_{\bar{c} \in \mc{C}} \Phi_{\bar{c}},
    \qquad \text{with} \quad
    \Phi_{\bar{c}} := \left(X_{R_1(\bar{c})} \land \cdots \land  X_{R_n(\bar{c})}\right).  \label{eq:cr_lineage_same_vars}
\end{equation}

Each $\bar{c} \in \mc{C}$ defines the disjunct $(X_{R_1(\bar{c})} \land \cdots \land X_{R_n(\bar{c})})$, which we call a \emph{block}, denoted $\Phi_{\bar{c}}$. Different blocks do not share $X$-variables.\
A block without exogenous tuples has the form in (\ref{eq:cr_lineage_same_vars}), as determined by  blocks, that is, no variables are set to true, and therefore, removed (see S1ction \ref{sec:lin}). \
Moreover, no block has all its tuples exogenous, since otherwise $D^\nit{ex} \models \mc{Q}$ (see Remark \ref{rem:sat}).\
Accordingly, the ex-reduced lineage is obtained by removing from each block its exogenous variables:
\[
    \mc{L}(\mc{Q},D) \ = \ \bigvee_{\bar{c} \in \mc{C}} \bigwedge_{\tau \in T_{\bar{c}}} X_\tau,
\]
where $T_{\bar{c}}$ is the set of endogenous tuples of block $\Phi_{\bar{c}}$. It holds: $T_{\bar{c}} \neq \emptyset$.

Now, let $\tau \in \Den$ be in block $\Phi_{\bar{c}}$. In relation to {\sf Resp}, a minimum contingency set $\Gamma$ for $\tau$ must make all the other blocks false, which requires removing one tuple from each of them. No tuple of block $\Phi_{\bar{c}}$ needs to be removed due to $\Gamma$ (otherwise, removing $\Gamma$ alone, without $\tau$, would falsify the query). \
Therefore, for every $\tau \in \Den$, $\Resp(D,\mc{Q},\tau) = 1/|\mc{C}|$, a constant. \ By Remark \ref{rem:same-value}, the scores are aligned for $(\mc{Q},D)$.\boxtheorem

Now we prove Theorem \ref{theo:ces_resp_aligned}(\ref{theo:ces_resp_aligned-a})(ii), as the following proposition.

\begin{proposition} \label{prop:cr_aligned_coin_atom_2} \em
    For $\mc{Q}$ a single-component SJF-BCQ with two atoms, {\sf CES} and {\sf Resp} are aligned for $(\mc{Q},D)$, for every instance $D$, with or without exogenous tuples.
\end{proposition}

\noindent {\bf Proof: \ } We show that {\sf Resp} is aligned with {\sf CES} for $(\mc{Q},D)$ (Definition \ref{def:rank_alignment}(b)), that is, for $\tau, \tau^\prime \in \Den$:
\begin{equation} \label{eq:goal_two_atoms}
    \text{If }\ \Resp(\tau) > \Resp(\tau^\prime),\ \text{ then }\ \CE(\tau) > \CE(\tau^\prime).
\end{equation}

Since $\mc{Q}$ is single-component, its two atoms share at least one variable. Then, since $\mc{Q}$  has the form:\footnote{See Section \ref{sub:constants} for the assumptions on queries, in particular, they do not have constants.} \
$\mc{Q}: \exists \bar{x} \exists \bar{y} \exists \bar{z} \, (R(\bar{x},\bar{y}) \land S(\bar{x},\bar{z}))$,
with $\bar{x} \neq \emptyset$. Both $\bar{y}$ and $\bar{z}$ may be empty.

Let $\mc{C}$ be the set of combinations of constants for $\bar{x}$ in $D$ (as in the proof of Proposition \ref{prop:cr_aligned_equal_vars}). For $\bar{c} \in \mc{C}$, denote by $R(\bar{c},\bar{a}_1),\ldots,R(\bar{c},\bar{a}_{m_{\bar{c}}})$ the tuples in $R^D$, with constants $\bar{c}$ for $\bar{x}$, and by $S(\bar{c},\bar{b}_1),\ldots,S(\bar{c},\bar{b}_{n_{\bar{c}}})$ the tuples in $S^D$, with constants $\bar{c}$ for $\bar{x}$. Since $D$ has no dummy tuples (see Remark~\ref{rem:dummy}), every $\bar{c} \in \mc{C}$ occurs in both $R^D$ and $S^D$: A tuple $R(\bar{c},\bar{a})$ with no tuple $S(\bar{c},\bar{b})$ in $D$ can never make $\mc{Q}$ true, so it would be dummy; and similarly for $S$. Hence, $m_{\bar{c}}, n_{\bar{c}} \geq 1$ for every $\bar{c} \in \mc{C}$. \ Accordingly, the lineage of $\mc{Q}$ on $D$ takes the form:
\begin{equation}\mc{L}(\mc{Q},D) \ = \ \bigvee_{\bar{c} \in \mc{C}} \Phi_{\bar{c}},
    \quad \text{with} \quad \
    \Phi_{\bar{c}} \ := \ \left(\bigvee_{i=1}^{m_{\bar{c}}} X_{R(\bar{c},\bar{a}_i)}\right) \land \left(\bigvee_{j=1}^{n_{\bar{c}}} X_{S(\bar{c},\bar{b}_j)}\right)\!. \label{eq:lin2}
\end{equation}

Each $\bar{c} \in \mc{C}$ defines the disjunct $\Phi_{\bar{c}}$, which we call a \emph{block}. Different blocks do not share $X$-variables. We say that a tuple $\tau$  \emph{belongs} to block $\Phi_{\bar{c}}$ when its variable $X_\tau$ occurs in it, i.e., when it has constants $\bar{c}$ for $\bar{x}$.

The ex-reduced lineage is obtained from (\ref{eq:lin2}) by setting to \nit{true} the variables of the exogenous tuples.\
Moreover, a block with no exogenous tuples is not affected, reobtaining $\Phi_{\bar{c}}$ in (\ref{eq:lin2}). For the other blocks, no block has  exogenous tuples $R^D$ and $S^D$ simultaneously, since otherwise $D^\nit{ex} \models \mc{Q}$ \ (see Remark \ref{rem:sat}).

If a block has an exogenous tuple in $R^D$, all its tuples in $R^D$ are exogenous. In fact, assume such a  tuple appears  in the disjunction $\left( X_{R(\bar{c},\bar{a}_{1})} \lor \cdots \lor X_{R(\bar{c},\bar{a}_{m_{\bar{c}}})} \right)$. This would imply that any endogenous tuples in it does not produce a change in the value of the query, that is, it is a dummy tuple. By Remark \ref{rem:dummy}, we assume $D$ has none.\
Then, given that the disjunction of $X_{R(\bar{c},\bar{a}_{i})}$-variables has all of its tuples exogenous, it reduces to the disjunction of its tuples in $S^D$; symmetrically for $S^D$.
It follows that the ex-reduced lineage is:
\begin{equation} \label{eq:lineage_two_atoms}
    \mc{L}(\mc{Q},D) \ = \ \bigvee_{\bar{c} \in \mc{C}^\prime} \Phi_{\bar{c}} \ \lor \ \bigvee_{\tau \in U} X_\tau,
\end{equation}
where $\mc{C}^\prime \subseteq \mc{C}$ is the set of $\bar{c}$ whose block has no exogenous tuples, and $U$ is the set of endogenous tuples belonging  the blocks where all the tuples in $R^D$, or all the tuples in $S^D$, are exogenous. Let $p:= |U|$.

The disjuncts of (\ref{eq:lineage_two_atoms}) share no variables.\
When $\bar{c}$ is clear from the context, we write $m$ and $n$ for $m_{\bar{c}}$ and $n_{\bar{c}}$, and $s := \min(m,n)$.

We now proceed to compute {\sf CES} and {\sf Resp} for  endogenous tuples, using the  lineage in (\ref{eq:lineage_two_atoms}); and next, to use those values to verify (\ref{eq:goal_two_atoms}).

Since $\mc{L}(\mc{Q},D)$ is monotone, and its disjuncts (blocks) share no variables, the scores of a tuple $\tau$ are obtained from: (A) the disjunct where $X_\tau$ occurs (its block, or $X_\tau$ itself if $\tau \in U$); and (B) the other blocks.

\ignore{
\comlb{The ``its/it" below refers to the lineage, right?}
\comfa{In the paragraph above we say that (A) and (B) are disjuncts of the lineage: (A) is the disjunct of $X_\tau$, and (B) is the rest.
When saying ``changes \underline{its} value'': the value of the disjunct of $X_{\tau}$.
When saying ``that \underline{it} is false'': the ``\underline{it}'' means (B), which is the disjunction of all the other blocks.\\
\textbf{I left in \green{green} below the \textit{explicit} alternative, but I think ``its'' and ``it'' should work, no?}}
}

For {\sf CES}, we compute the probability that $X_\tau$ changes the value of the lineage when each $X$-variable is independently true with probability $1/2$ (as for the definition of {\sf CES}).
\ This translates into the following: \ For (A), we compute the probability that $X_\tau$ changes the value of its block (disjunct), denoted by $P_{\text{A}}$. For (B), we compute the probability that the disjunction of the other blocks is false, denoted by $P_{\text{B}}$. It holds:  $\CE(\tau) = P_{\text{A}} \times P_{\text{B}}$. \

We illustrate the computation of this product for $\tau = R(\bar{c},\bar{a}_i)$. For (A), $X_\tau$ changes the value of $\Phi_{\bar{c}}$ iff the variables of the other $m-1$ tuples in $R^D$ of the same block are false, and the $X$-variable of at least one tuple in $S^D$ of the same block is true. This holds in $2^n - 1$ out of the $2^{m+n-1}$ valuations of the other variables of $\Phi_{\bar{c}}$, that is: $P_{\text{A}} = (2^n - 1)/2^{m+n-1}$.

Moreover, denote with $P_0(\psi)$ the probability that a propositional formula $\psi$ is false; and $P_0 := P_0(\mc{L}(\mc{Q},D))$. Since no pair of blocks  share variables, $\mc{L}(\mc{Q},D)$ is false when all of its disjuncts are false. So,  $P_0 = \prod_{\bar{c} \in \mc{C}^\prime} P_0(\Phi_{\bar{c}})\times 2 ^ {-p}$.

(B) is false when all of the other blocks are false, that is, its probability is $P_0$ divided by the probability that the disjunct where $X_\tau$ occurs is false: $P_{\text{B}} = P_0/P_0(\Phi_{\bar{c}})$.\
Then, the {\sf CES} of $\tau$ is:
\[
    \CE(\tau) \ = \ P_{\text{A}} \times P_{\text{B}} \ = \ \frac{(2^n - 1)}{2^{m+n-1}} \times \frac{P_0}{P_0(\Phi_{\bar{c}})} \ = \  \frac{2\times (2^m - 1)\times P_0}{2^m + 2^n - 1}.
\]
Table \ref{tab:scores_two_atoms}(a) summarizes the cases when $\tau = S(\bar{c},\bar{a}_{j})$ for some block $\Phi_{\bar{c}}$, or $\tau \in U$. Their computations are analogous.

For {\sf Resp}, we compute the size of a minimum contingency set for $\tau$, i.e., the minimum number of tuples to remove (i.e. falsify their $X$-variables), so that $X_\tau$ changes the value of the lineage.

This translates into the following: for (A), we compute the minimum number of tuples to remove for $X_\tau$ to change the value of its disjunct, denoted with $N_{\text{A}}$; and for (B), the minimum number of tuples to remove to make all the other disjuncts false, denoted with $N_{\text{B}}$. \ It holds: \  $\Resp(\tau) = 1/(1 + N_{\text{A}} + N_{\text{B}})$.
\ As we did above for for {\sf CES}, we illustrate the computation of the denominator for $\tau = R(\bar{c},\bar{a}_i)$.

For (A), $X_\tau$ changes the value of $\Phi_{\bar{c}}$ iff the variables of the other $m-1$ tuples in $R^D$ of the block are false, and the variable of at least one tuple in $S^D$ of the block is true. So, a minimum contingency set removes those $m-1$ tuples, and no tuple in $S^D$, that is: $N_{\text{A}} = m - 1$. \

Let us denote with $\kappa$ the minimum number of endogenous tuples whose removal makes $\mc{Q}$ false: Each $X_\tau$, $\tau \in U$, requires removing $\tau$, and each block requires removing all its tuples in $R^D$ or all its tuples in $S^D$. So, $\kappa = p + \sum_{\bar{c} \in \mc{C}^\prime} \min(m_{\bar{c}},n_{\bar{c}})$.

For (B), all the other blocks must be false, which requires removing $\kappa$ minus the number of tuples needed to make the block of $\tau$ false, that is: \ $N_{\text{B}} = \kappa - s$.
\ Putting all together:
\[
    \Resp(\tau) \ = \ \frac{1}{1 + N_{\text{A}} + N_{\text{B}}} \ = \ \frac{1}{1 + (m - 1) + (\kappa - s)} \ = \ \frac{1}{\kappa + m - s}.
\]
Table \ref{tab:scores_two_atoms}(b) summarizes the cases when $\tau = S(\bar{c},\bar{b}_j)$ for some block $\Phi_{\bar{c}}$, or $\tau \in U$.

\begin{table}
    \centering
    {\scriptsize \setlength{\arraycolsep}{4pt}
    $\begin{array}{>{\rule[-3.2mm]{0pt}{8.5mm}}l|c|c|c|}
        \hline
        \multicolumn{1}{l|}{\rule[-1.3mm]{0pt}{4.5mm}\tau} & ~~P_{\text{A}}~~ & ~~P_{\text{B}}~~ & {\sf CES} \\
        \hline
        R(\bar{c},\bar{a}_i) & \dfrac{2^n - 1}{2^{m+n-1}} & \dfrac{P_0}{P_0(\Phi_{\bar{c}})} & \dfrac{2(2^n - 1)P_0}{2^m + 2^n - 1} \\
        S(\bar{c},\bar{b}_j) & \dfrac{2^m - 1}{2^{m+n-1}} & \dfrac{P_0}{P_0(\Phi_{\bar{c}})} & \dfrac{2(2^m - 1)P_0}{2^m + 2^n - 1} \\
        \tau \in U & 1 & 2P_0 & 2P_0 \\
        \cline{2-4}
    \end{array}$~~~~~~~~~~~
    $\begin{array}{>{\rule[-3.2mm]{0pt}{8.5mm}}l|c|c|c|}
        \hline
        \multicolumn{1}{l|}{\rule[-1.3mm]{0pt}{4.5mm}\tau} & ~~N_\text{A}~~ & ~~N_\text{B}~~ & {\sf Resp} \\
        \hline
        R(\bar{c},\bar{a}_i) & m - 1 & \kappa - s & \dfrac{1}{\kappa + m - s} \\
        S(\bar{c},\bar{b}_j) & n - 1 & \kappa - s & \dfrac{1}{\kappa + n - s} \\
        \tau \in U & 0 & \kappa - 1 & \dfrac{1}{\kappa} \\
        \cline{2-4}
    \end{array}$

    }
    \vspace{2mm}\caption{ \hspace*{1cm} (a) \ \ {\sf CES} $= P_A \times P_B$. \hspace{1cm} (b) \ \ {\sf Resp} $= 1/(1 + N_\text{A} + N_\text{B})$.}
    \label{tab:scores_two_atoms}
\end{table}

We now proceed to compare the scores. For that, we assume, w.l.o.g., that $m_{\bar{c}} \leq n_{\bar{c}}$,  for every block $\bar{c} \in \mc{C}^\prime$ (since the scores are symmetric for tuples in $R^D$ and $S^D$).\footnote{If, $m_{\bar{c}} > n_{\bar{c}}$ for a particular block, we reverse the roles of the tuples in $R^D$ and $S^D$.} \ Then, $s = m$.

From Table \ref{tab:scores_two_atoms}(b), $\Resp(\tau) = 1/(\kappa + d_\tau)$, where $d_\tau := n_{\bar{c}} - m_{\bar{c}}$ if $\tau$ is a tuple in $S^D$ of a block $\bar{c} \in \mc{C}^\prime$; and $d_\tau := 0$, otherwise (i.e., if $\tau \in U$, or $\tau$ is a tuple in $R^D$). Therefore, $\Resp(\tau) > \Resp(\tau^\prime)$ iff $d_\tau < d_{\tau^\prime}$, and (\ref{eq:goal_two_atoms}) becomes:
\begin{equation} \label{eq:goal_d_two_atoms}
    \text{If }\ d_\tau < d_{\tau^\prime},\ \text{ then }\ \CE(\tau) > \CE(\tau^\prime).
\end{equation}

To verify that this holds, consider tuples $\tau,\tau^\prime$, with $d_\tau < d_{\tau^\prime}$. Let $m, n$ and $m^\prime, n^\prime$ be the number of tuples from relations $R^D$ and $S^D$ for the blocks of $\tau$ and $\tau^\prime$, respectively. Since $d_\tau \geq 0$, and by (\ref{eq:goal_d_two_atoms}), we obtain: $d_{\tau^\prime} \geq 1$. So, $\tau^\prime$ is a tuple in $S^D$ of a block with $d_{\tau^\prime} = n^\prime - m^\prime > 0$, and $\tau$ is either: (i) a tuple in $U$, (ii) a tuple in $R^D$, or (iii) a tuple in $S^D$ of a block such that $d_{\tau} < d_{\tau^\prime}$.

From Table \ref{tab:scores_two_atoms}(a), we observe the following inequality, which is used to prove (\ref{eq:goal_d_two_atoms}). Here, $\tau$ is for any of the three cases in the previous paragraph:
\[
    \frac{\CE(\tau^\prime)}{2P_0} \ = \ \frac{2^{m^\prime} - 1}{2^{m^\prime} + 2^{n^\prime} - 1} \ < \ \frac{2^{m^\prime}}{2^{m^\prime} + 2^{n^\prime}} \ = \ \frac{1}{1 + 2^{d_{\tau^\prime}}}.
\]
We now check each of the three cases:
\begin{enumerate}[(i)]
    \item If $\tau \in U$, then $\CE(\tau) = 2P_0 > \CE(\tau^\prime)$.
    \item If $\tau \in R^D$, then, since $1 \leq m \leq n$ and $d_{\tau^\prime} \geq 1$:
    \[
        \frac{\CE(\tau)}{2P_0} = \frac{2^n - 1}{2^m + 2^n - 1} \ \geq \ \frac{2^m - 1}{2^{m+1} - 1} \ \geq \ \frac{1}{3} \ \geq \ \frac{1}{1 + 2^{d_{\tau^\prime}}} \ > \ \frac{\CE(\tau^\prime)}{2P_0}.
    \]
    \item If $\tau$ is in $S^D$, then $d_\tau = n - m$. Since $m \geq 1$ and $d_\tau + 1 \leq d_{\tau^\prime}$:
    \[
        \frac{\CE(\tau)}{2P_0} = \frac{2^m - 1}{2^m + 2^n - 1} \ \geq \ \frac{1}{1 + 2^{d_\tau + 1}} \ \geq \ \frac{1}{1 + 2^{d_{\tau^\prime}}} \ > \ \frac{\CE(\tau^\prime)}{2P_0}.
    \]
\end{enumerate}
We have established that (\ref{eq:goal_d_two_atoms}), and therefore also (\ref{eq:goal_two_atoms}), holds. We obtain that {\sf CES} and {\sf Resp} are aligned for $(\mc{Q},D)$.
\boxtheorem

Now we prove Theorem \ref{theo:ces_resp_aligned}(\ref{theo:ces_resp_aligned-a})(iii), as the following proposition.

\begin{proposition}  \label{prop:cr_aligned_multi} \em
    Let $\mc{Q}$ be a SJF-BCQ with components $C_1,\ldots,C_n$, $n \geq 2$, and associated queries $\mc{Q}_1,\ldots,\mc{Q}_n$, such that, for  $i = 1,\ldots,n$, $|\nit{Atoms}(\mc{Q}_i)| = 1$. Then, {\sf CES} and {\sf Resp} are aligned for $(\mc{Q},D)$, for every instance $D$, with or without exogenous tuples.
\end{proposition}

\noindent {\bf Proof: \ } We show that {\sf Resp} is aligned with {\sf CES} for $(\mc{Q},D)$ (Definition \ref{def:rank_alignment}(b)), that is, for $\tau, \tau^\prime \in \Den$:
\begin{equation} \label{eq:goal_single_atoms}
    \text{If }\ \Resp(\tau) > \Resp(\tau^\prime),\ \text{ then }\ \CE(\tau) > \CE(\tau^\prime).
\end{equation}

Since each component of $\mc{Q}$ has a single atom, and the atoms are in different components, $\mc{Q}$ has the form: \
$\mc{Q}\!: \exists \bar{v}_1 \cdots \exists \bar{v}_n \, (R_1(\bar{v}_1) \land \cdots \land R_n(\bar{v}_n))$,
with $\bar{v}_1,\ldots,\bar{v}_n$ pairwise disjoint.

For $j = 1,\ldots,n$, let $R_j(\bar{a}_{j,1}),\ldots,R_j(\bar{a}_{j,m_j})$ be the tuples in $R_j^D$, where each $\bar{a}_{j,i}$ is a combination of constants for $\bar{v}_j$. Since $D \models \mc{Q}$, $m_j \geq 1$ for every $j$.\
Then, the lineage of $\mc{Q}$ on $D$ is the disjunction, over all the combinations of one tuple from each of $R_1^D,\ldots,R_n^D$, of the conjunction of their $X$-variables. Since the atoms share no variables, every such combination makes $\mc{Q}$ true. So, with $[m] := \{1,\ldots,m\}$:
\begin{align}
    \mc{L}(\mc{Q},D) \ &= \bigvee_{(i_1,\ldots,i_n) \,\in\, [m_1] \times \cdots \times [m_n]} \big( X_{R_1(\bar{a}_{1,i_1})} \land \cdots \land X_{R_n(\bar{a}_{n,i_n})} \big) \nonumber \\
    &\equiv \ \bigwedge_{j=1}^{n} \varphi_j, \qquad \text{with} \quad \varphi_j := \bigvee_{i=1}^{m_j} X_{R_j(\bar{a}_{j,i})}, \label{eq:lin_single}
\end{align}
where the equivalence follows from the factorization of the disjunctions.\
Each conjunct $\varphi_j$ contains the $X$-variables of the tuples in $R_j^D$. Distinct conjuncts share no $X$-variables.

The ex-reduced lineage is obtained from (\ref{eq:lin_single}) by setting to \nit{true} the $X$-variables for exogenous tuples.\
Moreover, a conjunct $\varphi_j$ with no exogenous tuples is not affected.

For the others, if $R_j^D$ has an exogenous tuple, its $X$-variable is set to \nit{true}; and then, its conjunct $\varphi_j = \left( X_{R_j(\bar{a}_{j,1})} \lor \cdots \lor X_{R_j(\bar{a}_{j,m_j})} \right)$ would also be true.
\ That is, an endogenous tuple inside $\varphi_j$ in $R_j^D$ could never change the value of the lineage. Then, it would be a dummy tuple, which, by Remark \ref{rem:dummy}, we assume $D$ has none. \
Given this, if $R_j^D$ has an exogenous tuple, all of its tuples are exogenous. Consequently, $\phi_j$ is equivalent to \nit{true}, and can be removed from the conjunction.

Furthermore, not all of $R_1^D,\ldots,R_n^D$ have exogenous tuples, since, otherwise, $D^\nit{ex} \models \mc{Q}$ (see Remark \ref{rem:sat}). \
W.l.o.g., let $R_1,\ldots,R_l$, with $1 \leq l \leq n$, be the predicates whose extensions in $D$ have no exogenous tuples. Then, the ex-reduced lineage is:
\begin{equation} \label{eq:lineage_single_atoms}
    \mc{L}(\mc{Q},D) \ = \ \bigwedge_{j=1}^{l} \varphi_j,
\end{equation}
and $\Den$ consists of the tuples in $R_1^D,\ldots,R_l^D$.

We now compute {\sf Resp} and {\sf CES} for every endogenous tuple For that, we use (\ref{eq:lineage_single_atoms}), and then use those values to verify (\ref{eq:goal_single_atoms}).

Let $\tau = R_j(\bar{a}_{j,i})$, with $j \leq l$. \
For {\sf Resp}, a minimum contingency set for $\tau$ must keep all the other conjuncts true, which requires removing no tuple, and make $X_\tau$ change the value of $\varphi_j$, which requires removing the other $m_j - 1$ tuples in $R_j^D$. \ Then, $\Resp(\tau) = 1/m_j$.

For {\sf CES}, denote with $P_1(\psi)$ the probability that a propositional formula $\psi$ is true when each $X$-variable is independently true with probability $1/2$ (as assumed for the definition of {\sf CES}); and $P_1 := P_1(\mc{L}(\mc{Q},D))$.\
Then, $P_1(\varphi_j) = 1 - 1/2^{m_j}$. Given that the conjuncts in (\ref{eq:lineage_single_atoms}) share no variables, $P_1 = \prod_{j=1}^{l} P_1(\varphi_j)$.

Since $\mc{L}(\mc{Q},D)$ is monotone, $\CE(\tau)$ is the probability that $X_\tau$ changes the value of $\mc{L}(\mc{Q},D)$. This happens iff $X_\tau$ changes the value of $\varphi_j$, i.e., the $X$-variables of the other $m_j - 1$ tuples in $R_j^D$ are false (with probability $1/2^{m_j - 1}$); and all the other conjuncts are true (with probability $P_1 / P_1(\varphi_j)$). \ Then:
\[
    \CE(\tau) \ = \ \frac{1}{2^{m_j - 1}} \times \frac{P_1}{P_1(\varphi_j)}.
\]
Now, let $\tau$ be a tuple in $R_j^D$, and $\tau^\prime$ a tuple in $R_{j^\prime}^D$, with $j, j^\prime \leq l$, such that $\Resp(\tau) > \Resp(\tau^\prime)$. \ Then, $1/m_j > 1/m_{j^\prime}$, i.e., $m_j < m_{j^\prime}$. \ It follows that
\[
    \frac{1}{2^{m_j - 1}} > \frac{1}{2^{m_{j^\prime} - 1}},
    \quad \text{and} \quad
    \frac{1}{P_1(\varphi_j)} \ = \ \frac{1}{1 - 2^{-m_j}} \ > \ \frac{1}{1 - 2^{-m_{j^\prime}}} = \frac{1}{P_1(\varphi_{j^\prime})}.
\]
Since $P_1 > 0$ is common to both $\CE(\tau)$ and $\CE(\tau^\prime)$, we obtain:
\[
    \CE(\tau) \ = \ P_1 \times \frac{1}{2^{m_j - 1}} \times \frac{1}{P_1(\varphi_j)} \ > \ P_1 \times \frac{1}{2^{m_{j^\prime}- 1}} \times \frac{1}{P_1(\varphi_{j^\prime})} \ = \ \CE(\tau^\prime).
\]
That is, (\ref{eq:goal_single_atoms}) holds; and, then,  {\sf CES} and {\sf Resp} are aligned for $(\mc{Q},D)$.
\boxtheorem

\subsection{{\sf CES} and {\sf Resp} Alignment - The Non-Aligned Case} \label{sub:ces-resp-non-aligned} \vspace{2mm}

In order to prove Theorem \ref{theo:ces_resp_aligned}(\ref{theo:ces_resp_aligned-b}), we have to provide a \emph{counterexample} for every query $\mc{Q}$ not satisfying Theorem \ref{theo:ces_resp_aligned}(\ref{theo:ces_resp_aligned-a}). In this case, a counterexample is an instance $D$ for which the {\sf CES} and {\sf Resp} are not aligned for $(\mc{Q},D)$.

Those counterexamples are built in a systematic way: \ We start with a  ``simplest'' pair formed by a  query and an instance for which the scores are not aligned. \ Next, for any query belonging to the same class, a new instance is built based on that pair. \
Remark \ref{rem:non-alignment-technique} below outlines this technique in more precise terms.

\begin{remark} \label{rem:non-alignment-technique}
    Consider the class of BCQs \  $\mc{C} \subseteq\ $SJF-BCQ. We proceed as follows:
    \begin{enumerate}[(a)]
        \item Identify the \emph{simplest} query $\mc{Q}$, or set of queries, in $\mc{C}$, in the sense that, for every instance $D$ (compatible with $\mc{Q}$) and  every query $\mc{Q}^\prime$ in $\mc{C}$, there exists an instance $D^\prime$ (possibly with exogenous tuples), such that $\mc{L}(D,\mc{Q}) \ \equiv \ \mc{L}(D^\prime,\mc{Q}^\prime)$ \ (see Section \ref{sec:lin}).

       \vspace{1mm} Intuitively, $\mc{Q}$ is \emph{contained} in every other query $\mc{Q}^\prime \in \mc{C}$: $\mc{Q}^\prime$ has atoms and variables that play the roles of those of $\mc{Q}$, i.e., they share variables in the same way, possibly together with additional atoms and variables.
        \item Identify a counterexample for these simplest query or queries, i.e. a particular instance $D_0$, for which a non-alignment result holds. This counterexample may contain exogenous tuples.
        \item In order to obtain a counterexample for an arbitrary query $\mc{Q} \in \mc{C}$, build an instance $D$ from $D_0$ by: (i) Copying the tuples from $D_0$ into the atoms of $\mc{Q}$ that play the roles of those of the simplest query (adding a fresh constant for the additional variables if required); and, (ii) For the atoms of $\mc{Q}^\prime$ that do not play any of those roles, add exogenous tuples.

       \vspace{1mm} In this way, every tuple in $D_0$ has a corresponding tuple in $D$, and tuples in $D$ will be endogenous iff the original tuples in $D_0$ are endogenous.
        \item The correspondence in (c) defines a bijection $f$ from the endogenous tuples of $D_0$ to the endogenous tuples of $D$, such that, for every set $W$ of endogenous tuples of $D_0$, the simplest query is true on $W \cup D_0^\nit{ex}$ iff $\mc{Q}$ is true on $f(W) \cup D^\nit{ex}$. \ Since this is the hypothesis of Proposition \ref{prop:same-values-lineage}, we obtain that, for every endogenous  tuple $\tau$ of $D_0$, $\tau$  and $f(\tau)$ have the same {\sf CES} and {\sf Resp} scores. As a consequence, those scores are aligned for $D_0$ iff they are aligned for $D$.\boxtheorem
    \end{enumerate}
\end{remark}

We illustrate Remark \ref{rem:non-alignment-technique} with an example.

\begin{example} \label{ex:simplest_rstu}
    Consider the class $\mc{C}$ of SJF-BCQs containing at least two components, each with at least two atoms, and its simplest query $\mc{Q}\!: \exists x \exists y\, (R(x) \land S(x) \land T(y) \land U(y))$ in this class. \
    Let $D_0$ be the instance in Table \ref{tab:simplest_rstu}(a), with $D_0^{\,\nit{ex}} = \{\tau_4,\tau_5,\tau_6\}$. \ Here:

\begin{table}[h]
    \centering
    \hspace*{2mm}{\scriptsize
    $\begin{tabu}{l|c|}
        \hline
        R & ~X~\\
        \hline
        \tau_{1} & a\\
        \tau_{2} & b\\
        \tau_{3} & c\\
        \cline{2-2}
    \end{tabu}$~~~$\begin{tabu}{l|c|}
        \hline
        S & ~X~\\
        \hline
        \underline{\tau_{4}} & a\\
        \underline{\tau_{5}} & b\\
        \underline{\tau_{6}} & c\\
        \cline{2-2}
    \end{tabu}$~~~$\begin{tabu}{l|c|}
        \hline
        T & ~Y~\\
        \hline
        \tau_{7} & a\\
        \tau_{8} & b\\
        \tau_{9} & c\\
        \tau_{10} & d\\
        \cline{2-2}
    \end{tabu}$~~~$\begin{tabu}{l|c|}
        \hline
        U & ~Y~\\
        \hline
        \tau_{11} & a\\
        \tau_{12} & b\\
        \tau_{13} & c\\
        \tau_{14} & d\\
        \cline{2-2}
    \end{tabu}$

    \par\vspace{4mm}

    $\begin{tabu}{l|c|c|}
        \hline
        R^\prime & ~X~ & ~Z~\\
        \hline
        \tau_{1} & a & c^\prime\\
        \tau_{2} & b & c^\prime\\
        \tau_{3} & c & c^\prime\\
        \cline{2-3}
    \end{tabu}$~~~$\begin{tabu}{l|c|c|}
        \hline
        S^\prime & ~X~ & ~Z~\\
        \hline
        \underline{\tau_{4}} & a & c^\prime\\
        \underline{\tau_{5}} & b & c^\prime\\
        \underline{\tau_{6}} & c & c^\prime\\
        \cline{2-3}
    \end{tabu}$~~~$\begin{tabu}{l|c|c|}
        \hline
        T^\prime & ~Y~ & ~W~\\
        \hline
        \tau_{7} & a & c^\prime\\
        \tau_{8} & b & c^\prime\\
        \tau_{9} & c & c^\prime\\
        \tau_{10} & d & c^\prime\\
        \cline{2-3}
    \end{tabu}$~~~$\begin{tabu}{l|c|c|}
        \hline
        U^\prime & ~Y~ & ~W~\\
        \hline
        \tau_{11} & a & c^\prime\\
        \tau_{12} & b & c^\prime\\
        \tau_{13} & c & c^\prime\\
        \tau_{14} & d & c^\prime\\
        \cline{2-3}
    \end{tabu}$~~~$\begin{tabu}{l|c|}
        \hline
        P & ~Z~\\
        \hline
        \underline{\tau_{15}} & c^\prime\\
        \cline{2-2}
    \end{tabu}$~~~$\begin{tabu}{l|c|}
        \hline
        V & ~W~\\
        \hline
        \underline{\tau_{16}} & c^\prime\\
        \cline{2-2}
    \end{tabu}$
    }
   \vspace{2mm} \caption{\hspace*{1cm} (a) \ Instance $D_0$. \hspace{2cm} (b) \ Instance $D^\prime$ for $\mc{Q}^\prime$.}
    \label{tab:simplest_rstu}
\end{table}

    \begin{eqnarray*}
         \mc{L}(\mc{Q},D_0) \ &=& \ (X_{\tau_1} \lor X_{\tau_2} \lor X_{\tau_3}) \ \land \ ( \ (X_{\tau_7} \land X_{\tau_{11}}) \lor   \label{eq:lin_d0_q} \\
                            && \ \ \  (X_{\tau_8} \land X_{\tau_{12}}) \lor (X_{\tau_9} \land X_{\tau_{13}})\lor (X_{\tau_{10}} \land X_{\tau_{14}}) \ ).
        \end{eqnarray*}
    The {\sf CES} scores of $\tau_1$ and $\tau_7$ are (approx.) $0.1709$ and $0.1846$, respectively, whereas their {\sf Resp} scores are $1/3$ and $1/4$. Thus, $\tau_1 \prec^\CE \! \tau_7$ and $\tau_7 \prec^\nit{Re} \! \tau_1$, so the scores are not aligned for $(\mc{Q},D_0)$.

    Now, consider the query $\mc{Q}^\prime\!: \exists x \exists y \exists z \exists w\,(R^\prime(x,z) \land S^\prime(x,z) \land P(z) \land T^\prime(y,w) \land U^\prime(y,w) \land V(w))$. Clearly $\mc{Q} \in \mc{C}$.\
    In order to build a counterexample for $\mc{Q}^\prime$, predicates $R^\prime$, $S^\prime$, $T^\prime$, and $U^\prime$ play the roles of $R$, $S$, $T$, and $U$ in $\mc{Q}$.

    According to Remark \ref{rem:non-alignment-technique}(c), we construct $D^\prime$, see Table \ref{tab:simplest_rstu}(b), by replacing each $A(a_1)$ in $D_0$, with $A \in \{R,S,T,U\}$, with $A^\prime(a_1,c^\prime)$, with $c^\prime$ a fresh constant, and preserving the endogenous/exogenous classification. We also add exogenous tuples $P(c^\prime)$ and $V(c^\prime)$.

    Before removing exogenous variables, the lineage of $\mc{Q}^\prime$ on $D^\prime$ is
    \begin{equation*}
        \begin{aligned}
            \mc{L}(\mc{Q}^\prime,D^\prime) &\ = \ X_{\tau_{15}} \land X_{\tau_{16}} \land ((X_{\tau_1} \land X_{\tau_4}) \lor (X_{\tau_2} \land X_{\tau_5}) \lor (X_{\tau_3} \land X_{\tau_6})
            )\\
            &\ \ \ \ \ \ \land ((X_{\tau_7} \land X_{\tau_{11}}) \lor (X_{\tau_8} \land X_{\tau_{12}})\lor (X_{\tau_9} \land X_{\tau_{13}}) \lor (X_{\tau_{10}} \land X_{\tau_{14}})).
        \end{aligned}
    \end{equation*}
    The ex-reduced lineage is: \
        $\mc{L}(\mc{Q}^\prime,D^\prime)
        \ = \ \mc{L}(\mc{Q},D_0)$. \
        By Remark \ref{rem:non-alignment-technique}(d) (that appeals to Proposition \ref{prop:same-values-lineage}), it follows that endogenous tuples $\tau_1,\tau_2,\tau_3,\tau_7,\ldots,\tau_{14}$ have the same {\sf CES} and {\sf Resp} scores in $D^\prime$ as in $D_0$. Thus, the scores are not aligned for $(\mc{Q}^\prime,D^\prime)$. \boxtheorem
\end{example}

According to Theorem~\ref{theo:ces_resp_aligned}, a counterexample for a SJF-BCQ $\mc{Q}$ exists iff $\mc{Q}$ satisfies none of the conditions of its part (\ref{theo:ces_resp_aligned-a}).\
If $\mc{Q}$ is single-component, this means that at least two atoms do not share the same variables, and that $\mc{Q}$ has at least three atoms.\
If $\mc{Q}$ is multi-component, this means that some component has at least two atoms.\
Given this, we separate our analysis in two mutually exclusive classes of queries:
\begin{enumerate}[(A)]
\item Single-component queries with at least three atoms, with at least two of them not having the same variables; \ and
\item Multi-component queries with at least two atoms in at least one component.
\end{enumerate}

The next proposition establishes that there exists a counterexample for every SJF-BCQ satisfying (A).

\begin{proposition}  \label{prop:cr-not-aligned-A} \em
    Let $\mc{Q}$ be a single-component SJF-BCQ, such that $|\nit{Atoms}(\mc{Q})| \geq 3$, and at least two atoms do not have the same variables. There is an instance $D$ with exogenous tuples where {\sf CES} and {\sf Resp} are not aligned. \boxtheorem
\end{proposition}

Before going into the proof of Proposition \ref{prop:cr-not-aligned-A}, we describe and explain a general methodology for the construction of the so-called ``simplest" counterexample from the ``simplest" query.  However, for the class of queries in Proposition \ref{prop:cr-not-aligned-A}, instead of a  single query, we need  a set of three different ``simplest" queries, those in (\ref{eq:cr-not-aligned-atoms3}) below. After doing that, we  provide counterexample instances for each of those ``simplest'' queries. Once those three examples are presented, we give the proof of Proposition \ref{prop:cr-not-aligned-A}.

\begin{equation} \label{eq:cr-not-aligned-atoms3}
    \begin{aligned}
        \mc{Q}_{a} \! &: \ \exists x \exists y (R(x,y) \land S(x,y) \land T(x)), \\
        \mc{Q}_{b} \! &: \ \exists x \exists y (R(x,y) \land S(x) \land T(x)), \text{ and} \\
        \mc{Q}_{c} \! &: \ \exists x \exists y (R(x,y) \land S(x) \land T(y)).\footnotemark
    \end{aligned}
\end{equation}
\footnotetext{Notice that $\mc{Q}_{c}$ is equivalent to $\mc{Q}_{\sf RST}$ in Example \ref{ex:comp_ces_resp_non_hierarchical}.}

Before building the counterexamples, consider the following propositional formula:
\begin{equation} \label{eq:lineage-counter-atoms3}
    ((X_1 \lor X_2) \land X_3) \lor (X_4 \land X_5 \land X_6).
\end{equation}

Even without an explicit query or an instance, we can compute scores from this lineage: For example, for the tuples associated to $X_1$ and $X_4$, {\sf CES} is $7/32$ and $5/32$ (assuming tuple independence and $U(\frac{1}{2})$- distribution), respectively. \  For {\sf Resp} we have $1/3$ and $1/2$, respectively.
\ We can see that {\sf CES} ranks higher the tuple associated with $X_1$ than the tuple associated with $X_4$. However, the ranking according to {\sf Resp} for the tuples associated to those variables is the inverse.

It follows that, if  $\mc{L}(\mc{Q},D)$ is logically equivalent to (\ref{eq:lineage-counter-atoms3}), for a given query $\mc{Q}$ and an instance $D$, then the pair $(D,\mc{Q})$ constitutes a counterexample.
\ We highlight that (\ref{eq:lineage-counter-atoms3}) is a minimal counterexample, in the sense that, {\sf CES} and {\sf Resp} would be aligned if any of its variables are removed (set to true or false).

Coming back to queries $\mc{Q}_{a}$, $\mc{Q}_{b}$ and $\mc{Q}_{c}$ in (\ref{eq:cr-not-aligned-atoms3}), notice that, for each of them, there is an instance whose ex-reduced lineage is logically equivalent to (\ref{eq:lineage-counter-atoms3}), up to renaming of the $X$-variables.
\ Next, we provide, as an example, the construction of such an instance $D_{a}$ for $\mc{Q}_{a}$.

Consider the general form of the lineage of $\mc{Q}_{a}$ on an arbitrary instance $D$:
\begin{equation} \label{eq:lineage-qa}
    \mc{L}(\mc{Q}_{a},D) \ = \bigvee_{c \in \nit{ad\!om}(D)}  \Phi_c ~  \text{, \ with } \ \ \Phi_c \ = \ \Big( X_{T(c)} \land \bigvee_{c^\prime \in B_c} \big( X_{R(c,c^\prime)} \land X_{S(c,c^\prime)} \big) \Big),
\end{equation}
with $c^\prime \in B_c$ iff $R(c,c^\prime)$, $S(c,c^\prime)$ and $T(c)$ are in $D$. Each $c$ defines the disjunct $\Phi_c$, that we call  \emph{a block}.
\ignore{Notice that $D$ \red{is not fixed here???}: $c$ and $c^\prime$ are constants from the domain of $D$, and the sets $\mc{C}$ and $B_c$ depend on $D$. (\ref{eq:lineage-qa}) only gives the general form of the lineage of $\mc{Q}_{a}$.\
Moreover, the reverse also holds:} It  holds that every propositional formula of this form is the lineage of $\mc{Q}_{a}$ on some instance, namely the one whose tuples are $T(c)$, $R(c,c^\prime)$ and $S(c,c^\prime)$, for $c \in \mc{C}$ and $c^\prime \in B_c$. This allows us to first choose a formula of this form, and then build the instance.

Assume, for an instance $D$, that $\mc{L}(\mc{Q}_{a},D)$ has two blocks $\Phi_{c_1}$ and $\Phi_{c_2}$, such that: $B_{c_1} = \{c^\prime_1,c^\prime_2\}$, and $B_{c_2} = \{c^\prime_3\}$, with $c_1,c_2,c^\prime_1,c^\prime_2,c^\prime_3 \in \nit{ad\!om}(D)$. \ Then, the lineage becomes:
\begin{equation}\label{eq:lineage-q3a}
    \begin{aligned}
        \mc{L}(\mc{Q}_{a},D) \  = \ \ &\left( X_{T(c_1)} \land \left((X_{R(c_1,c^\prime_1)} \land X_{S(c_1,c^\prime_1)}) \lor \left(X_{R(c_1,c^\prime_2)} \land X_{S(c_1,c^\prime_2)}\right)\right)\right) \\
        & ~~\lor \left(X_{T(c_2)} \land \left(X_{R(c_2,c^\prime_3)} \land X_{S(c_2,c^\prime_3)}\right)\right), \ \ \mbox{or} \\
        \mc{L}(\mc{Q}_{a},D) \ = \ \ &\big( X_3 \land ((X_1 \land X_7) \lor (X_2 \land X_8))\big) \lor (X_4 \land (X_5 \land X_6)),
    \end{aligned}
\end{equation}
where the second line is obtained by numbering the $X$-variables.

Now, if we set variables $X_7$ and $X_8$ to true, by  considering  their associated tuples to be exogenous, (\ref{eq:lineage-q3a}) becomes (\ref{eq:lineage-counter-atoms3}), as we wanted.

It remains to build an instance whose lineage becomes (\ref{eq:lineage-q3a}). \ In this direction, consider three mappings from the variables of $\mc{Q}_{a}$ to constants:
\[
    h_1(x) := a,\ h_1(y) := a; \quad h_2(x) := a,\ h_2(y) := b \text{; \ \ and} \quad  h_3(x) := b,\ h_3(y) := b.
\]
Applying each of them to the atoms of $\mc{Q}_{a}$, we obtain tuples we can use to build and instance, namely:\ignore{obtain ground tuples. We add these ground tuples to an (empty) instance, and obtain:}
\[
    D_{a} = \{R(a,a),\ S(a,a),\ R(a,b),\ S(a,b),\ T(a),\ R(b,b),\ S(b,b),\ T(b)\},
\]
for which the query lineage becomes:
\begin{align*}
    \left(\left( \left(X_{R(a,a)} \land X_{S(a,a)}) \lor (X_{R(a,b)} \land X_{S(a,b)}\right) \right) \land X_{T(a)}\right) \lor \ \left(X_{T(b)} \land X_{R(b,b)} \land X_{S(b,b)}\right).
\end{align*}
It has the form (\ref{eq:lineage-q3a}). Now, by declaring $S(a,a)$ and $S(a,b)$ as exogenous, the ex-reduced lineage of $\mc{Q}_{a}$ on $D_{a}$ becomes:
\[
    \left( \left(X_{R(a,a)} \lor X_{R(a,b)} \right) \land X_{T(a)}\right) \lor \left(X_{T(b)} \land X_{R(b,b)} \land X_{S(b,b)}\right).
\]
That is, of the form (\ref{eq:lineage-counter-atoms3}), with $X_{R(a,a)}$, $X_{R(a,b)}$, $X_{T(a)}$, $X_{T(b)}$, $X_{R(b,b)}$, $X_{S(b,b)}$ instead of $X_1, \ldots, X_6$, respectively.\
Then, $(\mc{Q}_{a},D_{a})$ is a counterexample: {\sf CES} ranks $R(a,a)$ higher than $T(b)$, and {\sf Resp} ranks them in the inverse order.\

The instances for $\mc{Q}_{b}$ and $\mc{Q}_{c}$ are built in the same way.
\ We now show the counterexamples for $\mc{Q}_{a}$, $\mc{Q}_{b}$ and $\mc{Q}_{c}$ obtained using this methodology.

\begin{table}[h]
    \centering
    {\scriptsize
    $\begin{array}{l|c|c|}
        \hline
        R & ~X~ & ~Y~\\
        \hline
        \tau_{1} & a & a\\
        \tau_{2} & a & b\\
        \tau_{3} & b & b\\
        \cline{2-3}
    \end{array}$~~~~~~~~~$\begin{array}{l|c|c|}
        \hline
        S & ~X~ & ~Y~\\
        \hline
        \underline{\tau_{4}} & a & a\\
        \underline{\tau_{5}} & a & b\\
        \tau_{6} & b & b\\
        \cline{2-3}
    \end{array}$~~~~~~~~~$\begin{array}{l|c|}
        \hline
        T & ~X~\\
        \hline
        \tau_{7} & a\\
        \tau_{8} & b\\
        \cline{2-2}
    \end{array}$
    ~~~~~~~~~~~~~
    $\begin{array}{l|c|c|}
        \hline
        \tau & ~~{\sf CES}~~ & ~~{\sf Resp~~}\\
        \hline
        \tau_{1} & ~~0.2188~~ & ~~1/3~~\\
        \tau_{3} & ~~0.1563~~ & ~~1/2~~\\
        \cline{2-3}
    \end{array}$
    }
   \vspace{2mm} \caption{\hspace*{1cm} (a) \ \ Instance $D_{a}$. \hspace{1cm} (b) \ \ Scores on $D_{a}$. }
    \label{tab:q3a}
\end{table}

\begin{example} \label{ex:q3a} \
    Consider the query $\mc{Q}_{a}$ in Equation (\ref{eq:cr-not-aligned-atoms3}), and $D_{a}$ the instance in Table~\ref{tab:q3a}, with $D_{a}^\nit{ex} = \{\tau_4, \tau_5\}$. \ The lineage of $\mc{Q}_{a}$ on $D_{a}$ is $((X_{\tau_1} \lor X_{\tau_2}) \land X_{\tau_7}) \lor (X_{\tau_3} \land X_{\tau_6} \land X_{\tau_8})$ (see  Equation (\ref{eq:lineage-counter-atoms3})). \ It holds: \ $D_{a} \models \mc{Q}_{a}$, \ $\tau_3 \prec^\CE \! \tau_1$ and $\tau_1 \prec^{\nit{Resp}} \! \tau_3$. \ {\sf CES} and {\sf Resp} are not aligned for $(\mc{Q}_{a},D_{a})$. \boxtheorem
\end{example}

\begin{table}[h]
    \centering
    {\scriptsize
    $\begin{array}{l|c|c|}
        \hline
        R & ~X~ & ~Y~\\
        \hline
        \tau_{1} & a & a\\
        \tau_{2} & a & b\\
        \tau_{3} & b & b\\
        \cline{2-3}
    \end{array}$~~~~~~$\begin{array}{l|c|}
        \hline
        S & ~X~\\
        \hline
        \tau_{4} & a\\
        \tau_{5} & b\\
        \cline{2-2}
    \end{array}$~~~~~~$\begin{array}{l|c|}
        \hline
        T & ~X~\\
        \hline
        \underline{\tau_{6}} & a\\
        \tau_{7} & b\\
        \cline{2-2}
    \end{array}$
    ~~~~~~~~~~~~
    $\begin{array}{l|c|c|}
        \hline
        \tau & ~~{\sf CES}~~ & ~~{\sf Resp}~~\\
        \hline
        \tau_{1} & ~~0.2188~~ & ~~1/3~~\\
        \tau_{3} & ~~0.1563~~ & ~~1/2~~\\
        \cline{2-3}
    \end{array}$
    }
   \vspace{2mm} \caption{\hspace*{1cm} (a) \ \ Instance $D_{b}$. \hspace{1cm} (b) \ \ Scores on $D_{b}$.}
    \label{tab:q3b}
\end{table}

\begin{example} \label{ex:q3b}
    Consider query $\mc{Q}_{b}$ in Equation (\ref{eq:cr-not-aligned-atoms3}), and $D_{b}$ the instance in Table~\ref{tab:q3b}, with $D_{b}^\nit{ex} = \{\tau_6\}$. \ The lineage of $\mc{Q}_{b}$ on $D_{b}$ is $((X_{\tau_1} \lor X_{\tau_2}) \land X_{\tau_4}) \lor (X_{\tau_3} \land X_{\tau_5} \land X_{\tau_7})$ (as in Equation (\ref{eq:lineage-counter-atoms3})). \ It holds: \ $D_{b} \models \mc{Q}_{b}$, \ $\tau_3 \prec^\CE \! \tau_1$ and $\tau_1 \prec^{\nit{Resp}} \! \tau_3$. \ {\sf CES} and {\sf Resp} are not aligned for $(\mc{Q}_{b},D_{b})$. \boxtheorem
\end{example}

\begin{table}[h]
    \centering
    {\scriptsize
    $\begin{array}{l|c|c|}
        \hline
        R & ~X~ & ~Y~\\
        \hline
        \tau_{1} & a & a\\
        \tau_{2} & a & b\\
        \tau_{3} & c & d\\
        \cline{2-3}
    \end{array}$~~~~~~$\begin{array}{l|c|}
        \hline
        S & ~X~\\
        \hline
        \tau_{4} & a\\
        \tau_{5} & c\\
        \cline{2-2}
    \end{array}$~~~~~~$\begin{array}{l|c|}
        \hline
        T & ~Y~\\
        \hline
        \underline{\tau_{6}} & a\\
        \underline{\tau_{7}} & b\\
        \tau_{8} & d\\
        \cline{2-2}
    \end{array}$
    ~~~~~~~~~~~~~
    $\begin{array}{l|c|c|}
        \hline
        \tau & ~~{\sf CES}~~ & ~~{\sf Resp}~~\\
        \hline
        \tau_{1} & ~~0.2188~~ & ~~1/3~~\\
        \tau_{3} & ~~0.1563~~ & ~~1/2~~\\
        \cline{2-3}
    \end{array}$
    }
   \vspace{2mm} \caption{\hspace*{1cm}(a) \ \ Instance $D_{c}$. \hspace{1cm} (b) \ \ Scores on $D_{c}$.}
    \label{tab:q3c}
\end{table}

\begin{example} \label{ex:q3c} \
    Consider query $\mc{Q}_{c}$  Equation (\ref{eq:cr-not-aligned-atoms3}), and $D_{c}$ the instance in Table~\ref{tab:q3c}, with $D_{c}^\nit{ex} = \{\tau_6, \tau_7\}$. \ The lineage of $\mc{Q}_{c}$ on $D_{c}$ is $((\tau_1 \lor \tau_2) \land \tau_4) \lor (\tau_3 \land \tau_5 \land \tau_8)$ (as in Equation (\ref{eq:lineage-counter-atoms3})). \ It holds: \ $D_{c} \models \mc{Q}_{c}$, \ $\tau_3 \prec^\CE \! \tau_1$ and $\tau_1 \prec^{\nit{Resp}} \! \tau_3$. \ {\sf CES} and {\sf Resp} are not aligned for $(\mc{Q}_{c},D_{c})$. \boxtheorem
\end{example}

\noindent {\bf Proof of Proposition \ref{prop:cr-not-aligned-A}: }
Since $\mc{Q}$ is single-component, there is a variable $x$ that appears in at least two atoms. If $x$ appears in all the atoms, let $y$ be a variable that does not, since not all atoms have the same variables. Otherwise, let $y$ be a variable that appears in an atom with $x$ and in an atom without $x$, which exists given that $\mc{Q}$ is single-component.

In both cases, $\nit{Atoms}(x) \cup \nit{Atoms}(y)$ has at least three atoms, and we select three of them: $A_R$ containing both variables, $A_T$ containing only one of them, and any other $A_S$. In what follows, the variables $x$ and $y$, and the atoms $A_R$, $A_S$ and $A_T$ of $\mc{Q}$ play the role of the variables $x$ and $y$ and the atoms $R$, $S$ and $T$ of $\mc{Q}_{a}$, $\mc{Q}_{b}$ or $\mc{Q}_{c}$, respectively.

For an atom $A$, define $\nit{V\!ar}_{\!x,y}(A) := \nit{V\!ar}(A) \cap \{x,y\}$. Then, depending on $A_S$, the atoms will satisfy one of the following mutually exclusive conditions:
\begin{align*}
    \text{(C1)} ~ & \ \nit{V\!ar}_{\!x,y}(A_T) \subsetneqq \nit{V\!ar}_{\!x,y}(A_S) = \nit{V\!ar}_{\!x,y}(A_R), \\
    \text{(C2)} ~ & \ \nit{V\!ar}_{\!x,y}(A_T), \nit{V\!ar}_{x,y}(A_S) \subsetneqq \nit{V\!ar}_{\!x,y}(A_R), \text{ with } \nit{V\!ar}_{\!x,y}(A_T) = \nit{V\!ar}_{\!x,y}(A_S), \text{ or} \\
    \text{(C3)} ~ & \ \nit{V\!ar}_{\!x,y}(A_T), \nit{V\!ar}_{\!x,y}(A_S) \subsetneqq \nit{V\!ar}_{\!x,y}(A_R), \text{ with } \nit{V\!ar}_{\!x,y}(A_T) \neq \nit{V\!ar}_{\!x,y}(A_S),
\end{align*}
Define $\atoms_{\sf \,RST}(\mc{Q},v) := \atoms(\mc{Q},v) \cap \{A_R,A_S,$ $A_T\}$.

If the query satisfies case (C1), (C2) or (C3), we build an instance $D_{(C1)}$, $D_{(C2)}$ or $D_{(C3)}$ starting with the tuples in the instances $D_{a}$, $D_{b}$ or $D_{c}$ in Examples \ref{ex:q3a}, \ref{ex:q3b} or \ref{ex:q3c}, respectively, as follows:

\vspace{1mm} \noindent (a) \ For each atom $U_x \in \nit{Atoms}(\mc{Q})$, with  $U_x \in \nit{Atoms}(\mc{Q},x) \smallsetminus \nit{Atoms}(\mc{Q}, y)$, and tuple $\tau_x$ made from the relations of the atoms in $\nit{Atoms}(\mc{Q}_{a},x)$, $\nit{Atoms}(\mc{Q}_{b},x)$ or $\nit{Atoms}(\mc{Q}_{c},x)$ that are in $D_{a}$, $D_{b}$ or $D_{c}$, respectively, we create a tuple from $U_x$ by putting in the $x$'s position the value of $x$ in tuple $\tau_x$, and replacing by a constant $c^\prime$ the rest of variables. The tuples created from atoms in $\atoms_{\sf \,RST}(\mc{Q},x)$ are endogenous iff the corresponding tuples in $D_{a}$, $D_{b}$ or $D_{c}$, respectively, are also endogenous. The other of tuples created from an $U_x$ atom, with $U_x$ different from the atoms in $\atoms_{\sf \,RST}(\mc{Q},x)$, are exogenous.

\vspace{1mm} \noindent (b) \ For each atom $U_y \in \nit{Atoms}(\mc{Q})$, with $U_y \in \nit{Atoms}(\mc{Q},y) \smallsetminus \nit{Atoms}(\mc{Q}, x)$, and tuple $\tau_y$ made from the relations of the atoms in $\nit{Atoms}(\mc{Q}_{a},y)$, $\nit{Atoms}(\mc{Q}_{b},y)$ or $\nit{Atoms}(\mc{Q}_{c},y)$ that are in $D_{a}$, $D_{b}$ or $D_{c}$, respectively, we create a tuple from $U_y$ by putting in the $y$'s position the value of $y$ in the tuple $\tau_y$, and replacing by a constant $c^\prime$ the rest of variables. The tuples created from $\atoms_{\sf \,RST}(\mc{Q},y)$ are endogenous iff the corresponding tuples in $D_{a}$, $D_{b}$ or $D_{c}$, respectively, are also endogenous. The rest of the tuples are exogenous. The other of tuples created from an $U_y$ atom, with $U_y$ different from the atoms in $\atoms_{\sf \,RST}(\mc{Q},y)$, are exogenous.

\vspace{1mm} \noindent (c) \ For each atom $U_{x,y} \in \nit{Atoms}(\mc{Q})$, with $U_{x,y} \in \nit{Atoms}(\mc{Q},x) \cap \nit{Atoms}(\mc{Q}, y)$, and tuple $\tau_{x,y}$ made from the relations of the atoms in $\nit{Atoms}(\mc{Q}_{a},\{x,y\})$, $\nit{Atoms}(\mc{Q}_{b},\{x,y\})$ or $\nit{Atoms}(\mc{Q}_{c},\{x,y\})$ that are in $D_{a}$, $D_{b}$ or $D_{c}$, respectively, we create a tuple from $U_{x,y}$ by putting in the $x$ and $y$'s position the value of $x$ and $y$ in the tuple $\tau_{x,y}$, and replacing by a constant $c^\prime$ the rest of variables. The tuples created from $\atoms_{\sf \,RST}(\mc{Q},\{x,y\})$ are endogenous iff the corresponding tuples in $D_{a}$, $D_{b}$ or $D_{c}$, respectively, are also endogenous. The other of tuples created from an $U_{x,y}$ atom, with $U_{x,y}$ different from the atoms in $\atoms_{\sf \,RST}(\mc{Q},x) \cup \atoms_{\sf \,RST}(\mc{Q},y)$, are exogenous.

\vspace{1mm} \noindent (d) \ For each atom $U \in \nit{Atoms}(\mc{Q})$, with  $U \not \in (\nit{Atoms}(x) \cup \nit{Atoms}(y))$, a tuple is created by replacing all variables in $U$ with a constant $c^\prime$. All tuples created in this way are exogenous.

\vspace{1mm}
By construction of $D_{(C1)}$, $D_{(C2)}$ or $D_{(C3)}$, there are bijections $f_a$, $f_b$ or $f_c$ from the endogenous tuples of $D_{a}$, $D_{b}$ or $D_{c}$ to the endogenous tuples of $D_{(C1)}$, $D_{(C2)}$ and $D_{(C3)}$, respectively, such that, for every set $S$ of endogenous tuples:
\begin{align*}
    \mc{Q}_{a}[S \cup D_{a}^\nit{ex}] &= \mc{Q}[f_a(S) \cup D_{(C1)}^\nit{ex}], \\
    \mc{Q}_{b}[S \cup D_{b}^\nit{ex}] &= \mc{Q}[f_b(S) \cup D_{(C2)}^\nit{ex}], \text{ or} \\
    \mc{Q}_{c}[S \cup D_{c}^\nit{ex}] &= \mc{Q}[f_c(S) \cup D_{(C3)}^\nit{ex}],
\end{align*}
 By Proposition \ref{prop:same-values-lineage}, $(\mc{Q}_{a},D_{a})$, $(\mc{Q}_{b},D_{b})$ or $(\mc{Q}_{c},D_{c})$ have the same lineage as $(\mc{Q},D_{(C1)})$, $(\mc{Q},D_{(C2)})$ or $(\mc{Q},D_{(C3)})$, respectively, and the corresponding tuples have the same {\sf CES} and {\sf Resp} scores. It follows that:
\begin{align*}
    f_a(\tau_1) \prec^\CE f_a(\tau_5) \quad \text{and} \quad & f_a(\tau_5) \prec^\nit{Re} f_a(\tau_1), && \text{ with $\tau_1,\tau_5 \in D_{a}$,}\\
    f_b(\tau_3) \prec^\CE f_b(\tau_8)  \quad \text{and} \quad & f_b(\tau_8) \prec^\nit{Re} f_b(\tau_3),&& \text{ with $\tau_3,\tau_8 \in D_{b}$, and}\\
    f_c(\tau_1) \prec^\CE f_c(\tau_5) \quad \text{and} \quad & f_c(\tau_5) \prec^\nit{Re} f_c(\tau_1),&& \text{ with $\tau_1,\tau_5 \in D_{c}$}.
\end{align*}
As a consequence, the scores are not aligned for $(\mc{Q},D_{(C1)})$, $(\mc{Q},D_{(C2)})$ or $(\mc{Q},D_{(C3)})$, respectively.
\boxtheorem

We now prove that there exists a counterexample for any SJF-BCQ satisfying condition (B), as stated right before Proposition \ref{prop:cr-not-aligned-A}.


\begin{proposition}  \label{prop:cr-not-aligned-B} \em
    Let $\mc{Q}$ be a SJF-BCQ with at least two components, of which at least one has two or more atoms. There is an instance $D$ with exogenous tuples where {\sf CES} and {\sf Resp} are not aligned. \boxtheorem
\end{proposition}

\begin{example} \label{ex:q2c-resp}
    The query  $\mc{Q}_{2C}\!: \ \exists x \exists y \, (R(x) \land S(y) \land T(y))$, for ``two components",  satisfies the hypothesis of the proposition.

\begin{table}[h]
    \centering
    {\scriptsize
    $\begin{tabu}{l|c|}
        \hline
        R & ~X~\\
        \hline
        \tau_{1} & a\\
        \tau_{2} & b\\
        \tau_{3} & c\\
        \cline{2-2}
    \end{tabu}$~~~~~~$\begin{tabu}{l|c|}
        \hline
        S & ~Y~\\
        \hline
        \tau_{4} & a\\
        \tau_{5} & b\\
        \tau_{6} & c\\
        \tau_{7} & d\\
        \cline{2-2}
    \end{tabu}$~~~~~~$\begin{tabu}{l|c|}
        \hline
        T & ~Y~\\
        \hline
        \tau_{8} & a\\
        \tau_{9} & b\\
        \tau_{10} & c\\
        \tau_{11} & d\\
        \cline{2-2}
    \end{tabu}$
    ~~~~~~~~~~~~~~~~~~~~~~
    $\begin{tabu}{l|c|c|}
        \hline
        \tau & ~~{\sf CES}~~ & ~~{\sf Resp}~~\\
        \hline
        \tau_{1} & ~~0.1709~~ & ~~1/3~~\\
        \tau_{4} & ~~0.1846~~ & ~~1/4~~\\
        \cline{2-3}
    \end{tabu}$
    }
   \vspace{2mm} \caption{\hspace*{0.3cm} (a) \ \ Instance $D_{2C}$ (all endogenous tuples). \hspace{0.5cm} (b) \ \ Scores on $D_{2C}$.}
    \label{tab:q2c-resp}
\end{table}

\begin{table}[h]
    \centering
    {\scriptsize
    $\begin{tabu}{l|c|}
        \hline
        R & ~X~\\
        \hline
        \tau_{1} & a\\
        \tau_{2} & b\\
        \cline{2-2}
    \end{tabu}$~~~~~~$\begin{tabu}{l|c|}
        \hline
        S & ~Y~\\
        \hline
        \tau_{3} & a\\
        \tau_{4} & b\\
        \tau_{5} & c\\
        \cline{2-2}
    \end{tabu}$~~~~~~$\begin{tabu}{l|c|}
        \hline
        T & ~Y~\\
        \hline
        \tau_{6} & a\\
        \tau_{7} & b\\
        \underline{\tau_{8}} & c\\
        \cline{2-2}
    \end{tabu}$
    ~~~~~~~~~~~~~~~~~~~~~~
    $\begin{tabu}{l|c|c|}
        \hline
        \tau & ~~{\sf CES}~~ & ~~{\sf Resp}~~\\
        \hline
        \tau_{1} & ~~0.3594~~ & ~~1/2~~\\
        \tau_{5} & ~~0.4219~~ & ~~1/3~~\\
        \cline{2-3}
    \end{tabu}$
    }
   \vspace{2mm} \caption{\hspace*{1cm} (a) \ \ Instance $D^\prime_{2C}$. \hspace{1cm} (b) \ \ Scores on $D^\prime_{2C}$.}
    \label{tab:q2c-resp-exo}
\end{table}

For the instance  $D_{2C}$ in Table~\ref{tab:q2c-resp}, with $D^{\,\nit{ex}}_{2C} = \emptyset$, \ it holds: \ $D_{2C} \models \mc{Q}_{2C}$, \ $\tau_1 \prec^\CE \! \tau_4$ and $\tau_4 \prec^{\nit{Re}} \! \tau_1$. \ {\sf CES} and {\sf Resp} are not aligned for $(\mc{Q}_{2C},D_{2C})$.

For the instance  $D^\prime_{2C}$ in Table~\ref{tab:q2c-resp-exo}, whose only exogenous tuple is $\tau_8$, \ it holds: \ $D^\prime_{2C} \models \mc{Q}_{2C}$, \ $\tau_1 \prec^\CE \! \tau_5$ and $\tau_5 \prec^{\nit{Re}} \! \tau_1$. \ {\sf CES} and {\sf Resp} are not aligned for $(\mc{Q}_{2C},D^\prime_{2C})$.
 \boxtheorem
\end{example}

In the particular case of Example \ref{ex:q2c-resp}, we provide two instances, one without exogenous tuples and one with them. \
The proof of Proposition \ref{prop:cr-not-aligned-B} uses $D_{2C}$ as the base counterexample, that is, an instance without exogenous tuples. \
However, instance $D_{2C}^\prime$ (with exogenous tuples) could have been used instead, without affecting either the structure or conclusion of the proof.

Moreover, it is worth highlighting that  we did not provide instances without exogenous tuples for the examples presented earlier, namely, Examples \ref{ex:q3a},  \ref{ex:q3b} and  \ref{ex:q3c} in this section.\
It turns out that, for queries $\mc{Q}_{3,a}$ and $\mc{Q}_{3,b}$, no counterexample without exogenous tuples exists where {\sf CES} and {\sf Resp} are not aligned. That is, the scores are always aligned with instances restricted to having only endogenous tuples. \
We investigate this issue in more detail in the Section \ref{sec:ces-resp_noex}.

\ignore{For other queries in the same class we may need them, as shown in the proof of Proposition \ref{prop:cr-not-aligned-B}.}

\vspace{3mm}
\noindent \textbf{Proof of Proposition \ref{prop:cr-not-aligned-B}: }
Let $C_1$ and $C_2$ be two components of $\mc{Q}$, such that $|C_2| \geq 2$.\
From $C_1$, select one atom, denoted with $A_R$. We assume that $A_R \in \nit{Atoms}(x)$. \
Similarly, from $C_2$ select two atoms, $A_S$ and $A_T$, such that, for a variable $y$, $A_S, A_T \in \nit{Atoms}(y)$.

No we proceed to build an instance $D$ starting from the tuples in the instance $D_{2C}$ in Example \ref{ex:q2c-resp}, as follows:

\vspace{1mm} \noindent (a) For each atom $U_x \in \nit{Atoms}(\mc{Q})$, with $U_x \in \nit{Atoms}(x)$, and for each tuple $\tau_R$ of relation $R$ of $D_{2C}$, we create a tuple $\tau_R$ from $U_R$ by putting in $x$'s position the value of $x$ in  tuple $\tau_R$, and replacing by a constant $c^\prime$ the rest of variables.

Only tuples that are instantiations of $A_R$ are declared as endogenous, and only when  the corresponding tuple in $D_{2C}$ is also endogenous. The other of tuples created from an $U_x$ atom, with $U_x$ different from $A_R$, are exogenous.

\vspace{1mm} \noindent (b) For each atom $U_S,U_T \in \nit{Atoms}(\mc{Q})$, with $U_S,U_T \in \nit{Atoms}(y)$, and tuples $\tau_S$ and $\tau_T$ from relations $S$ and $T$ of $D_{2C}$, we create two tuples, from $U_S$ and $U_T$, by putting in $y$'s position the value of $y$ in the tuple $\tau_S$ and $\tau_T$, respectively, and replacing by a constant $c^\prime$ the rest of variables.

Only tuples from $A_S$ and $A_T$ are endogenous iff the corresponding tuple in $D_{2C}$ is also endogenous.\ The other of tuples created from an $U_S$ or $U_T$ atom, with $U_S$ and $U_T$ different from $A_S$ and $A_T$, respectively, are exogenous.

\vspace{1mm} \noindent (c) For each atom $U \in \nit{Atoms}(\mc{Q})$, with $U \not \in (\nit{Atoms}(x) \cup \nit{Atoms}(y))$, a tuple is created by replacing all variables in $U$ with a constant $c^\prime$. All tuples created in this way are exogenous.

By construction of $D$, there is a bijection $f$ from the endogenous tuples of $D_{2C}$ to the endogenous tuples in $D$, such that, for every set $S$ of endogenous tuples, $\mc{Q}_{2C}[S \cup D_{2C}^\nit{ex}] = \mc{Q}[f(S) \cup D^\nit{ex}]$. By Proposition \ref{prop:same-values-lineage}, $(\mc{Q}_{2C},D_{2C})$ and $(\mc{Q},D)$ have the same lineage, and the corresponding tuples have the same {\sf CES} and {\sf Resp} scores. It follows that $f(\tau_1) \prec^\CE f(\tau_4)$ and $f(\tau_4) \prec^\nit{Re} f(\tau_1)$, and thus, the scores are not aligned for $(\mc{Q},D)$.
\boxtheorem

Given a SJF-BCQ, the proofs of Propositions \ref{prop:cr-not-aligned-A} and \ref{prop:cr-not-aligned-B} provide an explicit, algorithmic, way to create an instance $D$ where {\sf CES} and {\sf Resp} are not aligned for $(D,\mc{Q})$. \
Moreover, this construction can be done efficiently: It depends linearly on the size of the query, as formally captured in the following result.

\begin{corollary}  \label{cor:efficient_counterexamples} \em
    Let $\mc{Q}$ be a SJF-BCQ, satisfying one of the conditions of Propositions \ref{prop:cr-not-aligned-A} or \ref{prop:cr-not-aligned-B}. An instance $D$ where {\sf CES} and {\sf Resp} are not aligned for $(D,\mc{Q})$ can be built in linear time in the size of $\mc{Q}$.
\end{corollary}

\noindent \textbf{Proof: \ } Let $\mc{Q}$ be a SJF-BCQ satisfying the conditions of Proposition \ref{prop:cr-not-aligned-A} or \ref{prop:cr-not-aligned-B}.\
In either case, the proof of the proposition builds an instance $D$ for which {\sf CES} and {\sf Resp} are not aligned. This construction is done by: Creating, for each atom in $\mc{Q}$, a fixed number of tuples, by replacing its variables for particular constants. \
The number of tuples created per atom is fixed, and does not depend on the size of the query. This process of creating tuples is only carried out once per atom.
It follows that the instance $D$ can be created in linear time in the size of $\mc{Q}$.
\boxtheorem

We illustrate the proof of Proposition \ref{prop:cr-not-aligned-B}, and also Corollary \ref{cor:efficient_counterexamples} with an example.

\begin{example}
    \ Consider the query:
    \begin{equation} \label{eq:r1_r5}
        \mc{Q}:\ \exists x  \exists y \exists z \exists w\ (R_1(x) \land R_2(x,y) \land R_3(z) \land R_4(z) \land R_5(w)),
    \end{equation}
    that  satisfies the conditions of Proposition \ref{prop:cr-not-aligned-B}. Accordingly,  there is an instance $D$ where {\sf CES} and {\sf Resp} are not aligned for $(\mc{Q},D)$. \ Next,  we are going to build such an instance according the proof of  Proposition \ref{prop:cr-not-aligned-B}.

\begin{table}[h]
    \centering
    {\scriptsize
    $\begin{array}[t]{l|c|}
        \hline
        R_1 & ~X~\\
        \hline
        \tau_{1} & a\\
        \tau_{2} & b\\
        \tau_{3} & c\\
        \cline{2-2}
    \end{array}$~~~~~~~$\begin{array}[t]{l|c|c|}
        \hline
        R_2 & ~X~ & ~Y~\\
        \hline
        \underline{\tau_{4}} & a & c^\prime\\
        \underline{\tau_{5}} & b & c^\prime\\
        \underline{\tau_{6}} & c & c^\prime\\
        \cline{2-3}
    \end{array}$~~~~~~~$\begin{array}[t]{l|c|}
        \hline
        R_3 & ~Z~\\
        \hline
        \tau_{7} & a\\
        \tau_{8} & b\\
        \tau_{9} & c\\
        \tau_{10} & d\\
        \cline{2-2}
    \end{array}$~~~~~~~$\begin{array}[t]{l|c|}
        \hline
        R_4 & ~Z~\\
        \hline
        \tau_{11} & a\\
        \tau_{12} & b\\
        \tau_{13} & c\\
        \tau_{14} & d\\
        \cline{2-2}
    \end{array}$~~~~~~~$\begin{array}[t]{l|c|}
        \hline
        R_5 & ~W~\\
        \hline
        \underline{\tau_{15}} & c^\prime\\
        \cline{2-2}
    \end{array}$~~~~~~~~~~~~$\begin{array}[t]{l|c|c|}
        \hline
        \tau & ~~{\sf CES}~~ & ~~{\sf Resp}~~\\
        \hline
        \tau_{1} & ~~0.1709~~ & ~~1/3~~\\
        \tau_{7} & ~~0.1846~~ & ~~1/4~~\\
        \cline{2-3}
    \end{array}$
    }
    \vspace{2mm}\caption{\hspace*{1cm} (a) \ \ Instance $D$. \ \hspace{1cm} (b) \ \ Scores on $D$.}
    \label{tab:ex-efficient-counter}
\end{table}

    Let us start by selecting the atoms $R_1(x)$, $R_3(z)$ and $R_4(z)$, to play the roles of $A_R$, $A_S$ and $A_T$ in it, respectively. Then, proceed as follows:

    \begin{enumerate}[(a)]
      \item For atoms $R_1(x)$ and $R_2(x,y)$, we create a tuple, for each tuple $\tau_R$ in $R^D$ in Table \ref{tab:q2c-resp}, by replacing $x$ with the constant in the $x$ position of $\tau_R$, and replacing the other variables with $c^\prime$. For example, consider the tuple $\tau_1\! :\! R(a)$, we create the following tuples: $R_1(a)$ and $R_2(a,c^\prime)$. The same is done for tuples $\tau_2\!:\! R(b)$ and $\tau_3\!:\! R(c)$ in Table \ref{tab:q2c-resp}.\
        The tuples created are the following: $R_1(a), R_1(b), R_1(c), R_2(a,c^\prime), R_2(b,c^\prime)$ and $R_2(c,c^\prime)$.\
        Only the tuples created from the atom $R(x)$ are endogenous, since it plays the role of $A_R$.\
        All tuples from the atom $R_2(x,y)$ are exogenous.

      \item For atoms $R_3(z)$ and $R_4(z)$, we proceed as in (a), now with the tuples in $S^D$ and $T^D$ in Table \ref{tab:q2c-resp}, respectively: for each tuple in $S^D$ we create a tuple for $R_3(z)$, and for each tuple in $T^D$ a tuple for $R_4(z)$, by replacing $z$ with the constant of that tuple. For example, from $S(a)$ and $T(a)$, we create $R_3(a)$ and $R_4(a)$.\
        The tuples created are the following: $R_3(a), R_3(b), R_3(c), R_3(d), R_4(a), R_4(b), R_4(c), R_4(d)$. All of these are declared endogenous, since $R_3(z)$ and $R_4(z)$ play the roles of $A_S$ and $A_T$.

       \item Lastly, for atom $R_5(w)$, we obtain $R_5(c^\prime)$, which is declared exogenous.
    \end{enumerate}

    The resulting instance is shown in Table \ref{tab:ex-efficient-counter}(a). Moreover, as shown in Table \ref{tab:ex-efficient-counter}(b), {\sf CES} and {\sf Resp} are not aligned for $(\mc{Q},D)$. Notice that, compared to instance $D_{2C}$, $D$ has four (exogenous) tuples more: \ three  from $R_2^D$, and one from $R_5^D$.

    The construction in this example also gives an idea of how much the counterexample instance grows on the size of the query. \ If we add to $\mc{Q}$ in (\ref{eq:r1_r5}) a new atom, with a new predicate, the counterexample instance, built in the same way for the resulting query, has:
    (1) Three additional tuples, if the new atom has $x \in \nit{V\!ar}(\mc{Q})$ among its variables. \
    (2) Four additional tuples, if the new atom has $z \in \nit{V\!ar}(\mc{Q})$; and \
    (3) One additional tuple, if the new atom has none of the variables from $\mc{Q}$. \ That is, the number of tuples in the counterexample growths linearly in the size of the query.\boxtheorem
\end{example}

\subsection{Proof of Theorem \ref{theo:ces_resp_aligned} and Further Comments}
\label{sec:proof}\vspace{2mm}

For completeness, we now finalize the proof of the main result of this section.

\vspace{3mm}
\noindent {\bf Proof of Theorem \ref{theo:ces_resp_aligned}:} \
Consider a SJF-BCQ $\mc{Q}$ satisfying one of the conditions in the hypothesis. If it satisfies (i), (ii) or (iii), then, by Propositions \ref{prop:cr_aligned_equal_vars},  \ref{prop:cr_aligned_coin_atom_2} and  \ref{prop:cr_aligned_multi}, \ {\sf CES} and {\sf Resp} are aligned for $(\mc{Q},D)$, for every instance $D$, with or without exogenous tuples.\

If $\mc{Q}$ does not satisfy any of the conditions of Theorem \ref{theo:ces_resp_aligned}(\ref{theo:ces_resp_aligned-a}), it satisfies one of the following, with $n$ the number of components of the query:
\begin{enumerate}[(A)]
    \item $n=1$ and  $|\nit{Atoms}(\mc{Q})| \geq 3$; \ or
    \item $n \geq 2$, with at least two atoms in at least one component.
\end{enumerate}

If a query satisfies (A) or (B), according to Propositions \ref{prop:cr-not-aligned-A} and  \ref{prop:cr-not-aligned-B}, respectively, there exists an instance $D$ such {\sf CES} and {\sf Resp} are not aligned for $(\mc{Q},D)$. \boxtheorem

With Theorem \ref{theo:ces_resp_aligned} we have a full characterization of the space of SJF-BCQs, as shown in Figure \ref{fig:main-theorem-diagram}. The diagram allows us to classify any SJF-BCQ. \ We now illustrate the use of Theorem \ref{theo:ces_resp_aligned} with an example.

\begin{figure}
    \centering
    \begin{tikzpicture}[
        box/.style  = {draw, rounded corners=2pt, align=center, font=\scriptsize,
                       inner sep=2pt, execute at begin node={\hyphenpenalty=10000}},
        cond/.style = {box, text width=2.15cm, minimum height=1.3cm, anchor=north},
        res/.style  = {box, text width=2.15cm, anchor=north},
        arr/.style  = {-{Stealth[length=3pt]}, thin},
      ]

      \node[box] (root) at (0,0)
        {Are {\sf CES} and {\sf Resp} aligned for a SJF-BCQ $\mc{Q}$?};

      \node[box, anchor=north] (one)  at (-2.53,-0.52) {One component};
      \node[box, anchor=north] (many) at ( 2.53,-0.52) {More than one component};

      \node[cond] (one-yes)  at (-3.73,-1.26)
        {At most two atoms, or all atoms have the same variables};
      \node[cond] (one-no)   at (-1.33,-1.26)
        {Three or more atoms, not all with the same variables};
      \node[cond] (many-yes) at ( 1.33,-1.26)
        {Each component\\ has one atom};
      \node[cond] (many-no)  at ( 3.73,-1.26)
        {Some component has two or more atoms};

      \node[res] (r1) at (-3.73,-2.88) {The scores\\ \textbf{are aligned}};
      \node[res] (r2) at (-1.33,-2.88) {The scores\\ \textbf{are not aligned}};
      \node[res] (r3) at ( 1.33,-2.88) {The scores\\ \textbf{are aligned}};
      \node[res] (r4) at ( 3.73,-2.88) {The scores\\ \textbf{are not aligned}};

      \draw[arr] (root.south) -- (one.north);
      \draw[arr] (root.south) -- (many.north);
      \foreach \p/\c in {one/one-yes, one/one-no, many/many-yes, many/many-no}
        \draw[arr] (\p.south) -- (\c.north);
      \foreach \c/\r in {one-yes/r1, one-no/r2, many-yes/r3, many-no/r4}
        \draw[arr] (\c.south) -- (\r.north);
    \end{tikzpicture}
    \caption{ \ \ Decision Chart for  Theorem \ref{theo:ces_resp_aligned}.}
    \label{fig:main-theorem-diagram}
\end{figure}

\begin{example} \ \label{ex:prop16} Consider the query \
    $\mc{Q}\!\!:  \exists x \exists y \exists z \exists v \exists w ($ $R_1(x,y)\land R_2(z) \land R_3(z) \land R_4(z) \land R_5(w,v))$. \ It has the  components: $C_1 = \{R_1(x,y)\}$, $C_2 = \{R_2(z),
    R_3(z),$ $ R_4(z)\}$, $C_3 = \{R_5(w,v)\}$. \ Since $C_2$ has three atoms, by Theorem \ref{theo:ces_resp_aligned} (and Proposition \ref{prop:cr-not-aligned-B}), an instance $D$ can be built where the scores are not aligned. Actually, as follows from the instance $D_{2C}$ in Example~\ref{ex:q2c-resp}: \ Select components $C_2$ and $C_3$. \ From $C_2$,  select atoms $R_2(z), R_3(z)$; and  from $C_3$, atom $R_5(w,v)$. Tuples from  $R_2,R_3, R_5$ are endogenous; and the rest, exogenous. They are  underlined in Table \ref{tab:q_2comp_ex}(a) that shows instance $D$.

    Table \ref{tab:q_2comp_ex}(b) shows that $\tau_{14} \prec^\CE \! \tau_2$ and $\tau_2 \prec^\nit{Re} \! \tau_{14}$, so {\sf CES} and {\sf Resp}  are not aligned for $(\mc{Q},D)$. \boxtheorem
\end{example}

    \begin{table}
       \hspace*{10mm} {\scriptsize
        $\begin{tabu}{l|c|c|}
            \hline
            R_1 & ~A~ & ~B~\\
            \hline
            \underline{\tau_{1}}  & c^\prime & c^\prime\\
            \cline{2-3}
        \end{tabu}$
        ~~~~
        $\begin{tabu}{l|c|}
            \hline
            R_2~  & ~A~ \\\hline
            \tau_{2} & a\\
            \tau_{3} & b\\
            \tau_{4} & c\\
            \tau_{5} & d\\
            \cline{2-2}
        \end{tabu}$
        ~~~~
        $\begin{tabu}{l|c|}
            \hline
            R_3~  & ~A~ \\\hline
            \tau_{6} & a\\
            \tau_{7} & b\\
            \tau_{8} & c\\
            \tau_{9} & d\\
            \cline{2-2}
        \end{tabu}$~~~~
        $\begin{tabu}{l|c|}
            \hline
            R_4~  & ~A~ \\\hline
            \underline{\tau_{10}} & a\\
            \underline{\tau_{11}} & b\\
            \underline{\tau_{12}} & c\\
            \underline{\tau_{13}} & d\\
            \cline{2-2}
        \end{tabu}$~~~~
        $\begin{tabu}{l|c|c|}
            \hline
            R_5 & ~A~ & ~B~\\
            \hline
            \tau_{14} & a & c^\prime\\
            \tau_{15} & b & c^\prime\\
            \tau_{16} & c & c^\prime\\
            \cline{2-3}
        \end{tabu}$~~~~~~~~~~
        $\begin{tabu}{l|c|c|}
            \hline
            \tau   & ~~{\sf CES}~~ & ~~{\sf Resp}~~\\
            \hline
            \tau_{2} & ~~0.1846~~ & ~~1/4~~\\
            \tau_{14} & ~~0.1709~~ & ~~1/3~~\\
            \cline{2-3}
        \end{tabu}$
        }
       \vspace{2mm} \caption{\hspace*{1cm}(a) \ \ Instance $D$. \hspace{1cm} (b) \ \ {\sf CES} and {\sf Resp}.}
        \label{tab:q_2comp_ex}
    \end{table}

The previous example deserves a closer look. \ In relation to the scores for subqueries $\mc{Q}_1, \mc{Q}_2, \mc{Q}_3$ corresponding to $C_1, C_2, C_3$, Theorem \ref{theo:ces_resp_aligned}(\ref{theo:ces_resp_aligned-a})(i) and   Proposition \ref{prop:cr_aligned_equal_vars} tell us they are are always aligned for $\mc{Q}_1$, $\mc{Q}_2$ and $\mc{Q}_3$. We can see that  alignment for the subqueries does not entail alignment for original query. \
However, Corollary \ref{coro:multi_single_alignment} below tells us that the inverse does hold.

\ignore{\comfa{I removed the Remark here, is not being referenced anywhere else, and the key idea is being captured in the corollary below. I dont think it requires its own remark.}}

\begin{corollary} \label{coro:multi_single_alignment} \em
    Let $\mc{Q}$ be a SJF-BCQ with at least two components. If {\sf CES} and {\sf Resp} are aligned for $(D,\mc{Q})$, with $D$ with or without exogenous tuples,  then these scores are aligned on $D$ also for the subqueries built from its components. \boxtheorem
\end{corollary}

\vspace{2mm}
\noindent {\bf Proof: \ }
For $\mc{Q}$,  assume that {\sf CES} and {\sf Resp} are aligned for $(\mc{Q},D)$, where $D$ is an arbitrary instance. \
By Propositions \ref{prop:cr_aligned_multi}, all the components of $\mc{Q}$ have a single atom. It follows from Proposition \ref{prop:cr_aligned_equal_vars} that the scores are also aligned for $(\mc{Q}_i,D)$. {\boxtheorem}




\section{{\sf CES} and {\sf Resp} Alignment without Exogenous Tuples}
\label{sec:ces-resp_noex}\vspace{2mm}


The  non-alignment result in Theorem \ref{theo:ces_resp_aligned}(b) holds for instances with exogenous tuples. Until now, we do not know what happens  when instances do not have them. This is what we investigate in this section. \ Actually, we will show that exogenous tuples are required for non-alignment for a class of queries that is contained in the class in Theorem \ref{theo:ces_resp_aligned}(b). That class is introduced  in Proposition \ref{prop:ces_resp_hier2_no-exo} below.  \ For queries in it, {\sf CES} and {\sf Resp} are always aligned. \
{\em In this section, we assume that instances do not have exogenous tuples.}

Before presenting the main result of this section, we recall a  common syntactic characterization of conjunctive queries.

\begin{definition} \em \cite{suciu} \label{def:hierarchical}
    Let $\mc{Q}$ be a BCQ. $\mc{Q}$ is \emph{hierarchical}, for any two variables $x,y \in \nit{V\!ar}(\mc{Q})$, it holds: (a) $\atoms(x) \subseteq \atoms(y)$, (b) $\atoms(y) \subseteq \atoms(x)$, or (c) $\atoms(x) \cap \atoms(y) = \emptyset$.\boxtheorem
\end{definition}

\begin{example} \label{ex:hierarchical}
    Consider the query:
    \[
        \mc{Q} \ : \exists x \exists y \exists z \exists w (R(x,y,z,w) \land S(x,y,z,w) \land T(x,y) \land U(x,y)).
    \]
   It holds:
    \begin{eqnarray}
   \atoms(x) &=& \atoms(y) = \{(R(x,y,z,w),S(x,y,z,w),T(x,y),U(x,y))\},  \nonumber \\
        \atoms(z) &=& \atoms(w) = \{(R(x,y,z,w),S(x,y,z,w))\}.\label{eq:hierarchical-ex}
    \end{eqnarray}
    Accordingly,
        $\atoms(z),\atoms(w)  \subsetneqq  \atoms(x), \atoms(y)$. \ By Definition \ref{def:hierarchical}, $\mc{Q}$ is hierarchical.\

    Now, consider the queries $\mc{Q}_{a}$, $\mc{Q}_{b}$ and $\mc{Q}_{c}$ in Equation (\ref{eq:cr-not-aligned-atoms3}). For the first two, $\atoms(\mc{Q}_{a},y) \subseteq \atoms(\mc{Q}_{a},x)$ and $\atoms(\mc{Q}_{b},y) \subseteq \atoms(\mc{Q}_{b},x)$. According to Definition \ref{def:hierarchical}, both queries are hierarchical.

    However, $\atoms(\mc{Q}_{c},x) = \{R(x,y),S(x)\}$, and $\atoms(\mc{Q}_{c},y) = \{R(x,y),$ $T(y)\}$. This query does not satisfy any of the conditions of Definition \ref{def:hierarchical}. It follows that $\mc{Q}_{c}$ is non-hierarchical.\boxtheorem
\end{example}

The \emph{hierarchical} syntactic characterization was introduced in the context of Probabilistic Databases (see Section \ref{sub:prob-pdbs}) to classify the complexity of computing the probability of a query being true in a TID for SJF-BCQs \cite{dichotomy_CQ_other}.

\begin{proposition}  \label{prop:ces_resp_hier2_no-exo} \em
    Let $\mc{Q}$ be a single-component, hierarchical SJF-BCQ, such that the sets of atoms that contain each of its variables are at most two different sets. Then, {\sf CES} and {\sf Resp} are aligned for $(\mc{Q},D)$, for every instance $D$ without exogenous tuples. \boxtheorem
\end{proposition}

We illustrate  this proposition with an example.

\begin{example} \label{ex:ces_resp_hier2_no-exo}
    Consider the query $\mc{Q}$ in Example \ref{ex:hierarchical}. \
    There are only two different sets of atoms that encompass every variable in the query: The set $\atoms(x)$ (which is the same as $\atoms(y)$, and the set $\atoms(z)$ (which is the same as $\atoms(w)$) (see  Equation (\ref{eq:hierarchical-ex})).

    Furthermore, $\mc{Q}$ is hierarchical. Then, by Proposition \ref{prop:ces_resp_hier2_no-exo},  {\sf CES} and {\sf Resp} are aligned for $(\mc{Q},D)$, for any instance $D$ without exogenous tuples.

    Moreover, the queries $\mc{Q}_{a}$ and $\mc{Q}_{b}$ in Equation (\ref{eq:cr-not-aligned-atoms3}) also satisfy the conditions of Proposition \ref{prop:ces_resp_hier2_no-exo}: They are both hierarchical, and their two variables are contained in exactly two different sets of atoms. \
    It follows, that {\sf CES} and {\sf Resp} are aligned for $(\mc{Q}_{a},D^\prime)$ and  $(\mc{Q}_{b},D^{\prime\prime})$, for any instances $D^\prime$ and $D^{\prime\prime}$ without exogenous tuples.
    \boxtheorem
\end{example}

The proof of Proposition \ref{prop:ces_resp_hier2_no-exo} follows the same structure as the proofs in Section \ref{sub:ces-resp-aligned}, without requiring the application of third step (Item 3) outlined at the beginning of that section: In that step, exogenous tuples are reduced from the lineage, changing its structure; for the results of this section, no exogenous tuples are reduced from the lineage.

\vspace{3mm} \noindent {\bf Proof of Proposition \ref{prop:ces_resp_hier2_no-exo}: \ } Consider a query $\mc{Q}$ satisfying the conditions of the proposition,\ignore{\footnote{See Section \ref{sub:constants} for the assumptions on queries, in particular, they do not have constants.} } and an instance $D$ with $D^\nit{ex} = \emptyset$.
\ We show that {\sf Resp} is aligned with {\sf CES} for $(\mc{Q},D)$, that is, for $\tau, \tau^\prime \in \Den$:
\begin{equation} \label{eq:goal_hier2}
    \text{If }\ \Resp(\tau) > \Resp(\tau^\prime),\ \text{ then }\ \CE(\tau) > \CE(\tau^\prime).
\end{equation}

The sets $\nit{Atoms}(v)$, for $v$ a variable of $\mc{Q}$, are at most two different sets.
If they are all the same set, every atom of $\mc{Q}$ in that set contains all the variables of $\mc{Q}$; and then, all the atoms of $\mc{Q}$ have the same variables. The conclusion follows from Proposition \ref{prop:cr_aligned_equal_vars}.

Otherwise, the variables of $\mc{Q}$ are contained in exactly two different sets of atoms. Let $\bar{x}$ and $\bar{y}$ be the lists of variables contained in each of them. Since $\mc{Q}$ is single-component, some atom contains both $\bar{x}$ and $\bar{y}$ (otherwise, the atoms with $\bar{x}$ and those with $\bar{y}$ would share no variables). Since $\mc{Q}$ is hierarchical, the atoms that contain one of the lists, say $\bar{y}$, are among those that contain the other, $\bar{x}$.

Since every atom contains some variable, $\bar{x}$ appears in every atom. Let $R_1,\ldots,R_k$ be the predicates of the atoms that contain $\bar{y}$; and then, both $\bar{x}$ and $\bar{y}$. Let $S_1,\ldots,S_h$ be the remaining atoms, those that contain only $\bar{x}$. Here $k, h \geq 1$, since the two sets of atoms are different.

As a consequence, and up to a renaming of the predicates and a reordering of the variables, $\mc{Q}$ has the form:
\[
    \mc{Q}\!: \ \exists \bar{x} \exists \bar{y} \, \big(R_1(\bar{x},\bar{y}) \land \cdots \land R_k(\bar{x},\bar{y}) \land S_1(\bar{x}) \land \cdots \land S_h(\bar{x})\big), \quad \mbox{ with } \ k, h \geq 1.
\]

Let $\mc{C}$ be the set of constants for $\bar{x}$ in $D$, with $|\mc{C}| = m$. For $\bar{a} \in \mc{C}$, let $B_{\bar{a}}$ be the set of constants for $\bar{y}$ that appear with $\bar{a}$ in $D$, with $n_{\bar{a}} := |B_{\bar{a}}|$. \
Since $D$ has no dummy tuples (see Remark~\ref{rem:dummy}),  for every $\bar{a} \in \mc{C}$, it holds that $n_{\bar{a}} \geq 1$; that is, $D$ contains tuples $S_1(\bar{a}),\ldots,S_h(\bar{a})$. In fact, if such a tuple $S_k(\bar{a})$ is not in $D$, the tuples in $D$ with constants $\bar{a}$  could never make $\mc{Q}$ true: They would be dummy tuples. \
Similarly, for every $\bar{b} \in B_{\bar{a}}$, $D$ contains $R_1(\bar{a},\bar{b}),\ldots,R_k(\bar{a},\bar{b})$, otherwise, tuples with constants $\bar{a},\bar{b}$ would be dummy tuples.

For $\bar{a} \in \mc{C}$ and $\bar{b} \in B_{\bar{a}}$, write $\psi^{S}_{\bar{a}} := \bigwedge_{j=1}^{h} X_{S_j(\bar{a})}$, and $\psi^{R}_{\bar{a},\bar{b}} := \bigwedge_{l=1}^{k} X_{R_l(\bar{a},\bar{b})}$, respectively,  for the conjunctions of the variables of $S_1(\bar{a}),\ldots,S_h(\bar{a})$, and of the sub-block $(\bar{a},\bar{b})$. \
Then, the lineage of $\mc{Q}$ on $D$ becomes:
\begin{align}
    \mc{L}(\mc{Q},D) \ &= \ \bigvee_{\bar{a} \in \mc{C}} \ \bigvee_{\bar{b} \in B_{\bar{a}}} \big( \psi^{S}_{\bar{a}} \land \psi^{R}_{\bar{a},\bar{b}} \big) \nonumber \\
    &\equiv \ \bigvee_{\bar{a} \in \mc{C}} \Phi_{\bar{a}}, \quad \text{with} \quad \Phi_{\bar{a}} := \psi^{S}_{\bar{a}} \land \varphi^{R}_{\bar{a}}, \quad \text{and} \quad \varphi^{R}_{\bar{a}} := \bigvee_{\bar{b} \in B_{\bar{a}}} \psi^{R}_{\bar{a},\bar{b}}. \label{eq:lin_hier2}
\end{align}

The equivalence follows by factoring $\psi^{S}_{\bar{a}}$ out of the inner disjunction, for each $\bar{a} \in \mc{C}$.\
Each $\bar{a}\in \mc{C}$ defines the disjunct $\Phi_{\bar{a}}$, which we call a \emph{block}; and each $(\bar{a},\bar{b})$, with $\bar{b} \in B_{\bar{a}}$, defines the disjunct $\psi^{R}_{\bar{a},\bar{b}}$, which we call \emph{a sub-block}. \
Distinct blocks share no variables; and neither do $\psi^{S}_{\bar{a}}$ and the distinct sub-blocks $\psi^{R}_{\bar{a},\bar{b}}$ of a block  $\Phi_{\bar{a}}$.

We now compute {\sf Resp} and {\sf CES} for every endogenous tuple, using (\ref{eq:lin_hier2}). Next, we use those values to verify (\ref{eq:goal_hier2}).

For {\sf Resp}, a minimum contingency set for $\tau = S_j(\bar{a})$ must make all the other blocks false, which requires removing at least one tuple from each of them: Removing $S_1(\bar{a}^\prime)$, for each $\bar{a}^\prime \neq \bar{a}$, suffices; and no tuple from block $\bar{a}$ needs to be removed, because $X_\tau$ changes the value of its block when the other tuples in it are present.

For $\tau = R_l(\bar{a},\bar{b})$, a minimum contingency set needs, in addition to the tuples removed for $S_j(\bar{a})$, the removal of one tuple from each of the other $(n_{\bar{a}} - 1)$ sub-blocks of $\bar{a}$. \  Then: \ $\Resp(S_j(\bar{a})) = \frac{1}{m}$, \ and \ $ \Resp(R_l(\bar{a},\bar{b})) = \frac{1}{m + n_{\bar{a}} - 1}$.

Now, let $\tau, \tau^\prime \in \Den$, with $\Resp(\tau) > \Resp(\tau^\prime)$, as in (\ref{eq:goal_hier2}). Since $1/m$ is the highest responsibility, $\Resp(\tau^\prime) < 1/m$. So, $\tau^\prime$ is a tuple $R_{l'}(\bar{a}^\prime,\bar{b}^\prime)$ with $n_{\bar{a}^\prime} \geq 2$; and $\tau$ is either a tuple $S_j(\bar{a})$, or a tuple $R_l(\bar{a},\bar{b})$ of a block with fewer sub-blocks, that is, $n_{\bar{a}} < n_{\bar{a}^\prime}$.\
It follows that, in order to obtain (\ref{eq:goal_hier2}), we only need to prove, for $\tau^\prime = R_{l'}(\bar{a}^\prime,\bar{b}^\prime)$, with $n_{\bar{a}^\prime} \geq 2$, that:
\begin{enumerate}[(a)]
    \item $\CE(S_j(\bar{a})) > \CE(\tau^\prime)$, for every tuple $S_j(\bar{a})$ (possibly with $\bar{a} = \bar{a}^\prime$).
    \item $\CE(R_l(\bar{a},\bar{b})) > \CE(\tau^\prime)$, for every tuple $R_l(\bar{a},\bar{b})$ with $n_{\bar{a}} < n_{\bar{a}^\prime}$.
\end{enumerate}

For {\sf CES}, as in the proof of Proposition \ref{prop:cr_aligned_coin_atom_2}, we denote with $P_0(\psi)$ the probability that a propositional formula $\psi$ is false when each $X$-variable is independently true with probability $1/2$. \ Similarly, $P_0 := P_0(\mc{L}(\mc{Q},D))$.

The conjunction $\psi^{S}_{\bar{a}}$ is false with probability $1 - 2^{-h}$; and $\psi^{R}_{\bar{a},\bar{b}}$, with probability $1 - 2^{-k}$. \ For these probabilities, we drop the constants, writing instead $P_0(\psi^{S})$ and $P_0(\psi^{R})$, because they do not depend on them. \ Notice that $0 < P_0(\psi^{S}), P_0(\psi^{R}) < 1$.

\ignore{
\comfa{``these'' below means any two formulas of the above: for instance, $\Phi_{\bar{c}}$ and $\Phi_{\bar{c}^\prime}$, they do not share any variables. Same with $\psi^S_{\bar{a}}$ and $\psi^S_{\bar{a}^\prime}$;  $\varphi^R_{\bar{a}}$ and $\varphi^R_{\bar{a}^\prime}$ and $\psi^R_{\bar{a},\bar{b}}$ and $\psi^R_{\bar{a}^\prime,\bar{b}^\prime}$.

Creo que es excesivo colocar todos los casos... no?

Este claim es el que justifica que se pueden multiplicar las probabilidades: dos formulas con subscripts distintos no comparten variables.}

\comlb{Es facil, perderse. No colocar todos los casos, pero algo que aclare de cuales formulas estas hablando.

Chequea abajo, que no haya dejado ninguna cagada. Saque las dobleimplicaciones.}

\comfa{Creo que ahí queda mejor.}
}

Variable $X_\tau$ occurs exactly once in the conjunctions $\psi^{S}_{\bar{a}}$ and $\psi^{R}_{\bar{a},\bar{b}}$.\
So, two formulas built from different conjunctions share no variables.\
This is the case, for example, for two sub-blocks $\psi^{R}_{\bar{a},\bar{b}}$ and $\psi^{R}_{\bar{a},\bar{b}^\prime}$, for $\psi^{S}_{\bar{a}}$ and $\varphi^{R}_{\bar{a}}$; and for two different blocks $\Phi_{\bar{a}}$ and $\Phi_{\bar{a}^\prime}$. Then, $\varphi^{R}_{\bar{a}}$ is false when its $n_{\bar{a}}$ sub-blocks are false; and the block $\Phi_{\bar{a}}$ is true when both $\psi^{S}_{\bar{a}}$ and $\varphi^{R}_{\bar{a}}$ are true:
\[
    P_0(\varphi^{R}_{\bar{a}}) = P_0(\psi^{R})^{n_{\bar{a}}}, ~
    P_0(\Phi_{\bar{a}}) = 1 - (1 - P_0(\psi^{S}))\times (1 - P_0(\varphi^{R}_{\bar{a}})), ~
    P_0 = \prod_{\bar{a} \in \mc{C}} P_0(\Phi_{\bar{a}}).
\]

Let $\tau$ be a tuple belonging to block $\Phi_{\bar{a}}$. Since $\mc{L}(\mc{Q},D)$ is monotone, $\CE(\tau)$ is the probability that $X_\tau$ changes the value of $\mc{L}(\mc{Q},D)$, when every other variable is independently true with probability $1/2$.
Since the blocks share no variables, this probability is obtained from: (A) the block $\Phi_{\bar{a}}$, where $X_\tau$ occurs; and (B) the other blocks.

For (A), we compute the probability that $X_\tau$ changes the value of $\Phi_{\bar{a}}$, denoted with $P_{\text{A}}$. \ Now, if $\tau = S_j(\bar{a})$, $X_\tau$ changes the value of $\Phi_{\bar{a}}$ iff the $X$-variables of the other $h-1$ tuples $S_{j^\prime}(\bar{a})$ are true (probability $2^{1-h}$), and $\varphi^{R}_{\bar{a}}$ is true (probability $1 - P_0(\varphi^{R}_{\bar{a}})$).

If $\tau = R_l(\bar{a},\bar{b})$, $X_\tau$ changes the value of $\Phi_{\bar{a}}$ iff $\psi^{S}_{\bar{a}}$ and the $X$-variables of the other $k-1$ tuples of its sub-block are true (probability $2^{-h} \times 2^{1-k} = 2^{1-h}(1 - P_0(\psi^{R}))$), and the other $n_{\bar{a}} - 1$ sub-blocks of $\bar{a}$ are false (probability $P_0(\psi^{R})^{n_{\bar{a}} - 1} = P_0(\varphi^{R}_{\bar{a}})/P_0(\psi^{R})$).\
In both cases, $P_{\text{A}}$ is the product of these probabilities, since the propositional formulas involved share no variables.

For (B), we compute the probability that the disjunction of the other blocks is false, denoted with $P_{\text{B}}$, that is, $P_{\text{B}} = P_0/P_0(\Phi_{\bar{a}})$. \ It holds: \ $\CE(\tau) = P_{\text{A}} \times P_{\text{B}}$.

We have obtained:
\begin{align}
    \CE(S_j(\bar{a})) &\ = \ 2^{1-h} \times \big(1 - P_0(\varphi^{R}_{\bar{a}})\big) \times \frac{P_0}{P_0(\Phi_{\bar{a}})}, \label{eq:ces_S_hier2} \\
    \CE(R_l(\bar{a},\bar{b})) &\ = \ 2^{1-h} \times \frac{1 - P_0(\psi^{R})}{P_0(\psi^{R})} \times P_0(\varphi^{R}_{\bar{a}}) \times \frac{P_0}{P_0(\Phi_{\bar{a}})}. \label{eq:ces_R_hier2}
\end{align}
We can see that the {\sf CES} of a tuple only depends on whether it is a tuple of the forms $S_j(\bar{a})$ or $R_l(\bar{a},\bar{b})$, and on its block.

We proceed now to compare the scores in (\ref{eq:ces_S_hier2}) and (\ref{eq:ces_R_hier2}) in three cases, (i)-(iii), from which (a) and (b) will follow. Let $\bar{a} \neq \bar{a}^\prime$ be two blocks.

\vspace{1mm} \noindent (i) \ In the same block, $\CE(S_j(\bar{a})) \geq \CE(R_l(\bar{a},\bar{b}))$; with equality iff $n_{\bar{a}} = 1$. \ In fact, by (\ref{eq:ces_S_hier2}) and (\ref{eq:ces_R_hier2}):
\begin{align*}
        &\CE(S_j(\bar{a})) - \CE(R_l(\bar{a},\bar{b})) > 0 \\
    \text{iff} \quad & 1 - P_0(\varphi^{R}_{\bar{a}}) - \frac{1 - P_0(\psi^{R})}{P_0(\psi^{R})}\times P_0(\varphi^{R}_{\bar{a}}) > 0 \\
    \text{iff} \quad & P_0(\psi^{R}) - P_0(\psi^{R})\times P_0(\varphi^{R}_{\bar{a}}) - P_0(\varphi^{R}_{\bar{a}}) + P_0(\psi^{R})\times P_0(\varphi^{R}_{\bar{a}}) > 0 \\
    \text{iff} \quad &  \ P_0(\psi^{R}) - P_0(\varphi^{R}_{\bar{a}}) > 0 \\
    \text{iff} \quad & \ P_0(\psi^{R}) > P_0(\psi^{R})^{n_{\bar{a}}} \\
    \text{iff} \quad & \ n_{\bar{a}} \geq 2.
    \vspace{-6mm}
\end{align*}

\vspace{1mm} \noindent (ii) \ For the tuples $S_j(\bar{a})$, more sub-blocks (for the particular block) means higher score: $\CE(S_j(\bar{a})) \geq \CE(S_{j'}(\bar{a}^\prime))$ iff $n_{\bar{a}} \geq n_{\bar{a}^\prime}$. \ In fact, by (\ref{eq:ces_S_hier2}) and the expression for $P_0(\Phi_{\bar{a}})$:

\begin{align*}
                        &  \CE(S_j(\bar{a})) - \CE(S_{j'}(\bar{a}^\prime)) \geq 0\\
    \text{iff} \quad    & \frac{1 - P_0(\varphi^{R}_{\bar{a}})}{P_0(\Phi_{\bar{a}})} - \frac{1 - P_0(\varphi^{R}_{\bar{a}^\prime})}{P_0(\Phi_{\bar{a}^\prime})} \geq 0 \\
    \text{iff} \quad    & big(1 - P_0(\varphi^{R}_{\bar{a}})\big)\times P_0(\Phi_{\bar{a}^\prime}) - \big(1 - P_0(\varphi^{R}_{\bar{a}^\prime})\big)\times P_0(\Phi_{\bar{a}}) \geq 0 \\
    \text{iff} \quad    & P_0(\psi^{S})\times \big(P_0(\varphi^{R}_{\bar{a}^\prime}) - P_0(\varphi^{R}_{\bar{a}})\big) + \big(1 - P_0(\psi^{S})\big) \big(P_0(\varphi^{R}_{\bar{a}^\prime}) - P_0(\varphi^{R}_{\bar{a}})\big) \geq 0 \\
    \text{iff} \quad    & P_0(\varphi^{R}_{\bar{a}^\prime}) - P_0(\varphi^{R}_{\bar{a}}) \geq 0 \\
    \text{iff} \quad    &  P_0(\psi^{R})^{n_{\bar{a}^\prime}} \geq P_0(\psi^{R})^{n_{\bar{a}}} \\
    \text{iff} \quad    &  n_{\bar{a}} \geq n_{\bar{a}^\prime}
\end{align*}

\vspace{1mm} \noindent (iii) \ For the tuples $R_l(\bar{a},\bar{b})$, more sub-blocks (for a given block) means lower score: $\CE(R_l(\bar{a},\bar{b})) > \CE(R_{l'}(\bar{a}^\prime,\bar{b}^\prime))$ iff $n_{\bar{a}} < n_{\bar{a}^\prime}$. In fact, by (\ref{eq:ces_R_hier2}) and the expression for $P_0(\Phi_{\bar{a}})$:

\begin{align*}
                        & \CE(R_l(\bar{a},\bar{b})) - \CE(R_{l'}(\bar{a}^\prime,\bar{b}^\prime)) > 0 \\
    \text{iff} \quad    & \frac{P_0(\varphi^{R}_{\bar{a}})}{P_0(\Phi_{\bar{a}})} - \frac{P_0(\varphi^{R}_{\bar{a}^\prime})}{P_0(\Phi_{\bar{a}^\prime})} > 0 \\
    \text{iff} \quad    & P_0(\varphi^{R}_{\bar{a}})\times P_0(\Phi_{\bar{a}^\prime}) - P_0(\varphi^{R}_{\bar{a}^\prime})\times P_0(\Phi_{\bar{a}}) > 0 \\
    \text{iff} \quad    & P_0(\psi^{S})\times \big(P_0(\varphi^{R}_{\bar{a}}) - P_0(\varphi^{R}_{\bar{a}^\prime})\big) > 0 \\
    \text{iff} \quad    & P_0(\psi^{R})^{n_{\bar{a}}} > P_0(\psi^{R})^{n_{\bar{a}^\prime}} \\
    \text{iff} \quad    & n_{\bar{a}} < n_{\bar{a}^\prime}
\end{align*}

\vspace{1mm}Finally, we verify (a) and (b) above. \  For $\tau^\prime = R_{l'}(\bar{a}^\prime,\bar{b}^\prime)$, with $n_{\bar{a}^\prime} \geq 2$:
\begin{itemize}
    \item[(a)] Let $\tau = S_j(\bar{a})$. If $n_{\bar{a}} \geq n_{\bar{a}^\prime}$, then by (ii) (if $\bar{a} \neq \bar{a}^\prime$) and (i), $\CE(\tau) \geq \CE(S_j(\bar{a}^\prime)) > \CE(\tau^\prime)$, strictly since $n_{\bar{a}^\prime} \geq 2$. If $n_{\bar{a}} < n_{\bar{a}^\prime}$, then by (i) and (iii), $\CE(\tau) \geq \CE(R_l(\bar{a},\bar{b})) > \CE(\tau^\prime)$, for any sub-block $(\bar{a},\bar{b})$.
    \item[(b)] Let $\tau = R_l(\bar{a},\bar{b})$ with $n_{\bar{a}} < n_{\bar{a}^\prime}$ (by hypothesis from (\ref{eq:goal_hier2})). Then, $\CE(\tau) > \CE(\tau^\prime)$ by (iii).
\end{itemize}
We have verified that (\ref{eq:goal_hier2}) holds; and we obtain that {\sf CES} and {\sf Resp} are aligned for $(\mc{Q},D)$, for every instance $D$ without exogenous tuples.
\boxtheorem

Example \ref{ex:ces_resp_hier2_no-exo} shows something worth noticing: The existence of exogenous tuples changes the alignment results for $\mc{Q}_{a}$ and $\mc{Q}_{b}$, namely, there is a counterexample with exogenous tuples for each of those queries on which {\sf CES} and {\sf Resp} are {\em not aligned}. However, no such instance exists when there are no exogenous tuples.

Also notice that Proposition \ref{prop:ces_resp_hier2_no-exo} is an alignment result; and uses the same proof technique as the alignment results in Section \ref{sub:ces-resp-aligned}, that is, use the lineage to compute the scores for the tuples in it.\
However, the technique used to prove non-alignment results in Section \ref{sub:ces-resp-non-aligned} cannot be transferred to the case where there are no exogenous tuples. Actually, for a given class of queries $\mc{C}$ and the simplest query $\mc{Q}$ of such class, Remark \ref{rem:non-alignment-technique}(c) relies on the existence of exogenous tuples to reconstruct $\mc{Q}$ from an  arbitrary query in $\mc{C}$.

A different strategy must be used to deal with instances without exogenous tuples. \
In spite of this, we conjecture that the class of queries in Proposition \ref{prop:ces_resp_hier2_no-exo} is dychotomic for single component queries: For any single-component query outside the class,  there exists a counterexample, without exogenous tuples, such the scores are not aligned.

The following example provides  counterexamples for queries not satisfying the conditions of Proposition \ref{prop:ces_resp_hier2_no-exo}: A non-hierarchical query, and a hierarchical query; in both cases, with their variables contained in three different sets of atoms. In both cases, the instance has no exogenous tuples.

\begin{example} \label{ex:outside_hier2}
    \emph{Non-hierarchical.} \ Consider the non-hierarchical SJF-BCQ \ \ $\mc{Q}_{c}\!: \ \exists x \exists y (R(x,y) \land S(x) \land T(y))$. \  Let $D_1$ be the instance in Table~\ref{tab:counter_nonhier}, with $D_1^\nit{ex} = \emptyset$.\
    It holds: \ $\tau_8 \prec^\CE \! \tau_9$ and $\tau_9 \prec^{\nit{Resp}} \! \tau_8$. \ {\sf CES} and {\sf Resp} are not aligned for $(\mc{Q}_{c},D_1)$.

\begin{table}[h]
    \centering
    {\scriptsize
    $\begin{array}{l|c|c|}
        \hline
        R & ~X~ & ~Y~\\
        \hline
        \tau_{1} & a & d\\
        \tau_{2} & a & e\\
        \tau_{3} & b & d\\
        \tau_{4} & b & f\\
        \tau_{5} & c & g\\
        \cline{2-3}
    \end{array}$~~~~~~~$\begin{array}{l|c|}
        \hline
        S & ~X~\\
        \hline
        \tau_{6} & a\\
        \tau_{7} & b\\
        \tau_{8} & c\\
        \cline{2-2}
    \end{array}$~~~~~~~$\begin{array}{l|c|}
        \hline
        T & ~Y~\\
        \hline
        \tau_{9} & d\\
        \tau_{10} & e\\
        \tau_{11} & f\\
        \tau_{12} & g\\
        \cline{2-2}
    \end{array}$
    ~~~~~~~~~~
    $\begin{array}{l|c|c|}
        \hline
        \tau & ~~{\sf CES}~~ & ~~{\sf Resp}~~\\
        \hline
        \tau_{8} & ~~0.1548~~ & ~~1/3~~\\
        \tau_{9} & ~~0.2563~~ & ~~1/4~~\\
        \cline{2-3}
    \end{array}$
    }
  \vspace{2mm}  \caption{\hspace*{1cm} (a) \ \ Instance $D_1$. \hspace{1cm} (b) \ \ Scores on $D_1$.}
    \label{tab:counter_nonhier}
\end{table}

    \vspace{1mm}\noindent
    \emph{Three variables.} \ Consider the hierarchical SJF-BCQ \ $\mc{Q}\!: \ \exists x \exists y \exists z (R(x,y,z) \land S(x,y) \land T(x))$. Its variables are contained in three different sets of atoms: $\{R,S,T\}$ for $x$, $\{R,S\}$ for $y$, and $\{R\}$ for $z$. Let $D_2$ be the instance in Table~\ref{tab:counter_hier3}, with $D_2^\nit{ex} = \emptyset$.

    \begin{table}[h]
    \centering
    {\scriptsize
    $\begin{array}{l|c|c|c|}
        \hline
        R & ~X~ & ~Y~ & ~Z~\\
        \hline
        \tau_{1} & a & c & d\\
        \tau_{2} & b & c & d\\
        \tau_{3} & b & e & d\\
        \tau_{4} & b & e & f\\
        \cline{2-4}
    \end{array}$~~~~~~~$\begin{array}{l|c|c|}
        \hline
        S & ~X~ & ~Y~\\
        \hline
        \tau_{5} & a & c\\
        \tau_{6} & b & c\\
        \tau_{7} & b & e\\
        \cline{2-3}
    \end{array}$~~~~~~~$\begin{array}{l|c|}
        \hline
        T & ~X~\\
        \hline
        \tau_{8} & a\\
        \tau_{9} & b\\
        \cline{2-2}
    \end{array}$
    ~~~~~~~~~~
    $\begin{array}{l|c|c|}
        \hline
        \tau & ~~{\sf CES}~~ & ~~{\sf Resp}~~\\
        \hline
        \tau_{7} & ~~0.2461~~ & ~~1/3~~\\
        \tau_{8} & ~~0.1836~~ & ~~1/2~~\\
        \cline{2-3}
    \end{array}$
    }
   \vspace{2mm} \caption{\hspace*{1cm}(a) \ \ Instance $D_2$. \hspace{1cm} (b) \ \ Scores on $D_2$.}
    \label{tab:counter_hier3}
\end{table}

    It holds: \ $\tau_8 \prec^\CE \! \tau_7$ and $\tau_7 \prec^{\nit{Resp}} \! \tau_8$. \ Then, {\sf CES} and {\sf Resp} are not aligned for $(\mc{Q},D_2)$. \boxtheorem
\end{example}

As a final remark for this Section, we recall that Example \ref{ex:prop16} showed that the inverse of Corollary \ref{coro:multi_single_alignment} does not hold for a particular instance with  exogenous tuples. The following example shows that the inverse still does not hold without exogenous tuples.

\begin{table}[h]
    \centering
    {\scriptsize
    $\begin{array}{l|c|c|}
        \hline
        R & ~X~ & ~Y~\\
        \hline
        \tau_{1} & c & d\\
        \tau_{2} & c & a\\
        \cline{2-3}
    \end{array}$~~~~~~~$\begin{array}{l|c|}
        \hline
        S & ~X~\\
        \hline
        \tau_{3} & c\\
        \cline{2-2}
    \end{array}$~~~~~~~$\begin{array}{l|c|c|}
        \hline
        T & ~Z~ & ~W~\\
        \hline
        \tau_{4} & a & b\\
        \tau_{5} & a & c\\
        \tau_{6} & a & d\\
        \tau_{7} & b & a\\
        \tau_{8} & d & c\\
        \cline{2-3}
    \end{array}$~~~~~~~$\begin{array}{l|c|}
        \hline
        U & ~Z~\\
        \hline
        \tau_{9} & a\\
        \tau_{10} & b\\
        \tau_{11} & d\\
        \cline{2-2}
    \end{array}$\qquad\qquad
    $\begin{array}{l|c|c|}
        \hline
        \tau & ~~{\sf CES}~~ & ~~{\sf Resp}~~\\
        \hline
        \tau_{1} & ~~0.1709~~ & ~~1/2~~\\
        \tau_{9} & ~~0.1846~~ & ~~1/3~~\\
        \cline{2-3}
    \end{array}$
    }
    \vspace{2mm}\caption{\hspace*{1cm}(a) \ \ Instance $D_{\sf RS}$. \hspace{1cm} (b) \ \ Scores on $D_{\sf RS}$.}
    \label{tab:counter_q2compex}
\end{table}

\begin{example} \label{ex:q2comp_counterex}
    For $\mc{Q}\!: \ \exists x\exists y\exists z \exists w (R(x,y) \land S(x) \land T(z,w) \land U(z))$, and instance $D$ in Table \ref{tab:counter_q2compex}, with $D^\nit{ex} = \emptyset$, {\sf CES} and {\sf Resp} are not aligned for $(\mc{Q},D)$: \ $\tau_1 \prec^\CE \! \tau_9$ and $\tau_9 \prec^{\nit{Resp}} \! \tau_1$.

    $\mc{Q}$ has the subqueries associated to its components: $\mc{Q}_1\!: \exists x \exists y (R(x,y) \land S(x))$ \ and $\mc{Q}_2\!: \exists z \exists w (T(z,w) \land U(z))$. \ By Proposition \ref{prop:ces_resp_hier2_no-exo}, {\sf CES} and {\sf Resp} are aligned for both $(\mc{Q}_1,D^\prime_1)$ and $(\mc{Q}_2,D^\prime_2)$, for every $D^\prime_1$ and $D^\prime_2$ without exogenous tuples. \boxtheorem
\end{example}


\section{{\sf CES} and {\sf Shapley} Alignment}\label{sec:ces-shap}\vspace{2mm}

Mirroring Section \ref{sec:ces-resp}, in this section we investigate the alignment of {\sf CES} and {\sf Shapley}, with a BCQ $\mc{Q}$ as our main parameter. \emph{As in Section \ref{sec:ces-resp}, we allow instances to contain exogenous tuples.}

The main results of this section are Propositions \ref{prop:cs-not-aligned-qrs}, \ref{prop:cs-not-aligned-qrst} and \ref{prop:cs-not-aligned-multicomp}, which provide three broad classes of SJF-BCQ for which  {\sf CES} and {\sf Shapley} are not aligned.
\ We do not present positive, i.e. alignment results, in this section.


In order to prove these results, we use the same technique used in Section \ref{sub:ces-resp-non-aligned} (see Remark \ref{rem:non-alignment-technique}): \ Find a counterexample $D$ for the simplest query $\mc{Q}$ of a particular class $\mc{C}$. Next, for any other query in the class, build an instance from $D$ (with exogenous tuples) and $\mc{Q}$ .
\ The structure of this section is also similar to that of Section \ref{sub:ces-resp-non-aligned}: For each class of queries, we state the proposition. Then, if required, a new counterexample for the simplest query of the class; and finally, the proof of the proposition.

\vspace{1mm}
\begin{proposition}  \label{prop:cs-not-aligned-qrs} \em
    Let $\mc{Q}$ be a single-component SJF-BCQ with at least two atoms, say predicates $R$ and $S$, such that $\nit{V\!ar}(R) \smallsetminus \nit{V\!ar}(S) \neq \emptyset$,  $\nit{V\!ar}(S) \smallsetminus \nit{V\!ar}(R) \neq \emptyset$,  and $\nit{V\!ar}(R) \cap \nit{V\!ar}(S) \neq \emptyset$. \
    There is an instance $D$ with exogenous tuples where {\sf CES} and {\sf Shapley} are not aligned.
\end{proposition}

\vspace{2mm}
\noindent {\bf Proof:}
Let $A_R$ and $A_S$ be the  atoms of $\mc{Q}$ with predicates $R$ and $S$, respectively. By the conditions on the variables, there are variables $x \in \nit{V\!ar}(A_R) \cap \nit{V\!ar}(A_S)$, \ $y \in \nit{V\!ar}(A_R) \smallsetminus \nit{V\!ar}(A_S)$, and $z \in \nit{V\!ar}(A_S) \smallsetminus \nit{V\!ar}(A_R)$.\
We build an instance $D$ starting from the tuples in the instance of Example \ref{ex:comp_ces_shapley}, that already provides us with a simplest query, which we denote by $D_{\sf RS}$ and has only endogenous tuples, as follows:

\vspace{1mm} \noindent (a) For each tuple $\tau_R$ of relation $R$ in $D_{\sf RS}$, create a tuple for $A_R$ by putting in the positions of $x$ and $y$ the values of $x$ and $y$ in $\tau_R$; and replacing by a constant $c^\prime$ the rest of the variables.
\ Similarly, for each tuple $\tau_S$ of relation $S$ in $D_{\sf RS}$, create a tuple for $A_S$ by putting in the positions of $x$ and $z$ the values of $x$ and $z$ in $\tau_S$; and replacing by $c^\prime$ the rest of the variables. All tuples created in this way are endogenous.

\vspace{1mm} \noindent (b) For each atom $U \in \nit{Atoms}(\mc{Q})$, different from $A_R$ and $A_S$, and for each combination of constants of $D_{\sf RS}$ for the variables of $U$ among $x$, $y$ and $z$, create a tuple for $U$ by putting those constants in the positions of the corresponding variables; and replacing with $c^\prime$ the rest of the variables. If $U$ contains none of $x$, $y$, $z$, a single tuple is created, with $c^\prime$ in every position. All tuples created in this way are exogenous.

\vspace{1mm}
The constant $c^\prime$ does not appear in $D_{\sf RS}$.
\ By construction, every variable of $\mc{Q}$ other than $x$, $y$, $z$ takes the value $c^\prime$ in all the tuples of $D$. The atoms in (b) are satisfied by every combination of values of $D_{\sf RS}$ for $x$, $y$, $z$. Then, $\mc{Q}$ is true in a subinstance of $D$ that contains $D^\nit{ex}$ iff it contains a tuple from $A_R$ and a tuple from $A_S$ with the same value for $x$, that is, iff $\mc{Q}_{\sf RS}$ is true in the corresponding subinstance of $D_{\sf RS}$.

Accordingly, there is a bijection $f$ from the tuples of $D_{\sf RS}$ to the endogenous tuples of $D$, such that, for every set $S$ of tuples of $D_{\sf RS}$, $\mc{Q}_{\sf RS}[S] = \mc{Q}[f(S) \cup D^\nit{ex}]$. \ By Proposition \ref{prop:same-values-lineage}, $(\mc{Q}_{\sf RS},D_{\sf RS})$ and $(\mc{Q},D)$ have the same lineage, and the corresponding tuples have the same {\sf CES} and {\sf Shapley} scores. It follows that $f(\tau_6) \prec^{\CE} f(\tau_4)$ and $f(\tau_4) \prec^\nit{Sh} f(\tau_6)$, and thus, the scores are not aligned for $(\mc{Q},D)$.
\boxtheorem

\begin{proposition}  \label{prop:cs-not-aligned-qrst} \em
  \   Let $\mc{Q}$ be a non-hierarchical (see Definition \ref{def:hierarchical}), single-component SJF-BCQ. \
    There is an instance $D$ with exogenous tuples where {\sf CES} and {\sf Shapley} are not aligned. \boxtheorem
\end{proposition}

\begin{example} \label{ex:qrst_counter}
    Consider the non-hierarchical query $\mc{Q}_{\sf RST}\!: \exists x \exists y \, (R(x) \land S(x,y) \land T(y))$, and the instance $D_{\sf RST}$ in Table~\ref{tab:qrst_counter}(a), with $D_{\sf RST}^\nit{ex} = \{\tau_7, \tau_{12}, \tau_{13}\}$. \

    \begin{table}[h]
    \centering
    {\scriptsize
    $\begin{tabu}{l|c|}
        \hline
        R & ~X~\\
        \hline
        \tau_{1} & a\\
        \tau_{2} & b\\
        \tau_{3} & c\\
        \tau_{4} & d\\
        \cline{2-2}
    \end{tabu}$~~~~~~~$\begin{tabu}{l|c|c|}
        \hline
        S & ~X~ & ~Y~\\
        \hline
        \tau_{5} & a & a\\
        \tau_{6} & a & c\\
        \underline{\tau_{7}} & a & d\\
        \tau_{8} & b & a\\
        \tau_{9} & c & c\\
        \tau_{10} & c & d\\
        \tau_{11} & d & a\\
        \underline{\tau_{12}} & d & c\\
        \cline{2-3}
    \end{tabu}$~~~~~~~$\begin{tabu}{l|c|}
        \hline
        T & ~Y~\\
        \hline
        \underline{\tau_{13}} & a\\
        \tau_{14} & c\\
        \tau_{15} & d\\
        \cline{2-2}
    \end{tabu}$
    ~~~~~~~~~~~~~~~~
    $\begin{tabu}{l|c|c|}
        \hline
        \tau & ~~{\sf CES}~~ & ~~{\sf Shapley}~~\\
        \hline
        \tau_{3} & ~~0.0996~~ & ~~0.0500~~\\
        \tau_{5} & ~~0.0879~~ & ~~0.0667~~\\
        \cline{2-3}
    \end{tabu}$
    }
  \vspace{2mm}  \caption{\hspace*{1cm}(a) \ \ Instance $D_{\sf RST}$. \hspace{1cm} (b) \ \ Scores on $D_{\sf RST}$.}
    \label{tab:qrst_counter}
\end{table}

    As shown in Table~\ref{tab:qrst_counter}(b), $\tau_5 \prec^{\CE} \tau_3$ and $\tau_3 \prec^{\nit{Sh}} \tau_5$. \ {\sf CES} and {\sf Shapley} are not aligned for $(\mc{Q}_{\sf RST},D_{\sf RST})$. \boxtheorem
\end{example}

\noindent {\bf Proof of Proposition \ref{prop:cs-not-aligned-qrst}:} \
Since $\mc{Q}$ is non-hierarchical, there are variables $x$ and $y$, and atoms $A_R \in \nit{Atoms}(x) \smallsetminus \nit{Atoms}(y)$, \ $A_S \in \nit{Atoms}(x) \cap \nit{Atoms}(y)$, and $A_T \in \nit{Atoms}(y) \smallsetminus \nit{Atoms}(x)$.\
We build an instance $D$ starting from the tuples in the instance $D_{\sf RST}$ in Example \ref{ex:qrst_counter}, as follows:

\vspace{1mm} \noindent (a) \ For each tuple $\tau_R$ of relation $R$ in $D_{\sf RST}$, create a tuple from $A_R$ by putting in the positions of $x$ the value of $x$ in $\tau_R$; and replacing by a constant $c^\prime$ the rest of the variables. \
Similarly, for each tuple $\tau_T$ of relation $T$ of $D_{\sf RST}$, we create a tuple from $A_T$ by putting in the positions of $y$ the value of $y$ in $\tau_T$. Also,  for each tuple $\tau_S$ of relation $S$ in $D_{\sf RST}$, create a tuple from $A_S$ by putting in the positions of $x$ and $y$ the values of $x$ and $y$ in $\tau_S$. In both cases, replace by $c^\prime$ the rest of the variables. A tuple created in this way is endogenous iff the corresponding tuple in $D_{\sf RST}$ is endogenous.

\vspace{1mm} \noindent (b) \ For each atom $U \in \nit{Atoms}(\mc{Q})$, different from $A_R$, $A_S$ and $A_T$, and for each combination of constants of $D_{\sf RST}$ for the variables of $U$ among $x$ and $y$, create a tuple from $U$ by putting those constants in the positions of the corresponding variables, and replacing by $c^\prime$ the rest of the variables. If $U$ contains neither $x$ nor $y$, a single tuple is created, with $c^\prime$ in every position. All tuples created in this way are exogenous.

\vspace{2mm}
Constant $c^\prime$ does not appear in $D_{\sf RST}$. By construction, every variable of $\mc{Q}$, other than $x$ and $y$, takes the value $c^\prime$ in all the tuples of $D$. \ The atoms in (b) are satisfied by every combination of values in $D_{\sf RST}$ for $x$ and $y$. \ Then, $\mc{Q}$ is true in a subinstance of $D$ that contains $D^\nit{ex}$ iff it contains a tuple from $A_R$, a tuple from $A_S$, and a tuple from $A_T$ that coincide in the values for $x$ and $y$. That is, iff $\mc{Q}_{\sf RST}$ is true in the corresponding subinstance of $D_{\sf RST}$.

We obtain that there is a bijection $f$ from the endogenous tuples of $D_{\sf RST}$ to the endogenous tuples of $D$, such that, for every set $S$ of endogenous tuples, $\mc{Q}_{\sf RST}[S \cup D_{\sf RST}^\nit{ex}] = \mc{Q}[f(S) \cup D^\nit{ex}]$. \ By Proposition \ref{prop:same-values-lineage}, $(\mc{Q}_{\sf RST},D_{\sf RST})$ and $(\mc{Q},D)$ have the same lineage, and the corresponding tuples have the same {\sf CES} and {\sf Shapley} scores. It follows that $f(\tau_5) \prec^{\CE} f(\tau_3)$ and $f(\tau_3) \prec^\nit{Sh} f(\tau_5)$, and thus, the scores are not aligned for $(\mc{Q},D)$.
\boxtheorem

\begin{proposition}  \label{prop:cs-not-aligned-multicomp} \em
    Let $\mc{Q}$ be a SJF-BCQ with at least two components, of which at least one has two or more atoms. There is an instance $D$ with exogenous tuples where {\sf CES} and {\sf Shapley} are not aligned. \boxtheorem
\end{proposition}

\begin{table}[h]
    \centering
    {\scriptsize
    $\begin{tabu}{l|c|}
        \hline
        R & ~X~\\
        \hline
        \tau_{1} & a\\
        \tau_{2} & b\\
        \tau_{3} & c\\
        \tau_{4} & d\\
        \cline{2-2}
    \end{tabu}$~~~~~~~$\begin{tabu}{l|c|}
        \hline
        S & ~Y~\\
        \hline
        \underline{\tau_{5}} & a\\
        \underline{\tau_{6}} & b\\
        \tau_{7} & c\\
        \cline{2-2}
    \end{tabu}$~~~~~~~$\begin{tabu}{l|c|}
        \hline
        T & ~Y~\\
        \hline
        \tau_{8} & a\\
        \tau_{9} & b\\
        \tau_{10} & c\\
        \cline{2-2}
    \end{tabu}$
    ~~~~~~~~~~
    $\begin{tabu}{l|c|c|}
        \hline
        \tau & ~~{\sf CES}~~ & ~~{\sf Shapley}~~\\
        \hline
        \tau_{1} & ~~0.1016~~ & ~~0.0893~~\\
        \tau_{7} & ~~0.1172~~ & ~~0.0655~~\\
        \cline{2-3}
    \end{tabu}$
    }
  \vspace{2mm}  \caption{\hspace*{1cm}(a) \ \ Instance $D_{2C}$. \ \hspace{1cm} (b) \ \ Scores on $D_{2C}$.}
    \label{tab:q2c-shapley}
\end{table}

\begin{example} \label{ex:q2c-shapley}
    For $\mc{Q}_{2C}: \! \exists x \exists y \, (R(x) \land S(y) \land T(y))$ and $D_{2C}$ the instance in Table~\ref{tab:q2c-shapley}, with $D^{\,\nit{ex}}_{2C} = \{\tau_5, \tau_6\}$, it holds: \ $D_{2C} \models \mc{Q}_{2C}$, \ $\tau_1 \prec^\CE \! \tau_7$ and $\tau_7 \prec^{\nit{Sh}} \! \tau_1$. \ That is, {\sf CES} and {\sf Shapley} are not aligned for $(\mc{Q}_{2C},D_{2C})$. \boxtheorem
\end{example}

\noindent \textbf{Proof of Proposition \ref{prop:cs-not-aligned-multicomp}: }
Let $C_1$ and $C_2$ be two components of $\mc{Q}$, such that $|C_2| \geq 2$.\
From $C_1$, select one atom, denoted with $A_R$, and a variable $x$ with $A_R \in \nit{Atoms}(x)$. \
Similarly, from $C_2$ select two atoms, $A_S$ and $A_T$, such that, for a variable $y$, $A_S, A_T \in \nit{Atoms}(y)$. They exist because $C_2$ is a component with at least two atoms.
We build an instance $D$ starting from the tuples in the instance $D_{2C}$ from Example \ref{ex:q2c-shapley}, as follows:

\vspace{1mm} \noindent (a) \ For each tuple $\tau_R$ from relation $R$ of $D_{2C}$, create a tuple from $A_R$ by putting in the positions of $x$ the value of $x$ in $\tau_R$, and replacing by a constant $c^\prime$ the rest of the variables. \ Similarly, for each tuple $\tau_S$ and $\tau_T$ of relations $S$ and $T$ in $D_{2C}$,  create a tuple from $A_S$ and $A_T$, respectively, by putting in the positions of $y$ the value of $y$ in $\tau_S$ and $\tau_T$, and replacing by $c^\prime$ the rest of the variables. A tuple created in this way is endogenous iff the corresponding tuple in $D_{2C}$ is endogenous.

\vspace{1mm} \noindent (b) \ For each atom $U \in \nit{Atoms}(\mc{Q})$, different from $A_R$, $A_S$ and $A_T$: \ If $U \in \nit{Atoms}(x)$, for each constant in $D_{2C}$, create a tuple from $U$ by putting that constant in the positions of $x$, and replacing by $c^\prime$ the rest of the variables. Similarly with $y$, if $U \in \nit{Atoms}(y)$. \ If $U \not\in (\nit{Atoms}(x) \cup \nit{Atoms}(y))$, a single tuple is created, by replacing all the variables in $U$ with $c^\prime$. All tuples created in this way are exogenous.

\vspace{1mm}
The constant $c^\prime$ does not appear in $D_{2C}$. Also, no atom contains both $x$ and $y$, since they belong to different components. \ By construction, every variable of $\mc{Q}$, other than $x$ and $y$, takes the value $c^\prime$ in all the tuples of $D$. Also, the atoms in (b) are satisfied by every combination of values of $D_{2C}$ for $x$ and $y$. Then, $\mc{Q}$ is true in a subinstance of $D$ that contains $D^\nit{ex}$ iff it contains a tuple from $A_R$, and a tuple from $A_S$ and a tuple from $A_T$ with the same value for $y$. That is, iff $\mc{Q}_{2C}$ is true in the corresponding subinstance of $D_{2C}$.

We obtain that there is a bijection $f$ from the endogenous tuples of $D_{2C}$ to the endogenous tuples of $D$, such that, for every set $S$ of endogenous tuples, $\mc{Q}_{2C}[S \cup D_{2C}^\nit{ex}] = \mc{Q}[f(S) \cup D^\nit{ex}]$. \ By Proposition \ref{prop:same-values-lineage}, $(\mc{Q}_{2C},D_{2C})$ and $(\mc{Q},D)$ have the same lineage; and the corresponding tuples have the same {\sf CES} and {\sf Shapley} scores. It follows that $f(\tau_1) \prec^\CE f(\tau_7)$ and $f(\tau_7) \prec^\nit{Sh} f(\tau_1)$. The scores are not aligned for $(\mc{Q},D)$.
\boxtheorem


The class of Proposition \ref{prop:cs-not-aligned-qrs} includes SJF-BCQs for which {\sf CES} and {\sf Resp} are aligned for every instance, with or without exogenous tuples. \
For example, consider the query $\mc{Q}_{\sf RS}$ in Equation (\ref{eq:qrs}): It satisfies the conditions of Theorem \ref{theo:ces_resp_aligned}(\ref{theo:ces_resp_aligned-a})(ii).\
However, $\mc{Q}_{\sf RS}$ also satisfies the conditions of Proposition \ref{prop:cs-not-aligned-qrs}. As a consequence, there is an instance such {\sf CES} and {\sf Shapley} are not aligned.

In contrast, consider the class of queries satisfying the conditions of Proposition \ref{prop:cs-not-aligned-qrst}: For every SJF-BCQ in the class, there exists a counterexample instance, such that {\sf CES} and {\sf Shapley} are not aligned. \ It turns out that, by Theorem \ref{theo:ces_resp_aligned}(\ref{theo:ces_resp_aligned-b}) (in particular, Proposition \ref{prop:cr-not-aligned-A}), there is also a counterexample instance for which {\sf CES} and {\sf Resp} are not aligned.

 Similarly for the class defined in Proposition \ref{prop:cs-not-aligned-multicomp}, there is a counterexample instance where {\sf CES} and {\sf Resp} are not aligned for every query satisfying its syntactic conditions.

We still have to identify and characterize  a class of SJF-BCQs for which the {\sf CES} and {\sf Shapley} are always aligned. \ Considering the results in this section, it is not clear if there is a dychotomic syntactic class of SJF-BCQs, that separates  the alignment of {\sf CES} and {\sf Shapley}. This remains as an open problem.


\section{{\sf CES} and {\sf Resp} Alignment for BCQs with Self-Joins} \label{sec:sj-bcq}\vspace{2mm}

Up until this point, we have restricted our analyzes of score alignment to SJF-BCQs. As discussed in Section \ref{sub:constants}, this  has allowed us to simplify the syntactic characterization of queries: Neither the occurrence of constants  nor the order and number of occurrences of the variables are relevant to determine alignment of two scores.

Turning our attention to  BCQs with self-joins, that is, queries where a same relational predicate appears in more than one atom, we cannot make similar simplifications: \ Constants in a query may affect its lineage, so as the order and number of occurrences of variables, which would end up affecting a  score.

\ignore{This issue is illustrated in the following example where the alignment of {\sf CES} and {\sf Resp} is analyzed.}

\begin{example} \label{ex:qp2-qc2}
    Consider these  queries, both of which exhibit a self-join:
    \begin{align*}
        \mc{Q}\! &: \ \exists x \exists y \exists z\, (R(x,y) \land R(z,y)), \\
        \mc{Q}^\prime\! &: \ \exists x \exists y \exists z\, (R(x,y) \land R(y,z)).
    \end{align*}

    If Theorem~\ref{theo:ces_resp_aligned} is applied to them, ignoring the fact that $R$ occurs twice, for both of them {\sf CES} and {\sf Resp} should be aligned on every compatible instance. This is because they are single-component and have two atoms (Theorem \ref{theo:ces_resp_aligned}(\ref{theo:ces_resp_aligned-a})(ii). \ This is actually the case for $\mc{Q}$, but not for $\mc{Q}^\prime$, as we show below.

    Query $\mc{Q}$ is equivalent to $\exists x \exists y\, R(x,y)$, which is an SJF-BCQ with one atom. By Theorem~\ref{theo:ces_resp_aligned}(a)(ii) applied to it, {\sf CES} and {\sf Resp} are aligned for $(\mc{Q},D)$, for every instance $D$, with or without exogenous tuples.

\begin{table}[h]
    \centering
    {\scriptsize
    $\begin{tabu}{l|c|c|}
        \hline
        R & ~A~ & ~B~\\
        \hline
        \tau_{1} & a & d\\
        \tau_{2} & b & a\\
        \tau_{3} & b & c\\
        \tau_{4} & b & d\\
        \tau_{5} & c & d\\
        \tau_{6} & d & a\\
        \tau_{7} & d & c\\
        \cline{2-3}
    \end{tabu}$~~~~~~~~~~~~~~~~~~~~~~~~~~~~~~~~~~~$\begin{tabu}{l|c|c|}
        \hline
        \tau & ~~{\sf CES}~~ & ~~{\sf Resp}~~\\
        \hline
        \tau_{4} & ~~0.1875~~ & ~~1/3~~\\
        \tau_{7} & ~~0.2188~~ & ~~1/4~~\\
        \cline{2-3}
    \end{tabu}$
    }
    \vspace{2mm}\caption{\hspace*{0.5cm}(a) \ \ Instance $D$ (all endogenous tuples).  \hspace{1cm} (b) \ \ Scores on $D$.}
    \label{tab:qp2-resp}
\end{table}

    Now, for $\mc{Q}^\prime$ and the instance $D$ in Table~\ref{tab:qp2-resp}, with $D^{\,\nit{ex}} = \emptyset$, it holds: \ $D \models \mc{Q}^\prime$, \ $\tau_4 \prec^\CE \! \tau_7$ and $\tau_7 \prec^{\nit{Re}} \! \tau_4$. \ {\sf CES} and {\sf Resp} are not aligned for $(\mc{Q}^\prime,D)$.

  Notice that  $\mc{Q}$ and $\mc{Q}^\prime$ have the same variables, and the same number of atoms; and they only differ in the order of the variables in the second atom. This is enough to change the alignment of {\sf CES} and {\sf Resp}.  \boxtheorem
\end{example}

In the rest of this section, we provide the following:
\begin{enumerate}[(A)]
    \item A class of BCQs with self-joins for which {\sf CES} and {\sf Resp} are always aligned (Proposition~\ref{prop:cr_aligned_permutation}).

    \item A class of BCQs, with or without self-joins, for which there is an database instance where {\sf CES} and {\sf Resp} are not aligned. (Proposition~\ref{prop:cr-not-aligned-cycles}).
\end{enumerate}

\begin{proposition}  \label{prop:cr_aligned_permutation} \em
    Let $\mc{Q}$ be a BCQ with only one relational predicate $R$, of arity $m$, whose atoms have  all the possible permutations of the variables $x_1, \ldots, x_m$, that is, of the form:
    \[
        \mc{Q}: \ \exists x_1 \cdots \exists x_m \bigwedge_{\pi \in S_m} R(x_{\pi(1)}, \ldots, x_{\pi(m)}),
    \]
    where $S_m$ is the set of all permutations of $\{1, \ldots, m\}$. \ For this query, {\sf CES} and {\sf Resp} are aligned for $(\mc{Q},D)$, for every instance $D$, with or without exogenous tuples.
\end{proposition}

\vspace{2mm}
\noindent {\bf Proof: \ } We show that all endogenous tuples have the same responsibility, which, by Remark \ref{rem:same-value}, implies that the scores are aligned.

For a sequence $\bar{c} = (c_1, \ldots, c_m)$ of constants (not necessarily different), and $\pi \in S_m$, \  $\bar{c}_\pi := (c_{\pi(1)}, \ldots,$ $ c_{\pi(m)})$. \ A valuation that maps $x_1, \ldots, x_m$ to $\bar{c}$, also  maps the atoms of $\mc{Q}$ to the tuples $R(\bar{c}_\pi)$, with $\pi \in S_m$. Accordingly, we call the set $B_{\bar{c}} := \{R(\bar{c}_\pi) \mid \pi \in S_m\}$ the \emph{block} of $\bar{c}$. \ Denote with $\mc{B}$ the set of blocks $B_{\bar{c}}$, with $B_{\bar{c}} \subseteq D$. Notice that a block may have fewer than $m!$ tuples, when $\bar{c}$ has repeated constants, e.g. \ $B_{(a,a)} = \{R(a,a)\}$, for $m = 2$.

Different blocks do not share tuples. Indeed, if $R(\bar{e}) \in B_{\bar{c}} \cap B_{\bar{d}}$, then $\bar{e}$ is a permutation of both $\bar{c}$ and $\bar{d}$.  So, $\bar{d}$ is a permutation of $\bar{c}$. Since a permutation of a permutation of $\bar{c}$ is again a permutation of $\bar{c}$, it follows that $B_{\bar{d}} \subseteq B_{\bar{c}}$. By symmetry, $B_{\bar{c}} = B_{\bar{d}}$. \ We are guaranteed to find all permutations of $\bar{c}$, because $S_m$ contains all the permutations of the variables, and $D$ has no dummy tuples.

The lineage of $\mc{Q}$ on $D$ takes the form:
\begin{equation}
    \mc{L}(\mc{Q},D) \ = \ \bigvee_{B \in \mc{B}} \Phi_{B},
    \qquad \text{with} \quad
    \Phi_{B} \ := \ \bigwedge_{\tau \in B} X_{\tau}. \label{eq:cr_lineage_perm}
\end{equation}

Each $B \in \mc{B}$ defines the disjunct $\Phi_B$, which we also call {\em a block}.  Different blocks do not share $X$-variables. \
A block without exogenous tuples has the form in (\ref{eq:cr_lineage_perm}), that is, no $X$-variables are removed (see Section \ref{sec:lin}). \
Moreover, no block has all its tuples exogenous, since, otherwise, $D^\nit{ex} \models \mc{Q}$ (see Remark \ref{rem:sat}). \
Accordingly, the ex-reduced lineage is obtained by removing from each block its exogenous variables:
\[
    \mc{L}(\mc{Q},D) \ = \ \bigvee_{B \in \mc{B}} \bigwedge_{\tau \in T_{B}} X_\tau,
\]
where $T_{B}$ is the set of endogenous tuples of block $\Phi_{B}$. It holds: $T_{B} \neq \emptyset$.

Now, let $\tau \in \Den$ be in block $\Phi_{B}$. In relation to {\sf Resp}, a minimum contingency set $\Gamma$ for $\tau$ must make all the other blocks false, which requires removing one tuple from each of them. No tuple of block $\Phi_{B}$ needs to be removed due to $\Gamma$ (otherwise, removing $\Gamma$ alone, without $\tau$, would falsify the query). \
Therefore, for every $\tau \in \Den$, $\Resp(D,\mc{Q},\tau) = 1/|\mc{B}|$, a constant. \ By Remark \ref{rem:same-value}, the scores are aligned for $(\mc{Q},D)$.\boxtheorem

We illustrate this proposition with an example.

\begin{example} \label{ex:qp2-qp3}
    The following queries satisfy the conditions of Proposition~\ref{prop:cr_aligned_permutation}, with $m = 2$ and $m = 3$, respectively:
    \begin{eqnarray*}
        \mc{Q}_{P_2}\!\!\!\!&:&\exists x \exists y\, (R(x,y) \land R(y,x)),\\
        \mc{Q}_{P_3}\!\!\!\!&:&\exists x \exists y \exists z\, (R(x,y,z) \land R(x,z,y) \land R(y,x,z) \land R(y,z,x) \land \\
        && \hspace*{1.5cm}R(z,x,y) \land R(z,y,x)).
    \end{eqnarray*}

    By Proposition~\ref{prop:cr_aligned_permutation}, {\sf CES} and {\sf Resp} are aligned for $(\mc{Q}_{P_2},D_2)$ and $(\mc{Q}_{P_3},D_3)$, for all instances $D_2, D_3$, with or without exogenous tuples. \boxtheorem
\end{example}

For the next result, we consider a family of BCQs with self-joins, all whose atoms have the same predicate $R$. For $n \geq 1$, a \ {\emph{cyclic query} \ is of the form:
\[
    \mc{Q}_{C_n}\!: \ \exists \bar{x}_1 \cdots \exists \bar{x}_{n+1} \, \left( R(\bar{x}_1,\bar{x}_2) \land R(\bar{x}_2,\bar{x}_3) \land \cdots \land R(\bar{x}_n,\bar{x}_{n+1})\right),
\]
where $\bar{x}_1,\ldots,\bar{x}_{n+1}$ are {\em disjoint} lists of variables of the same length, and $\bar{x}_{n+1}$ is either a list of new variables or $\bar{x}_1$.\
If $\bar{x}_{n+1} \neq \bar{x}_1$, the query is called \emph{open}, and denoted by $\mc{Q}^{\nit{op}}_{C_n}$. \ Otherwise,  if $\bar{x}_{n+1} = \bar{x}_1$, the query is called \emph{closed}, and denoted with $\mc{Q}^{\nit{cl}}_{C_n}$. \ For example,
\begin{equation} \label{eq:cycle-queries}
    \begin{aligned}
        \mc{Q}^{\nit{op}}_{C_2} \! &: \exists x \exists y \exists z \, (R(x,y) \land R(y,z)), \\
        \mc{Q}^{\nit{cl}}_{C_3} \! &: \exists x \exists y \exists z \, (R(x,y) \land R(y,z) \land R(z,x)), \ \text{ and} \\
        \mc{Q}^{\nit{cl}}_{C_4} \! &: \exists x \exists y \exists z \exists w \, (R(x,y) \land R(y,z) \land R(z,w) \land R(w,x)).
    \end{aligned}
\end{equation}
Here, each string of variables is a single variable.

\begin{proposition} \label{prop:cr-not-aligned-cycles} \em
    Let $\mc{Q}$ be a BCQ with self-joins whose atoms with predicate $R$ are exactly the atoms of a cycle query equal to either: \
    (a) \  $\mc{Q}^{\nit{op}}_{C_n}$, with $n \geq 2$; \ or  \ (b) \ $\mc{Q}^{\nit{cl}}_{C_n}$, with $n \geq 3$. \
    Then, there is an instance $D$ with exogenous tuples where  {\sf CES} and {\sf Resp} are not aligned. \boxtheorem
\end{proposition}

\begin{example} \label{ex:cycle_5}
    Consider the queries:
    \begin{equation}\label{eq:ex_cycle_5}
        \begin{aligned}
        \mc{Q}&: \exists x_1 \cdots \exists x_6\, (R(x_1,x_2) \land R(x_2,x_3) \land R(\bar{x}_3,x_4) \land R(x_4,x_5) \land R(x_5,x_6)) \\
        \mc{Q}^\prime&: \exists x_1 \cdots \exists x_6 \exists x\, (R(x_1,x_2) \land R(x_2,x_3) \land R(x_3,x_4) \land R(x_4,x_5) \land R(x_5,x_6) \land A(x)).
    \end{aligned}
    \end{equation}

    The query $\mc{Q}$ satisfies condition (a) of Proposition \ref{prop:cr-not-aligned-cycles}: its atoms are of those of the cyclic query $\mc{Q}^{\nit{op}}_{C_5}$, where each disjoint string of variables is a single variable.\
    Accordingly, there is an instance $D$ where {\sf CES} and {\sf Resp} are not aligned for $(\mc{Q},D)$.

    Moreover, the query $\mc{Q}^\prime$ also satisfies the conditions of the proposition: it has a subset of its atoms corresponding to the ones of $\mc{Q}^{\nit{op}}_{C_5}$. It follows that an instance $D^\prime$ also exists such that {\sf CES} and {\sf Resp} are not aligned for $(\mc{Q}^\prime,D^\prime)$. \boxtheorem
\end{example}

We will prove this proposition using the technique described in Remark \ref{rem:non-alignment-technique}: For a given class of queries $\mc{C}$, find a counterexample $(\mc{Q},D)$ for the simplest queries in $\mc{C}$. Then, from $D$, build an instance where the scores are not aligned for any other query belonging to $\mc{C}$.\
The simplest queries satisfying the conditions of Proposition \ref{prop:cr-not-aligned-cycles} are $\mc{Q}^{\nit{op}}_{C_2}$ and $\mc{Q}^{\nit{cl}}_{C_3}$ in (\ref{eq:cycle-queries}), for (a) and (b), respectively. However, in the proof, the counterexample for $\mc{Q}^{\nit{op}}_{C_2}$ is used to build the instances for all the other queries, open or closed, except for the closed ones with $n = 3$ or $n = 4$, which are built from the counterexamples for $\mc{Q}^{\nit{cl}}_{C_3}$ and $\mc{Q}^{\nit{cl}}_{C_4}$. \ First we give some examples.

\begin{table}[h]
    \centering
    {\scriptsize
    $\begin{array}{l|c|c|}
        \hline
        R & ~X~ & ~Y~\\
        \hline
        \tau_{1} & a & b\\
        \tau_{2} & a & c\\
        \tau_{3} & a & d\\
        \tau_{4} & b & a\\
        \tau_{5} & b & d\\
        \tau_{6} & c & a\\
        \tau_{7} & c & d\\
        \cline{2-3}
    \end{array}$
    ~~~~~~
    $\begin{array}{l|c|c|}
        \hline
        \tau & ~~{\sf CES}~~ & ~~{\sf Resp}~~\\
        \hline
        \tau_{3} & ~~0.1875~~ & ~~1/3~~\\
        \tau_{4} & ~~0.2188~~ & ~~1/4~~\\
        \cline{2-3}
    \end{array}$
    }
    \vspace{2mm}\caption{\hspace*{1cm} (a) \ \ Instance $D_{C_2}$. \ \hspace{1cm} (b) \ \ Scores on $D_{C_2}$.}
    \label{tab:cycle-open}
\end{table}

\begin{example} \label{ex:cycle-open} \
    Consider $\mc{Q}^{\nit{op}}_{C_2}\!: \ \exists x \exists y \exists z \, (R(x,y) \land R(y,z))$, and $D_{C_2}$, the instance in Table~\ref{tab:cycle-open}, with $D_{C_2}^\nit{ex} = \emptyset$.\
    It holds:\ $D_{C_2} \models \mc{Q}^{\nit{op}}_{C_2}$, \ $\tau_3 \prec^\CE \! \tau_4$ and $\tau_4 \prec^{\nit{Resp}} \! \tau_3$. \
    {\sf CES} and {\sf Resp} are not aligned for $(\mc{Q}^{\nit{op}}_{C_2},D_{C_2})$. \boxtheorem
\end{example}

\begin{table}[h]
    \centering
    {\scriptsize
    $\begin{array}{l|c|c|}
        \hline
        R & ~X~ & ~Y~\\
        \hline
        \tau_{1} & a & b\\
        \tau_{2} & a & c\\
        \tau_{3} & d & a\\
        \tau_{4} & e & f\\
        \tau_{5} & f & g\\
        \tau_{6} & g & e\\
        \underline{\tau_{7}} & b & d\\
        \underline{\tau_{8}} & c & d\\
        \cline{2-3}
    \end{array}$
    ~~~~~~
    $\begin{array}{l|c|c|}
        \hline
        \tau & ~~{\sf CES}~~ & ~~{\sf Resp}~~\\
        \hline
        \tau_{1} & ~~0.2188~~ & ~~1/3~~\\
        \tau_{4} & ~~0.1563~~ & ~~1/2~~\\
        \cline{2-3}
    \end{array}$
    }
    \vspace{2mm}\caption{\hspace*{1cm} (a) \ \ Instance $D_{C_3}$. \ \hspace{1cm} (b) \ \ Scores on $D_{C_3}$.}
    \label{tab:cycle-closed}
\end{table}

\begin{example} \label{ex:cycle-closed} \
    Consider \ $\mc{Q}^{\nit{cl}}_{C_3}\!: \ \exists x \exists y \exists z \, (R(x,y) \land R(y,z) \land R(z,x))$, and $D_{C_3}$, the instance in Table~\ref{tab:cycle-closed}, with $D_{C_3}^\nit{ex} = \{\tau_7, \tau_8\}$. \
    It holds: \ $D_{C_3} \models \mc{Q}^{\nit{cl}}_{C_3}$, \ $\tau_4 \prec^\CE \! \tau_1$, and $\tau_1 \prec^{\nit{Resp}} \! \tau_4$. \
    {\sf CES} and {\sf Resp} are not aligned for $(\mc{Q}^{\nit{cl}}_{C_3},D_{C_3})$. \boxtheorem
\end{example}

\begin{table}[h]
    \centering
    {\scriptsize
    $\begin{array}{l|c|c|}
        \hline
        R & ~X~ & ~Y~\\
        \hline
        \tau_{1} & a & b\\
        \tau_{2} & a & c\\
        \tau_{3} & d & a\\
        \tau_{4} & e & f\\
        \tau_{5} & f & g\\
        \tau_{6} & g & h\\
        \underline{\tau_{7}} & b & i\\
        \underline{\tau_{8}} & i & d\\
        \underline{\tau_{9}} & c & j\\
        \underline{\tau_{10}} & j & d\\
        \underline{\tau_{11}} & h & e\\
        \cline{2-3}
    \end{array}$
    ~~~~~~
    $\begin{array}{l|c|c|}
        \hline
        \tau & ~~{\sf CES}~~ & ~~{\sf Resp}~~\\
        \hline
        \tau_{1} & ~~0.2188~~ & ~~1/3~~\\
        \tau_{4} & ~~0.1563~~ & ~~1/2~~\\
        \cline{2-3}
    \end{array}$
    }
   \vspace{2mm} \caption{\hspace*{1cm} (a) \ \ Instance $D_{C_4}$. \ \hspace{1cm} (b) \ \ Scores on $D_{C_4}$.}
    \label{tab:cycle-closed4}
\end{table}

\begin{example} \label{ex:cycle-closed4} \
    Consider $\mc{Q}^{\nit{cl}}_{C_4}$ in Equation (\ref{eq:cycle-queries}), and $D_{C_4}$, the instance in Table~\ref{tab:cycle-closed4}, with $D_{C_4}^\nit{ex} = \{\tau_7, \ldots, \tau_{11}\}$.\
    It holds: \ $D_{C_4} \models \mc{Q}^{\nit{cl}}_{C_4}$, \ $\tau_4 \prec^\CE \! \tau_1$ and $\tau_1 \prec^{\nit{Resp}} \! \tau_4$. \
    {\sf CES} and {\sf Resp} are not aligned for $(\mc{Q}^{\nit{cl}}_{C_4},D_{C_4})$. \boxtheorem
\end{example}

\noindent \textbf{Proof of Proposition \ref{prop:cr-not-aligned-cycles}: }
Consider a query $\mc{Q}$ satisfying the conditions of the proposition, and let $R(\bar{x}_1,\bar{x}_2), \ldots, R(\bar{x}_n,\bar{x}_{n+1})$ be its atoms with predicate $R$, with lists of length $k \geq 1$. We assume that $\mc{Q}$ is not closed with $n = 3$ or $n = 4$; these two cases are treated separately at the end.\
We select the consecutive atoms $R(\bar{x}_1,\bar{x}_2)$ and $R(\bar{x}_2,\bar{x}_3)$, which play the roles of $R(x,y)$ and $R(y,z)$ in $\mc{Q}^{\nit{op}}_{C_2}$.

Now we proceed to build an instance $D$ starting from the tuples in the instance $D_{C_2}$ in Example \ref{ex:cycle-open}, as follows. For a constant $s$, denote with $\bar{s}$ the list of length $k$ with all its entries equal to $s$; and let $e_4, \ldots, e_{n+1}$ and $c^\prime$ be new constants.

\vspace{1mm} \noindent (a) \ For each tuple $R(a_1,a_2)$ of $D_{C_2}$, with $a_1, a_2 \in \nit{ad\!om}(D_{C_2})$, we create the tuple $R(\bar{a}_1,\bar{a}_2)$, which is endogenous iff $R(a_1,a_2)$ is endogenous.

\vspace{1mm} \noindent (b) \ For each atom $R(\bar{x}_i,\bar{x}_{i+1})$, with $i \geq 3$, we create tuples by replacing $\bar{x}_i$ with $\bar{s}$, for each constant $s \in \nit{ad\!om}(D_{C_2})$, if $i = 3$, and with $\bar{e}_i$ otherwise. Additionally, we replace $\bar{x}_{i+1}$ with $\bar{e}_{i+1}$, except for the last atom $R(\bar{x}_n,\bar{x}_1)$ of a closed query, where $\bar{x}_1$ is replaced with $\bar{s}$, for each constant $s \in \nit{ad\!om}(D_{C_2})$. All tuples created in this way are exogenous.

\vspace{1mm} \noindent (c) \ For each atom $U \in \nit{Atoms}(\mc{Q})$ with a predicate other than $R$, and for each path given by the atoms of $\mc{Q}$ with predicate $R$ in the tuples created in (a) and (b), we create a tuple from $U$ by replacing each variable of a list $\bar{x}_i$ with the constant that the path assigns to it, and replacing by the constant $c^\prime$ the rest of variables. All tuples created in this way are exogenous. Since their predicates are not $R$, they do not create new paths.

\vspace{1mm}By construction of $D$, in every tuple created in (a) and (b), all the positions of each list of variables have the same constant.\
If $\mc{Q}$ is open, the tuples created in (b) form a path of $n-2$ tuples through the constants $s, e_4, \ldots, e_{n+1}$, for each constant $s \in \nit{ad\!om}(D_{C_2})$.\
If $\mc{Q}$ is closed, those tuples form a path through $a_1, e_4, \ldots, e_n, a_2$, for each pair of constants $a_1, a_2 \in \nit{ad\!om}(D_{C_2})$.\
Given this, it follows that every path given by $\mc{Q}$ in $D$ contains two consecutive tuples created in (a), that is, a path given by $\mc{Q}^{\nit{op}}_{C_2}$ in $D_{C_2}$.

Then, there is a bijection $f$ from the endogenous tuples of $D_{C_2}$ to the endogenous tuples in $D$, given by $f(R(a_1,a_2)) = R(\bar{a}_1,\bar{a}_2)$, such that, for every set $S$ of endogenous tuples, $\mc{Q}^{\nit{op}}_{C_2}[S \cup D_{C_2}^\nit{ex}] = \mc{Q}[f(S) \cup D^\nit{ex}]$. By Proposition \ref{prop:same-values-lineage}, $(\mc{Q}^{\nit{op}}_{C_2},D_{C_2})$ and $(\mc{Q},D)$ have the same lineage, and the corresponding tuples have the same {\sf CES} and {\sf Resp} scores. It follows that $f(\tau_3) \prec^\CE f(\tau_4)$ and $f(\tau_4) \prec^\nit{Re} f(\tau_3)$, and thus, the scores are not aligned for $(\mc{Q},D)$.

In the particular cases where $\mc{Q}$ is closed with $n = 3$ or $n = 4$, the construction above does not work, since the tuples created in (b) give paths of $\mc{Q}$ with only exogenous tuples: for $n = 3$, the tuple $R(\bar{s},\bar{s})$, and, for $n = 4$, the cycle through $s, e_4, s$. In these cases, we start instead from the instance $D_{C_3}$ in Example \ref{ex:cycle-closed}, or $D_{C_4}$ in Example \ref{ex:cycle-closed4}, respectively, whose cycles already have $n$ tuples, and $D$ is built from it with items (a) and (c).\
As before, with $f$ now defined on the endogenous tuples of $D_{C_3}$ or $D_{C_4}$, $f(\tau_4) \prec^\CE f(\tau_1)$ and $f(\tau_1) \prec^\nit{Re} f(\tau_4)$ in both cases. It follows that the scores are not aligned for $(\mc{Q},D)$.
\boxtheorem

We illustrate the proof of Proposition \ref{prop:cr-not-aligned-cycles} with an example.

\begin{example}
    Consider the query $\mc{Q}^\prime$ in (\ref{eq:ex_cycle_5}) from Example \ref{ex:cycle_5}. We build an instance $D^\prime$ such {\sf CES} and {\sf Resp} are not aligned for $(\mc{Q}^\prime,D^\prime)$, according to the proof of Proposition \ref{prop:cr-not-aligned-cycles}.

    \begin{table}[h]
        \centering
        {\scriptsize
        $\begin{array}[t]{l|c|c|}
            \hline
            R & ~X~ & ~Y~\\
            \hline
            \tau_{1} & a & b\\
            \tau_{2} & a & c\\
            \tau_{3} & a & d\\
            \tau_{4} & b & a\\
            \tau_{5} & b & d\\
            \tau_{6} & c & a\\
            \tau_{7} & c & d\\
            \underline{\tau_{8}} & a & e_4\\
            \underline{\tau_{9}} & b & e_4\\
            \underline{\tau_{10}} & c & e_4\\
            \underline{\tau_{11}} & d & e_4\\
            \underline{\tau_{12}} & e_4 & e_5\\
            \underline{\tau_{13}} & e_5 & e_6\\
            \cline{2-3}
        \end{array}$~~~~~~$\begin{array}[t]{l|c|}
            \hline
            A & ~X~\\
            \hline
            \underline{\tau_{14}} & c^\prime\\
            \cline{2-2}
        \end{array}$~~~~~~~~~~~~$\begin{array}[t]{l|c|c|}
            \hline
            \tau & ~~{\sf CES}~~ & ~~{\sf Resp}~~\\
            \hline
            \tau_{3} & ~~0.1875~~ & ~~1/3~~\\
            \tau_{4} & ~~0.2188~~ & ~~1/4~~\\
            \cline{2-3}
        \end{array}$
        }
       \vspace{2mm} \caption{\hspace*{1cm} (a) \ \ Instance $D$. \ \hspace{1cm} (b) \ \ Scores on $D$.}
        \label{tab:ex-cycle-c5}
    \end{table}

    Let us start by selecting the atoms $R(x_1,x_2)$ and $R(x_2,x_3)$, to play the roles of $R(x,y)$ and $R(y,z)$ in $\mc{Q}^{\nit{op}}_{C_2}$, respectively. Then, proceed as follows:

    \begin{enumerate}[(a)]
        \item For each tuple $R(a_1,a_2)$ of $D_{C_2}$ in Table \ref{tab:cycle-open}, we create the tuple $R(a_1,a_2)$, since the lists of variables have length one. For example, from $\tau_1\! :\! R(a,b)$ we create $R(a,b)$.\
        The tuples created are the following: $R(a,b)$, $R(a,c)$, $R(a,d)$, $R(b,a)$, $R(b,d)$, $R(c,a)$ and $R(c,d)$.\
        All of these are endogenous, as those of $D_{C_2}$.

        \item For atom $R(x_3,x_4)$, we create a tuple for each constant $s \in \nit{ad\!om}(D_{C_2})$, by replacing $x_3$ with $s$, and $x_4$ with a fresh constant $e_4$. For atoms $R(x_4,x_5)$ and $R(x_5,x_6)$, we replace their variables with the fresh constants $e_4$, $e_5$ and $e_6$.\
        The tuples created are the following: $R(a,e_4)$, $R(b,e_4)$, $R(c,e_4)$, $R(d,e_4)$, $R(e_4,e_5)$ and $R(e_5,e_6)$. All of these are declared exogenous.\
        From each constant of $D_{C_2}$, they form a path of three tuples, which completes every path of $\mc{Q}^{\nit{op}}_{C_2}$ in $D_{C_2}$ to a path of $\mc{Q}^{\nit{op}}_{C_5}$. For example, the path $\tau_4\!\!:\! R(b,a)$, \ $\tau_3\!\!:\! R(a,d)$ is completed with $R(d,e_4)$, $R(e_4,e_5)$ and $R(e_5,e_6)$.

        \item Lastly, for atom $A(x)$, which has a predicate other than $R$ and none of the variables $x_1, \ldots, x_6$, we create the tuple $A(c^\prime)$, by replacing $x$ with the constant $c^\prime$. It is declared exogenous.
    \end{enumerate}

    The resulting instance is shown in Table \ref{tab:ex-cycle-c5}(a). Moreover, as shown in Table \ref{tab:ex-cycle-c5}(b), {\sf CES} and {\sf Resp} are not aligned for $(\mc{Q}^\prime,D^\prime)$. Notice that, compared to instance $D_{C_2}$, $D$ has seven (exogenous) tuples more: \ Four from $R(x_3,x_4)$, one from each of $R(x_4,x_5)$ and $R(x_5,x_6)$, and one from $A(x)$.\

    \begin{figure}
        \centering
        \begin{tikzpicture}[
            v/.style   = {circle, draw, inner sep=0pt, minimum size=5mm, font=\scriptsize},
            en/.style  = {-{Stealth[length=3pt]}, thin},
            ex/.style  = {en, dashed},
            hl/.style  = {en, very thick},
            lbl/.style = {font=\tiny, inner sep=0.5pt, fill=white},
          ]
          \node[v] (a) at (0,0)    {$a$};
          \node[v] (b) at (1.6,1)  {$b$};
          \node[v] (c) at (1.6,-1) {$c$};
          \node[v] (d) at (3.2,0)  {$d$};
          \draw[en] (a) to[bend left=18]  node[lbl, above left]  {$\tau_1$} (b);
          \draw[hl] (b) to[bend left=18]  node[lbl, below right] {$\tau_4$} (a);
          \draw[en] (a) to[bend left=18]  node[lbl, above right] {$\tau_2$} (c);
          \draw[en] (c) to[bend left=18]  node[lbl, below left]  {$\tau_6$} (a);
          \draw[hl] (a) -- node[lbl, above, pos=0.75] {$\tau_3$} (d);
          \draw[en] (b) -- node[lbl, above right] {$\tau_5$} (d);
          \draw[en] (c) -- node[lbl, below right] {$\tau_7$} (d);
          \node[font=\scriptsize] at (1.6,-1.75) {(a) \ Instance $D_{C_2}$.};

          \begin{scope}[shift={(4.4,0)}]
            \node[v] (a)  at (0,0)    {$a$};
            \node[v] (b)  at (1.6,1)  {$b$};
            \node[v] (c)  at (1.6,-1) {$c$};
            \node[v] (d)  at (3.2,0)  {$d$};
            \node[v] (e4) at (4.6,0)  {$e_4$};
            \node[v] (e5) at (5.8,0)  {$e_5$};
            \node[v] (e6) at (7.0,0)  {$e_6$};
            \draw[en] (a) to[bend left=18]  node[lbl, above left]  {$\tau_1$} (b);
            \draw[hl] (b) to[bend left=18]  node[lbl, below right] {$\tau_4$} (a);
            \draw[en] (a) to[bend left=18]  node[lbl, above right] {$\tau_2$} (c);
            \draw[en] (c) to[bend left=18]  node[lbl, below left]  {$\tau_6$} (a);
            \draw[hl] (a) -- node[lbl, above, pos=0.75] {$\tau_3$} (d);
            \draw[en] (b) -- node[lbl, above right] {$\tau_5$} (d);
            \draw[en] (c) -- node[lbl, below right] {$\tau_7$} (d);
            \draw[ex] (a) .. controls (-0.3,2.4) and (4.4,2.4) .. node[lbl] {$\tau_8$} (e4);
            \draw[ex] (b) to[bend left=15]  node[lbl, above=1pt] {$\tau_9$}    (e4);
            \draw[ex] (c) to[bend right=15] node[lbl, below=1pt] {$\tau_{10}$} (e4);
            \draw[ex] (d) -- node[lbl, above=1pt] {$\tau_{11}$} (e4);
            \draw[ex] (e4) -- node[lbl, above=1pt] {$\tau_{12}$} (e5);
            \draw[ex] (e5) -- node[lbl, above=1pt] {$\tau_{13}$} (e6);
            \node[draw, dashed, rounded corners=2pt, font=\scriptsize, inner sep=2pt]
              at (6.4,-1) {$\tau_{14}\!: A(c^\prime)$};
            \node[font=\scriptsize] at (3.5,-1.75) {(b) \ Instance $D^\prime$.};
          \end{scope}
        \end{tikzpicture}
        \caption{\ \ \ignore{Instances} $D_{C_2}, D^\prime$ as graphs.\ignore{, where each tuple $R(a_1,a_2)$ is an edge from $a_1$ to $a_2$.} Solid edges are endogenous\ignore{tuples}, dashed edges\ignore{ (and the box)} are exogenous\ignore{ tuples, and}, thick edges are $\tau_3$ and $\tau_4$.}
        \label{fig:cycle-c5-graph}
    \end{figure}

    If we see the database instance as a graph, with a tuple $R(a_1,a_2)$ seen as an edge from $a_1$ to $a_2$, a path of $\mc{Q}^{\nit{op}}_{C_2}$ in $D_{C_2}$ is a pair of consecutive edges, that is, the second edge starts where the first one ends. \
    Figure \ref{fig:cycle-c5-graph}(a) depicts those edges for $D_{C_2}$, the base counterexample. \
    Furthermore, Figure \ref{fig:cycle-c5-graph}(b) depicts the graph for $D^\prime$. In that case, we can see that the instance keeps the paths of $\mc{Q}^{\nit{op}}_{C_2}$ in $D_{C_2}$, and continues each of them with three exogenous edges. The tuple $A(c^\prime)$, created from the extra atom $A(x)$, does not belong to the path structure.\boxtheorem
\end{example}

The construction in the proof of Proposition \ref{prop:cr-not-aligned-cycles} is meant to be {\em sufficient}: \ For every query satisfying the conditions of the proposition, it guarantees that the built instance is a counterexample, but this instance may contain more (exogenous) tuples than needed. \
For example, in Figure \ref{fig:cycle-c5-graph}(b), the tuple (edge) $\tau_8$ can be removed without affecting the scores, and thus, nor their alignment.

\ignore{\green{Note that the result given in Proposition \ref{prop:cr_aligned_permutation} only includes queries where the self-join occurs for exactly one predicate.\
We conjecture that a similar result holds for queries that have self-joins for more predicates, for example, the query $\mc{Q} : \exist x \exists y (R(x,y)\land R(y,x) \land S(x,y) \land (y,x))$}.}

Finally, notice that, in order  to prove the non-alignment result in Proposition \ref{prop:cr-not-aligned-cycles}, we used the technique described in Remark \ref{rem:non-alignment-technique}, that was proposed for SJF-BCQs. \ This became possible due to the structures of the queries: Each counterexample has paths; and a path for any arbitrary query and instance can be ``completed'' by adding exogenous tuples, resulting in the ex-reduced lineage provided for $\mc{Q}_{C_2}^\nit{op}$ and $D_{C_2}$ in Example \ref{ex:cycle-open}.

It is unclear if this technique can be used for arbitrary BCQs with self-joins. \
Moreover, this technique may not be applicable to every class of queries: The structure of the queries in a class may not allow building counterexamples from those of its simplest queries; something we have frequently exploited in our work. Identifying the classes of queries for which the technique applies remains an open problem.


\section{Conclusions} \label{sec:concl}\vspace{2mm}

The comparison of attribution scores for tuples participating in  query answering in DBs, such {\sf Resp}, {\sf Shapley}, {\sf BPI}, {\sf CES}, has not received much attention.  We have provided a first analysis of their alignment, i.e. of whether  they induce  compatible rankings of tuples. \ We have identified classes of self-join-free BCQs for which {\sf Resp} and {\sf CES} are aligned, and others for which they are not. Our syntactic test is applicable to any SJF-BCQ, obtaining a dichotomy result. Exogenous tuples become crucial.

We have exhibited a broad class of queries for which {\sf CES} and  {\sf Resp} are  aligned on DBs without exogenous tuples. Identifying a dychotomic class for this case is still open.

In relation to the alignment of CES and Shapley, we have presented only negative results, that is, classes of queries for which they are not aligned. Obtaining positive results, that is, alignment for a -hopefully dychotomic- syntactic class is also an open problem.

We have just started an alignment analysis for queries with self-joins, but there is still much to do.

Questions about the alignment of  {\sf Shapley} and {\sf BPI}, and, indirectly about {\sf CES} in our case, are still open in  Game Theory \cite{freixas}, where commonly no distinction between exogenous and endogenous players is made.

In a different direction, not considered in our work, aligned scores may have very {\em different score distributions} on database tuples.  Obtaining  results  about the {\em comparison of score distributions} seems to be an interesting open problem.

And again in a different but important direction, it is not obvious how to address the practical question about which score to use for a particular application. Appealing to the characterizing properties of the scores seems to be unavoidable. Both {\sf Shapley} and {\sf BPI} have them. Those of {\sf BPI} are inherited by {\sf CES}. There is also an axiomatization for a {\em generalized} {\sf CES} score,  introduced and investigated in \cite{clear25}.  In \cite{izza}, there is a recent axiomatization of Responsibility for explainable machine learning.

Those general properties of scores may shed some light on the score distribution problem as well. \ However, the general properties of these scores by themselves will not solve the problem of applicability in databases. We need an agreement on (and also a characterization of) what is a good score-based explanation for query answering. We do not have them yet.

\vspace{3mm}
\noindent {\bf Acknowledgments:} \ Felipe Az\'ua has been partially supported by  the Millennium Institute for Foundational Research on Data (IMFD, Chile).  Part of this work was done while Felipe Az\'ua was visiting the ``Laboratory of Informatics,
Modelling and Optimization of the Systems" (LIMOS), at U. Clermont-Ferrand, France. He appreciates their support and hospitality. He is also grateful to Nina Pardal and the Huawei Grant for giving him time and financial support during his PhD studies. \
Leopoldo Bertossi has been financially supported by the  IMFD, and NSERC-DG 2023-04650.

	\bibliographystyle{plain}

\begin{thebibliography}{10}

\bibitem{deutchBanzhaf}
O.~Abramovich, D.~Deutch, N.~Frost, A.~Kara, and D.~Olteanu.
\newblock {Banzhaf Values for Facts in Query Answering}.
\newblock In {\em Proc. SIGMOD}, 2024.

\bibitem{deutch2}
D.~Arad, D.~Deutch, and N.~Frost.
\newblock {LearnShapley: Learning to Predict Rankings of Facts Contribution Based on Query Logs}.
\newblock In {\em Proc. CIKM}, pages 4788--4792, 2022.

\bibitem{clear25}
F.~Az{\'u}a and L.~Bertossi.
\newblock {The Causal-Effect Score in Data Management}.
\newblock In {\em Proc. 4th Conference on Causal Learning and Reasoning (CLeaR 2025)}, volume 275 of {\em PMLR}, pages 874--893, 2025.

\bibitem{foiks26}
F.~Az{\'u}a and L.~Bertossi.
\newblock {Causality-Based Scores Alignment in Explainable Data Management}.
\newblock In {\em Proc. FoIKS 2026}, volume 16475 of {\em LNCS}, pages 263--282. Springer, 2026.

\bibitem{banzhaf_value}
J.~Banzhaf~III.
\newblock {Weighted Voting Doesn't Work: A Mathematical Analysis}.
\newblock {\em Rutgers L. Rev.}, 19(2):317--343, 1965.

\bibitem{bda22}
L.~Bertossi.
\newblock {From Database Repairs to Causality in Databases and Beyond}.
\newblock {\em Transactions on Large-Scale Data- and Knowledge-Centered Systems LIV}, 14160:119--131, 2023.

\bibitem{sigRec23}
L.~Bertossi, B.~Kimelfeld, E.~Livshits, and M.~Monet.
\newblock {The Shapley Value in Database Management}.
\newblock {\em ACM SIGMOD Record}, 52(2):6--17, 2023.

\bibitem{DeclAI26}
L.~Bertossi and N.~Pardal.
\newblock {Sufficient Explanations in Databases and their Connections to Database Repairs}.
\newblock In {\em Proc. 10th International Joint Conference on Rules and Reasoning (RuleML+RR 2026)}, LNCS. Springer, 2026.
\newblock To appear. ArXiv paper 2511.15623.

\bibitem{flairs17}
L.~Bertossi and B.~Salimi.
\newblock {Causes for Query Answers from Databases: Datalog Abduction, View-Updates, and Integrity Constraints}.
\newblock {\em Int. J. Approx. Reason.}, 90:226--252, 2017.

\bibitem{tocs}
L.~Bertossi and B.~Salimi.
\newblock {From Causes for Database Queries to Repairs and Model-Based Diagnosis and Back}.
\newblock {\em Theory Comput. Syst.}, 61(1):191--232, 2017.

\bibitem{bienvenu}
M.~Bienvenu, D.~Figueira, and P.~Lafourcade.
\newblock {When is Shapley Value Computation a Matter of Counting?}
\newblock {\em Proc. ACM Manag. Data (PODS)}, 2:105, 2024.

\bibitem{izza}
G.~Biradar, Y.~Izza, E.~Lobo, V.~Viswanathan, and Y.~Zick.
\newblock {Axiomatic Aggregations of Abductive Explanations}.
\newblock In {\em Proc. AAAI}, 2024.

\bibitem{carreras-freixas-ordinal-equivalence}
F.~Carreras and J.~Freixas.
\newblock {On Ordinal Equivalence of Power Measures Given by Regular Semivalues}.
\newblock {\em Mathematical Social Sciences}, 55(2):221--234, 2008.

\bibitem{chockler}
H.~Chockler and J.~Halpern.
\newblock {Responsibility and Blame: {A} Structural-Model Approach}.
\newblock {\em J. Artif. Intell. Res.}, 22:93--115, 2004.

\bibitem{dichotomy_CQ_other}
N.~Dalvi and D.~Suciu.
\newblock {The Dichotomy of Conjunctive Queries on Probabilistic Structures}.
\newblock In {\em Proc. ACM PODS}, pages 293--302, 2007.

\bibitem{dichotomy_UCQ}
N.~Dalvi and D.~Suciu.
\newblock {The Dichotomy of Probabilistic Inference for Unions of Conjunctive Queries}.
\newblock {\em J. {ACM}}, 59(6):30:1--30:87, 2012.

\bibitem{deutch3}
S.~Davidson, D.~Deutch, N.~Frost, B.~Kimelfeld, O.~Koren, and M.~Monet.
\newblock {ShapGraph: An Holistic View of Explanations through Provenance Graphs and Shapley Values}.
\newblock In {\em Proc. SIGMOD}, 2022.

\bibitem{benny22}
D.~Deutch, N.~Frost, B.~Kimelfeld, and M.~Monet.
\newblock {Computing the Shapley Value of Facts in Query Answering}.
\newblock In {\em Proc. SIGMOD}, pages 1570--1583, 2022.

\bibitem{SwapRobust02}
L.~Diffo~Lambo and J.~Moulen.
\newblock {Ordinal Equivalence of Power Notions in Voting Games}.
\newblock {\em Theory and Decision}, 53(4):313--325, 2002.

\bibitem{BPI_Math_Props}
P.~Dubey and L.~S. Shapley.
\newblock {Mathematical Properties of the Banzhaf Power Index}.
\newblock {\em Math. Oper. Res.}, 4(2):99--131, 1979.

\bibitem{freixas}
J.~Freixas.
\newblock {On Ordinal Equivalence of the Shapley and Banzhaf Values for Cooperative Games}.
\newblock {\em Int. J. Game Theory}, 39:513--527, 2010.

\bibitem{freixas12}
J.~Freixas, D.~Marciniak, and M.~Pons.
\newblock {On the Ordinal Equivalence of the Johnston, Banzhaf and Shapley Power Indices}.
\newblock {\em European Journal of Operational Research}, 216:367--375, 2012.

\bibitem{gelman}
A.~Gelman and J.~Hill.
\newblock {\em {Data Analysis Using Regression and Multilevel/Hierarchical Models}}.
\newblock Cambridge Univ. Press, 2007.

\bibitem{halpern15}
J.~Halpern.
\newblock {A Modification of the {Halpern-Pearl} Definition of Causality}.
\newblock In {\em Proc. IJCAI}, pages 3022--3033, 2015.

\bibitem{H16}
J.~Halpern.
\newblock {\em {Actual Causality}}.
\newblock MIT Press, 2016.

\bibitem{HP05}
J.~Halpern and J.~Pearl.
\newblock {Causes and Explanations: A Structural-Model Approach. Part I: Causes}.
\newblock {\em The British Journal for the Philosophy of Science}, 56(4):843--887, 2005.

\bibitem{holland}
P.~W. Holland.
\newblock {Statistics and Causal Inference}.
\newblock {\em Journal of the American Statistical Association}, 81(396):945--960, 1986.

\bibitem{kara}
A.~Kara, D.~Olteanu, and D.~Suciu.
\newblock {From Shapley Value to Model Counting and Back}.
\newblock {\em Proc. ACM Manag. Data (PODS)}, 2:79, 2024.

\bibitem{senellart24}
P.~Karmakar, M.~Monet, P.~Senellart, and S.~Bressan.
\newblock {Expected Shapley-Like Scores of Boolean Functions: Complexity and Applications to Probabilistic Databases}.
\newblock {\em Proc. ACM Manag. Data (PODS)}, 2:92, 2024.

\bibitem{sigRec21}
E.~Livshits, L.~Bertossi, B.~Kimelfeld, and M.~Sebag.
\newblock {Query Games in Databases}.
\newblock {\em ACM SIGMOD Record}, 50(1):78--85, 2021.

\bibitem{Shapley_Tuple_Bertossi}
E.~Livshits, L.~Bertossi, B.~Kimelfeld, and M.~Sebag.
\newblock {The Shapley Value of Tuples in Query Answering}.
\newblock {\em Log. Methods Comput. Sci.}, 17(3), 2021.

\bibitem{wolfgang}
N.~Makhija and W.~Gatterbauer.
\newblock {A Unified Approach for Resilience and Causal Responsibility with Integer Linear Programming (ILP) and LP Relaxations}.
\newblock {\em Proc. ACM Manag. Data (SIGMOD)}, 1:228, 2023.

\bibitem{QA_CausalityII}
A.~Meliou, W.~Gatterbauer, J.~Halpern, {\relax Ch}.~Koch, K.~F. Moore, and D.~Suciu.
\newblock {Causality in Databases}.
\newblock {\em {IEEE} Data Eng. Bull.}, 33(3):59--67, 2010.

\bibitem{QA_Causality}
A.~Meliou, W.~Gatterbauer, K.~F. Moore, and D.~Suciu.
\newblock {The Complexity of Causality and Responsibility for Query Answers and Non-Answers}.
\newblock {\em Proc. {VLDB} Endow.}, 4(1):34--45, 2010.

\bibitem{Pearl_CI}
J.~Pearl.
\newblock {\em {Causality: Models, Reasoning, and Inference}}.
\newblock Cambridge University Press, USA, 2nd edition, 2009.

\bibitem{roth1988shapley}
A.~E. Roth, editor.
\newblock {\em {The Shapley Value: Essays in Honor of Lloyd S. Shapley}}.
\newblock Cambridge University Press, 1988.

\bibitem{roy&salimi23}
S.~Roy and B.~Salimi.
\newblock {Causal Inference in Data Analysis with Applications to Fairness and Explanations}.
\newblock In {\em Reasoning Web. Causality, Explanations and Declarative Knowledge}, volume 13759 of {\em LNCS}, pages 105--131. Springer, 2023.

\bibitem{rubin}
D.~B. Rubin.
\newblock {Estimating Causal Effects of Treatments in Randomized and Nonrandomized Studies}.
\newblock {\em Journal of Educational Psychology}, 66:688--701, 1974.

\bibitem{QA_Caus_Salimi}
B.~Salimi.
\newblock {\em {Query-Answer Causality in Databases and its Connections with Reverse Reasoning Tasks in Data and Knowledge Management}}.
\newblock PhD thesis, Carleton University, Canada, 2016.

\bibitem{Causal_Effect}
B.~Salimi, L.~Bertossi, D.~Suciu, and G.~{Van den Broeck}.
\newblock {Quantifying Causal Effects on Query Answering in Databases}.
\newblock In {\em Proc. TaPP}. {USENIX} Association, 2016.

\bibitem{shapley1953original}
L.~Shapley.
\newblock {A Value for n-Person Games}.
\newblock In {\em Contributions to the Theory of Games II}, pages 307--317. Princeton Univ. Press, 1953.

\bibitem{suciu}
D.~Suciu, D.~Olteanu, {\relax Ch}.~R{\'e}, and {\relax Ch}.~Koch.
\newblock {\em {Probabilistic Databases}}.
\newblock Synthesis Lectures on Data Management. Morgan \& Claypool Publishers, 2011.

\bibitem{Tomiyama87}
Y.~Tomiyama.
\newblock {Simple Game, Voting Representation and Ordinal Power Equivalence}.
\newblock {\em International Journal on Policy and Information}, 11(1):67--75, 1987.

\end{thebibliography}

\end{document}